\documentclass[a4paper, 11pt]{amsart}

\usepackage{amssymb, amsmath, amsthm, graphicx, booktabs, xcolor}
\usepackage{epstopdf}
\usepackage[caption = false]{subfig}
\usepackage{bm} 
\usepackage{mathrsfs} 
\usepackage{commath} 

\usepackage[margin=1in]{geometry}

\usepackage[ruled, vlined]{algorithm2e}

\SetCommentSty{mycommfont}

\usepackage[bookmarks = false]{hyperref}
\hypersetup{
    colorlinks=true,
    linkcolor=red,
		citecolor=blue,
    filecolor=magenta,      
    urlcolor=cyan,
}

\usepackage{natbib}
\bibpunct[, ]{(}{)}{;}{a}{}{,}

\allowdisplaybreaks[4]
\theoremstyle{plain}
\newtheorem{thm}{Theorem}[section]

\newtheorem{prop}[thm]{Proposition}

\theoremstyle{definition}

\newtheorem{asp}{Assumption}

\theoremstyle{remark}

\newcommand{\pder}[2][]{\frac{\partial#1}{\partial#2}}

\DeclareMathOperator*{\esssup}{ess\,sup}

\newcommand{\calA}{\mathcal{A}}

\newcommand{\calC}{\mathcal{C}}
\newcommand{\calF}{\mathcal{F}}
\newcommand{\calL}{\mathcal{L}}

\newcommand{\calS}{\mathcal{S}}

\newcommand{\scrV}{\mathscr{V}}

\newcommand{\E}{\mathbb{E}}				
\newcommand{\prob}{\mathbb{P}}		
\newcommand{\Q}{\mathbb{Q}}

\newcommand{\vm}[1]{\mathbf{#1}}

\begin{document}


\title{A Numerical Approach to Pricing Exchange Options under Stochastic Volatility and Jump-Diffusion Dynamics}


\title[Exchange Options under SVJD Dynamics]{A Numerical Approach to Pricing Exchange Options under Stochastic Volatility and Jump-Diffusion Dynamics}

\author[LPDM Garces]{Len Patrick Dominic M. Garces$\dag$}
\author[GHL Cheang]{Gerald. H. L. Cheang$\ddag$}

\address[$\dag$, $\ddag$]{University of South Australia, UniSA STEM, Centre for Industrial and Applied Mathematics, Adelaide SA 5000, Australia}
\address[$\dag$]{Ateneo de Manila University, School of Science and Engineering, Department of Mathematics, Quezon City 1108, Metro Manila, Philippines}

\email[$\dag$, Corresponding author]{len\_patrick\_dominic.garces@mymail.unisa.edu.au}
\email[$\ddag$]{gerald.cheang@unisa.edu.au}


\thanks{This is the preprint of the article of the same title published by Taylor \& Francis in \emph{Quantitative Finance}. The final version is available online at \url{https://doi.org/10.1080/14697688.2021.1926534}.}

\begin{abstract}
We consider a method of lines (MOL) approach to determine prices of European and American exchange options when underlying asset prices are modelled with stochastic volatility and jump-diffusion dynamics. As the MOL, as with any other numerical scheme for PDEs, becomes increasingly complex when higher dimensions are involved, we first simplify the problem by transforming the exchange option into a call option written on the ratio of the yield processes of the two assets. This is achieved by taking the second asset yield process as the num\'eraire. We also characterize the near-maturity behavior of the early exercise boundary of the American exchange option and analyze how model parameters affect this behavior. Using the MOL scheme, we conduct a numerical comparative static analysis of exchange option prices with respect to the model parameters and investigate the impact of stochastic volatility and jumps to option prices. We also consider the effect of boundary conditions at far-but-finite limits of the computational domain on the overall efficiency of the MOL scheme. Toward these objectives, a brief exposition of the MOL and how it can be implemented on computing software are provided.
\end{abstract}


\maketitle


\section{Introduction}

We investigate the pricing of European and American exchange options written on assets with prices driven by stochastic volatility and jump-diffusion (SVJD) dynamics. The earliest analysis of European exchange options was that of \citet{Margrabe-1978} who, by noting that the European exchange option price is linear homogeneous in the stock prices, transformed the problem to the classical European call option pricing problem which was then solved using the method of \citet{BlackScholes-1973}. \citet{Bjerskund-1993} considered a similar approach in pricing American exchange options in a pure diffusion setting. They suggested that by choosing one of the stocks as the num\'eraire, the American exchange option pricing problem may be simplified to the problem of pricing an American call or put option. \citet{Bjerskund-1993} refer to this technique as the \emph{put-call transformation}.

With well-established evidence pointing to the deficiencies of the geometric Brownian motion in accurately modelling asset price returns, there has since been a movement to study option prices (including exchange options) under alternative asset price models.\footnote{The empirical literature addressing the limitations of the \citet{BlackScholes-1973} is extremely rich and will not be reviewed here. Instead, we invite the reader to consult \citet{Bakshi-1997}, \citet{Duffie-2000}, \citet{Cont-2001}, \citet{Andersen-2002}, \citet{Chernov-2003}, \citet{Eraker-2003}, \citet{Kou-2008}, and the references therein.} \citet{CheangChiarellaZiogas-2006}, \citet{CheangChiarella-2011}, \citet{Caldana-2015}, \citet{PetroniSabino-2018}, and \citet{Ma-2020} analyzed European exchange options when asset prices are modelled using jump-diffusion processes. \citet{Antonelli-2010}, \citet{Alos-2017}, and \citet{KimPark-2017} priced European exchange options where underlying assets are driven by stochastic volatility models. \citet{CheangChiarella-2011} also considered the case of American exchange options in their analysis. More recently, \citet{CheangGarces-2019}, derived analytical representations for the European and American exchange option prices assuming that stock prices are modelled using a pair of \citeauthor{Bates-1996} stochastic volatility and jump-diffusion dynamics. Among the aforementioned investigations, those of \citet{CheangChiarellaZiogas-2006} and \citet{Alos-2017} priced exchange options using the put-call transformation approach. \citet{Fajardo-2006} also used a similar transformation, which they called the ``dual market method'', to price options (including perpetual exchange options) when underlying prices are driven by L\'evy processes. The others priced exchange options under the equivalent martingale measure (EMM) corresponding to the money market account.

The addition of both stochastic volatility and jump-diffusion precludes the availability of closed-form solutions (in the sense of the Black-Scholes or the Margrabe formulas) for exchange option prices, hence we resort to a numerical approximation of exchange option prices. Our main contributions toward this objective are as follows:
\begin{enumerate}
	\item We extend the \citet{Bjerskund-1993} strategy for valuing American exchange options in a pure-diffusion setting into the SVJD framework. This analysis also extends the closely related dual market approach of \citet{Fajardo-2006} for L\'evy processes to accommodate stochastic volatility. To the best of our knowledge, not much focus has been placed on the use of the put-call transformation technique to pricing finite maturity American options in the stochastic volatility \emph{and} jump-diffusion settings. As such, we aim to discuss how this technique can be applied to pricing exchange options with such asset price dynamics. Although we focus on exchange options only, as \citet{Fajardo-2006} notes, the approach is just as useful when pricing derivatives with homogeneous payoff functions.
	\item In the simplified framework, we investigate the behavior of the early exercise boundary of the American exchange option near maturity, thereby extending the analysis of \citet{ChiarellaZiogas-2009} for American calls under jump-diffusion. We also study analytically and numerically how some key model parameters, namely the dividend yields and the jump parameters, affect the behavior of the early exercise boundary.
	\item We give a detailed discussion of a method of lines (MOL) scheme to numerically determine the price of exchange options, expressed in units of the second asset yield process, and the joint transition density function of the underlying state variables in the SVJD framework.
	\item We present an accessible and more detailed discussion of the MOL under more complex underlying asset price dynamics, since a detailed structure of the implementation, as is shown in Algorithms \ref{pseudo-MOL-EuExcOp} and \ref{pseudo-MOL-AmExcOp} in Section \ref{sec-MOL} of this paper, is usually excluded in papers that use the MOL. 
	\item As alternatives to the often assumed boundary condition $\lim_{v\to\infty}\frac{\partial V}{\partial v} = 0$ for the behavior of option prices at high volatility levels \citep[e.g.][]{Clarke-1999, Chiarella-2009}, we consider Venttsel boundary conditions for the far-but-finite limits of the computational domain and their impact on option prices and the performance of the MOL algorithm.
	\item Using MOL-generated prices, we investigate how stochastic volatility and jumps in the asset prices affect option prices. Furthermore, we conduct an extensive numerical comparative static analysis to see how key model parameters affect the exchange option prices and the early exercise boundary.
\end{enumerate}
As such, this paper serves as a numerical complement to the work of \citet{CheangGarces-2019} who focused more on the analytical aspects of pricing exchange options under SVJD dynamics.

In this analysis, pricing takes place under the EMM $\hat{\Q}$ corresponding to setting the second asset yield process as the num\'eraire. Under $\hat{\Q}$, we find that the the no-arbitrage price of the European exchange option can be written as a function of only the asset yield ratio $\tilde{s}$ and the instantaneous variance $v$. Furthermore, we verify an early exercise representation of the discounted American exchange option price, which can also be written as a function of only $\tilde{s}$ and $v$. The pricing integro-partial differential equations for European and American exchange options, as well as the associated boundary conditions, under this measure are then derived.

The pricing IPDEs are then solved using the method of lines. The method of lines is a numerical method to solve PDEs which consists of discretizing the equation in all but one variable resulting to a sequence or system of ODEs in the remaining continuous variable. \citet{Schiesser-2009} provide a general exposition on the method, but for applications in option and fixed income instrument pricing, the time-discrete MOL, expertly discussed by \citet{Meyer-2015}, has gained particular traction.\footnote{The time-discrete MOL involves discretizing the PDE in all but one \emph{spatial variable}, as opposed to most applications where time is left as the continuous variable \citep[see][]{Schiesser-2009}. This approach is also known as Rothe's method or the horizontal MOL.} This type of MOL has been applied to pricing American put options in the Black-Scholes framework \citep{vanderHoek-1997}, put options under jump-diffusion dynamics \citep{Meyer-1998}, call options under stochastic volatility \citep{Adolfsson-2013, ChiarellaZiveyi-2013}, call options under \citeauthor{Bates-1996} SVJD \citep{Chiarella-2009}, American options with SV and stochastic interest rates \citep{Kang-2014}, American options under a regime-switching GBM \citep{Chiarella-2016}, and spread options under pure-diffusion dynamics \citep{ChiarellaZiveyi-2014}. It is particularly useful for approximating American option prices as the algorithm can be easily adjusted to accommodate unknown free boundaries. It is also attractive for financial applications as the option delta and gamma are calculated as part of the algorithm with no additional computational cost. However as with any numerical technique for solving PDEs, the MOL becomes highly complex the more spatial variables are involved. In this paper, we thus use the put-call transformation technique in an effort to simplify the MOL approximation of the exchange option price.

While the succeeding analysis focuses on exchange options written on stocks, one may consider exchange options written on other assets such as indices and foreign currencies. For foreign currencies, in particular, the dividend yields are replaced by risk-free interest rates in the domestic and foreign money markets. \citet{Siegel-1995} explains how exchange options can be used to estimate the ``implicit beta'' between an underlying stock and a given market index. The exchange option framework may be adapted to investigate real options \citep{Kensinger-1988, Carr-1995}, outperformance options \citep{CheangChiarella-2011}\footnote{\citet{CheangChiarella-2011} assumed that only one asset price process had jumps while the other was modelled as a pure-diffusion process. \citet{QuittardPinon-2010} discuss in greater detail the European exchange option pricing problem under a similar model specification.}, energy market options \citep[surveyed in][]{Benth-2015}, and the option to enter/exit an emerging market \citep{Miller-2012}, among others. \citet{Ma-2020} provide additional examples of financial contracts which can be priced under the exchange option framework.

The rest of the paper is organized as follows. Section \ref{sec-PutCall-SVJDModel} describes the SVJD model for the asset prices and the stochastic variance and discusses construction of the measure $\hat{\Q}$. In Section \ref{sec-PutCall-ExchangeOptionIPDE}, we derive the pricing IPDE for the European and American exchange options. Section \ref{sec-PutCall-EEBLimit} discusses the behavior of the early exercise boundary near the expiry of the American exchange option. Section \ref{sec-MOL} explains the MOL algorithm for the solution of the pricing IPDE, the results of which are shown in Section \ref{sec-NumericalResults}. Section \ref{sec-Conclusion} concludes the paper. The focus of this paper is on the numerical implementation, hence we only briefly describe the proof of some technical results shown here and instead refer to \citep{GarcesCheang-2020} as it focuses on the probabilistic and analytical representation of exchange option prices under this framework.

\section{Asset Price Dynamics and the Put-Call Transformation}
\label{sec-PutCall-SVJDModel}

In this section, we discuss the model specification for the underlying stock prices. We assume that the financial market consists of a risk-free money market account and two risky assets over a finite time period $[0,T]$. We also let $T$ be the maturity of the exchange option. The dynamics of the asset prices are discussed below.

Let $(\Omega,\calF,\prob)$ be a probability space equipped with a filtration $\{\calF_t\}_{0\leq t\leq T}$ satisfying the usual conditions. Let $\{W_1(t)\}$, $\{W_2(t)\}$, and $\{Z(t)\}$ be standard $\prob$-Brownian motions with instantaneous correlations given by $\dif W_1(t)\dif W_2(t) = \rho_w \dif t$ and $\dif W_j(t)\dif Z(t) = \rho_j\dif t$, for $j=1,2.$ Denote by $\bm{\Sigma}$ the correlation matrix of the random vector $\vm{B}(t) = (W_1(t),W_2(t),Z(t))^\top$. Let $p(\dif y_j,\dif t)$ ($j=1,2$) be the counting measure associated to a marked Poisson process with $\prob$-local characteristics $(\lambda_j,m_\prob(\dif y_j))$.\footnote{See \citet{Runggaldier-2003} for more details.} Underlying $p(\dif y_j,\dif t)$ is a sequence of ordered pairs $\{(T_{i,n},Y_{i,n})\}$ where $Y_{i,n}$ is the ``mark'' of the $n$th occurrence of an event that occurs at a non-explosive time $T_{i,n}$. The marks $Y_{j,1},Y_{j,2},\dots$ are i.i.d. real-valued random variables with non-atomic density $m_\prob(\dif y_j)$. Associated to the event times, we define a Poisson counting process $\{N_j(t)\}$ given by $N_j(t) = \sum_{n=1}^\infty \vm{1}(T_{j,n}\leq t)\vm{1}(Y_{j,n}\in\mathbb{R}),$ where $\vm{1}(\cdot)$ is the indicator function.

We assume that the counting measures are independent of the Brownian motions and of each other. Henceforth, we assume that $\{\calF_t\}$ is the natural filtration generated by the Brownian motions and the counting measures, augmented with the collection of $\prob$-null sets.

Denote by $\{S_1(t)\}$ and $\{S_2(t)\}$ the price processes of two assets that pay a constant dividend yield of $q_1$ and $q_2$, respectively, per annum. As stock prices may jump, we let $S_1(t)$ and $S_2(t)$ denote the stock prices \emph{prior to any jumps occurring at time $t$}. Let $\{v(t)\}$ be the instantaneous variance process that governs the volatility of both stock price processes. We assume that the dynamics of the stock prices and the instantaneous variance satisfy the stochastic differential equations
\begin{align}
\label{eqn-PutCall-StockPriceSDE-P}
\frac{\dif S_j(t)}{S_j(t)}	& = (\mu_j-\lambda_j\kappa_j)\dif t+\sigma_j \sqrt{v(t)}\dif W_j(t)+\int_{\mathbb{R}}\left(e^{y_j}-1\right)p(\dif y_j,\dif t), \qquad j=1,2,\\
\label{eqn-PutCall-VarianceSDE-P}
\dif v(t) & = \xi\left(\eta-v(t)\right)\dif t + \omega\sqrt{v(t)}\dif Z(t).
\end{align}
Here, $\kappa_j \equiv \E_\prob[e^{Y_j}-1] = \int_{\mathbb{R}}(e^{y_j}-1)m_\prob(\dif y_j)$ is the mean jump size of the price of asset $j$ under $\prob$, and $\mu_j$, $\sigma_j$, $\xi$, $\eta$, and $\omega$ are positive constants. It is also assumed that $S_1(0), S_2(0), v(0)>0$. 

We refer to this model as the \emph{proportional stochastic volatility and jump-diffusion (SVJD) model}. While there is only one variance processes feeding into the diffusion component of each of the asset prices, the degree of influence the stochastic volatility process has on the asset price dynamics is governed by the proportionality coefficients $\sigma_1$ and $\sigma_2$.\footnote{In contrast, \citet{CheangGarces-2019} assume one variance process for each asset price. However, their analytical representations require that the asset price processes are uncorrelated with each other and with the variance processes. The current model specification allows such dependence structure.} In turn, the dynamics of the stochastic volatility process, modelled by a CIR square-root process, is dictated by the speed of mean reversion $\xi$, the long-run variance $\eta$, and the volatility of volatility $\omega$.

As described above, the model features a common instantaneous variance process and independent jump terms for each asset. The individual jump processes may be taken to model idiosyncratic risk factors in each asset that cause sudden changes in returns.\footnote{In contrast, \citet{CheangChiarella-2011} introduced an additional compound Poisson process appearing in both asset return processes which capture macroeconomic shocks or systematic risk factors which may introduce sudden jumps in returns.} Although extremely rare, it is possible that jumps for both stocks arrive at the same time, representing market shocks or sudden events that may affect both assets. In addition, the common variance process models systematic market volatility or volatility at the macroeconomic level. As such, individual asset prices may provide feedback to each other via the correlation between the diffusion components and the dependence on a common stochastic volatility.

The price process of the money market account is denoted by $\{M(t)\}$, with $M(t)=e^{rt}$ for $t\geq 0$, where $r>0$ is the (constant) risk-free interest rate. 

In this paper, we assume that the dividend yields, the risk-free rate, the parameters $\mu_j$, $\sigma_j$, $\xi$, $\eta$, and $\omega$, the jump the intensities $\lambda_j$ and jump-size densities $m_\prob(\dif y_j)$ are constant through time, but the analysis can be extended to the case where these parameters are deterministic functions of time.

We require the following assumption on the parameters of the variance process and the correlation parameters to ensure that $\{v(t)\}$ remains strictly positive and finite for all $0\leq t\leq T$ under $\prob$ and any other probability measure equivalent to $\prob$ \citep{AndersenPiterbarg-2007, CheangGarces-2019}.

\begin{asp}
\label{asp-PutCall-VarianceParameterAssumptions}
The parameters $\xi$, $\eta$, and $\omega$ and the correlation coefficients $\rho_1$ and $\rho_2$ satisfy $2\xi\eta\geq\omega^2$ (Feller condition) and $-1<\rho_j<\min\left\{\xi/\omega,1\right\}$, $j=1,2$.
\end{asp}

Straightforward calculations using It\^o's Lemma for jump-diffusions show that equation \eqref{eqn-PutCall-StockPriceSDE-P} admits a solution of the form 
\begin{align*}
S_j(t) & = S_j(0)\exp\Bigg\{(\mu_j-\lambda_j\kappa_j)t-\frac{1}{2}\sigma_j^2\int_0^t v(s)\dif s + \sigma_j\int_0^t \sqrt{v(s)}\dif W_j(s)+\sum_{n=1}^{N_j(t)}Y_{j,n}\Bigg\},
\end{align*}
for $0<t\leq T$, $j=1,2$. Assumption \ref {asp-PutCall-VarianceParameterAssumptions} and the non-explosion assumption on the point processes imply that the integrals and summation that appear above are well-defined. It also follows that $S_j(t)>0$ $\prob$-a.s. for all $t\in[0,T]$, and hence either asset can be used as a num\'eraire.

Instead of $\{M(t)\}$, we take $\{S_2(t)e^{q_2 t}\}$, the second asset yield process, as the num\'eraire and define the probability measure $\hat{\Q}\sim\prob$, such that the first asset yield process and the money market account, when discounted by $S_2(t)e^{q_2 t}$, are martingales under $\hat{\Q}$. With the second asset yield process as the num\'eraire, the \emph{discounted price} of any other asset with price process $\{X(t)\}$ is defined by $\tilde{X}(t) = X(t)(S_2(t)e^{q_2 t})^{-1}$. 

Next, we discuss the construction of $\hat{\Q}$. The following standard proposition specifies the form of the Radon-Nikod\'ym derivative $\frac{\dif\hat{\Q}}{\dif\prob}$.

\begin{prop}
\label{prop-PutCall-ChangeofMeasure}
Suppose $\bm{\theta}(t) = \left(\psi_{1}(t),\psi_{2}(t),\zeta(t)\right)^\top$ is a vector of $\calF_t$-adapted processes and let $\gamma_1,\gamma_2,\nu_1,\nu_2$ be constants. Define the process $\{L_t\}$ by 
\begin{align}
\begin{split}
\label{eqn-PutCall-RNDerivative}
L(t)	& = \exp\left\{-\int_0^t\left(\bm{\Sigma}^{-1}\bm{\theta}(s)\right)^\top\dif\vm{B}(s)-\frac{1}{2}\int_0^t\bm{\theta}(s)^\top\bm{\Sigma}^{-1}\bm{\theta}(s)\dif s\right\}\\
		& \qquad \times\exp\left\{\sum_{n=1}^{N_{1}(t)}(\gamma_1 Y_{1,n}+\nu_1)-\lambda_1 t\left(e^{\nu_1}\E_\prob(e^{\gamma_1 Y_{1}})-1\right)\right\}\\
		& \qquad \times\exp\left\{\sum_{n=1}^{N_{2}(t)}(\gamma_2 Y_{2,n}+\nu_2)-\lambda_2 t\left(e^{\nu_2}\E_\prob(e^{\gamma_2 Y_{2}})-1\right)\right\}
\end{split}
\end{align}
and suppose that $\{L(t)\}$ is a strictly positive $\prob$-martingale such that $\E_\prob[L(t)]=1$ for all $t\in[0,T]$. Then $L(T)$ is the Radon-Nikod\'ym derivative of some probability measure $\hat{\Q}\sim\prob$ and the following hold:\footnote{The Radon-Nikod\'ym derivative $L(T) = \frac{\dif\hat{\Q}}{\dif\prob}$ can be used to characterize any probability measure $\hat{\Q}$ equivalent to $\prob$ as parameterized by the vector process $\{\bm{\theta}(t)\}$ and the constants $\gamma_1,\gamma_2,\nu_1,\nu_2$. We assume that $\gamma_1,\gamma_2,\nu_1,\nu_2$ are constant to preserve the time-homogeneity of the intensity and the jump size distribution. As the market under the SVJD is generally incomplete, one can construct multiple equivalent martingale measures consistent with the no-arbitrage assumption.}
\begin{enumerate}
	\item Under $\hat{\Q}$, the vector process $\vm{B}(t)$ has drift $-\bm{\theta}(t)$;
	\item $N_j(t)$ has intensity $\tilde{\lambda}_j = \lambda_j e^{\nu_j}\E_\prob[e^{\gamma_j Y_j}]$, $j=1,2$	under $\hat{\Q}$; and
	\item The mgf of $Y_j$ under $\hat{\Q}$ is given by $M_{\hat{\Q},Y_j}(u) = M_{\prob,Y_j}(u+\gamma_j)/M_{\prob,Y_j}(\gamma_j)$, $j=1,2.$
\end{enumerate}
\end{prop}

\begin{proof}
See e.g. \citet[Theorem 2.4]{Runggaldier-2003} and \citet[Theorem 1]{CheangTeh-2014}.
\end{proof}

We now specify the parameters of $L(T)$ so that $\hat{\Q}$ is an EMM corresponding to the num\'eraire $\{S_2(t)e^{q_2 t}\}$. Let $\{\tilde{S}(t)\}$ and $\{\tilde{M}(t)\}$, where $\tilde{S}(t) = S_1(t)e^{q_1 t}/(S_2(t)e^{q_2 t})$ and $\tilde{M}(t) = e^{rt}/(S_2(t)e^{q_2 t}),$ be the first asset yield process and the money market account when discounted using the second stock's yield process. We will refer to $\{\tilde{S}(t)\}$ as the \emph{asset yield ratio process}. If we choose $\{\psi_1(t)\}$, $\{\psi_2(t)\}$, and $\{\zeta(t)\}$ as
\begin{align}
\label{eqn-PutCall-RiskPremium-1}
\psi_1(t) & = \frac{\mu_1+q_1-r-\rho_w\sigma_1\sigma_1 v(t)-\lambda_1\kappa_1+\tilde{\lambda}_1\tilde{\kappa}_1}{\sigma_1\sqrt{v(t)}}\\
\label{eqn-PutCall-RiskPremium-2}
\psi_2(t) & = \frac{\mu_2+q_2-r-\sigma_2^2 v(t)-\lambda_2\kappa_2-\tilde{\lambda}_2\tilde{\kappa}_2^-}{\sigma_2\sqrt{v(t)}}\\
\label{eqn-PutCall-MarketPriceofVolRisk}
\zeta(t) & = \frac{\Lambda}{\omega}\sqrt{v(t)} \qquad \text{for some constant $\Lambda\geq 0$},
\end{align} 
where $\tilde{\kappa}_1 = \E_{\hat{\Q}}[e^{Y_1}-1]$ and $\tilde{\kappa}_2^- = \E_{\hat{\Q}}[e^{-Y_2}-1]$, then $\{\tilde{S}(t)\}$ and $\{\tilde{M}(t)\}$ are $\hat{\Q}$-martingales on $[0,T]$.\footnote{This assertion can be proved using It\^o's Lemma on $\tilde{S}(t)$ and $\tilde{M}(t)$ and eliminating the resulting drift term as required by the martingale representation for jump-diffusion processes \citep[see][Theorem 2.3]{Runggaldier-2003}.}

With this choice of parameters for $\hat{\Q}$, the dynamics of $\{v(t)\}$ becomes
\begin{equation}
\label{eqn-PutCall-VarianceSDE-Q}
\dif v(t) = \left[\xi\eta-(\xi+\Lambda)v(t)\right]\dif t+\omega\sqrt{v(t)}\dif\bar{Z}(t).
\end{equation}
where $\{\bar{Z}(t)\}$ is a $\hat{\Q}$-Wiener process. The choice of $\zeta(t)$ preserves the structure of the instantaneous variance as a square-root process. Assumption \ref{asp-PutCall-VarianceParameterAssumptions} ensures that this process is strictly positive and finite $\hat{\Q}$-a.s.

Under $\hat{\Q}$, $\tilde{S}(t)$ satisfies the equation
\begin{align}
\begin{split}
\label{eqn-PutCall-YieldRatioSDE-Q}
\dif\tilde{S}(t)
	& = -\tilde{S}(t)\left(\tilde{\lambda}_1\tilde{\kappa}_1+\tilde{\lambda}_2\tilde{\kappa}_2^-\right)\dif t + \sigma\sqrt{v(t)}\tilde{S}(t)\dif\bar{W}(t)\\
	& \qquad + \int_{\mathbb{R}} \left(e^{y_1}-1\right)\tilde{S}(t) p(\dif y_1,\dif t) + \int_{\mathbb{R}}\left(e^{-y_2}-1\right)\tilde{S}(t)p(\dif y_2,\dif t).
\end{split}
\end{align}
where we define $\sigma\dif\bar{W}(t) \equiv \sigma_1\dif\bar{W}_1(t)-\sigma_2\dif\bar{W}_2(t)$ with standard $\hat{\Q}$-Wiener processes $\{\bar{W}_1(t)\}$ and $\{\bar{W}_2(t)\}$ and $\sigma^2 = \sigma_1^2+\sigma_2^2-2\rho_w\sigma_1\sigma_2$.\footnote{In view of Proposition \ref{prop-PutCall-ChangeofMeasure}, we note that $\dif\bar{W}_j(t) = \psi_j(t)\dif t+\dif W_j(t)$.} 
Lastly, we note that the instantaneous correlation between the $\hat{\Q}$-Brownian motions $\{\bar{W}(t)\}$ and $\{\bar{Z}(t)\}$ is given by $\E_{\hat{\Q}}\left[\dif\bar{W}(t)\dif\bar{Z}(t)\right] = (1/\sigma)(\sigma_1 \rho_1-\sigma_2\rho_2)\dif t$.

\section{The Exchange Option Pricing IPDE}
\sectionmark{Exchange Option IPDE}
\label{sec-PutCall-ExchangeOptionIPDE}

Now we derive the integro-partial differential equation (IPDE) for the price of an exchange option written on $S_1$ and $S_2$. Denote by $C^E(t,S_1(t),S_2(t),v(t))$ the price of a European exchange option whose terminal payoff is given by $C^E\left(T,S_1(T),S_2(T),v(T)\right) = \left(S_1(T)-S_2(T)\right)^+,$ where $x^+\equiv\max\{x,0\}$. A rearrangement of terms expresses the discounted terminal payoff as
\begin{equation*}
\frac{C^E\left(T,S_1(T),S_2(T),v(T)\right)}{S_2(T)e^{q_2 T}} = e^{-q_1 T}\left(\tilde{S}(T)-e^{(q_1-q_2)T}\right)^+.
\end{equation*}
Let $\tilde{C}^E(t,S_1(t),S_2(t),v(t)) \equiv C^E\left(t,S_1(t),S_2(t),v(t)\right)/\left(S_2(t)e^{q_2 t}\right)$ denote the discounted European exchange option price. Then, assuming that no arbitrage opportunities exist, $\tilde{C}^E(t,S_1(t),S_2(t),v(t))$ is given by
\begin{align}
\begin{split}
\label{eqn-PutCall-DiscountedExcOpPrice}
\tilde{C}^E\left(t,S_1(t),S_2(t),v(t)\right)
	& = \E_{\hat{\Q}}\left[\left.\tilde{C}\left(T,S_1(T),S_2(T),v(T)\right)\right|\calF_t\right]\\
	& = e^{-q_1 T}\E_{\hat{\Q}}\left[\left.\left(\tilde{S}(T)-e^{(q_1-q_2)T}\right)^+\right|\calF_t\right].
\end{split}
\end{align}
In other words, the price at any time $t<T$ of the European exchange option measured in units of the second asset yield process is the $\hat{\Q}$-expectation of the terminal payoff measured in units of the second asset yield process \citep{Geman-1995}. From the last equation, we also note that the terminal payoff is variable only in the asset yield ratio $\tilde{S}(t)$. Thus, we assume that the discounted European exchange option price is represented by the process $\tilde{V}^E(t,\tilde{S}(t),v(t))$ and so
\begin{equation}
\label{eqn-PutCall-DiscountedExcOpPrice2}
\tilde{V}^E(t,\tilde{S}(t),v(t)) = e^{-q_1 T}\E_{\hat{\Q}}\left[\left.\left(\tilde{S}(T)-e^{(q_1-q_2)T}\right)^+\right|\calF_t\right].
\end{equation}

Thus we have shown that, by taking the second asset yield process as the num\'eraire asset, \emph{the exchange option pricing problem is equivalent to pricing a European call option on the asset yield price ratio $\tilde{S}(t)$ with maturity date $T$ and strike price $e^{(q_1-q_2)T}$}. 
If we choose $\{S_1(t)e^{q_1 t}\}$ as the num\'eraire, then the problem simplifies to the valuation of a \emph{put option} written on the asset yield ratio $(S_{2}(t)e^{q_2 t})/(S_1(t)e^{q_1 t})$.

The following technical assumption is required to implement It\^o's formula for jump-diffusion processes.

\begin{asp}
\label{asp-PutCall-Differentiability}
For $t\in[0,T]$, $\tilde{V}^E(t,\tilde{s},v)$ is (at least) twice-differentiable in $\tilde{s}$ and $v$ and differentiable in $t$ with continuous partial derivatives.
\end{asp}

The following proposition provides the IPDE that characterizes the discounted European exchange option price. This result can be proved using It\^o's formula for jump-diffusion processes and employing usual martingale arguments \citep[see][Appendix 2]{CheangGarces-2019}.

\begin{prop}
\label{prop-PutCall-DiscountedEuExcOpPrice}
The price at time $t\in[0,T)$ of the European exchange option is given by
\begin{equation}
C^E(t,S_1(t),S_2(t),v(t)) = S_2(t)e^{q_2 t}\tilde{V}^E(t,\tilde{S}(t),v(t)),
\end{equation}
where $\tilde{V}^E$, satisfying Assumption \ref{asp-PutCall-Differentiability}, is the solution of the terminal value problem
\begin{align}
\label{eqn-PutCall-IPDE-Vtilde}
0 & = \pder[\tilde{V}^E]{t}+\calL_{\tilde{s},v}\left[\tilde{V}^E(t,\tilde{S}(t),v(t))\right], \qquad (t,\tilde{S}(t),v(t))\in[0,T]\times\mathbb{R}_+^2\\
\label{eqn-PutCall-IPDE-Vtilde-TerminalCondition}
\tilde{V}^E(T) & = e^{-q_1 T}\left(\tilde{S}(T)-e^{(q_1-q_2)T}\right)^+,
\end{align}
with $\mathbb{R}_+^2 = (0,\infty)\times(0,\infty)$ and the IPDE operator $\calL_{\tilde{s},v}$ defined as 
\begin{align}
\begin{split}
\label{eqn-PutCall-IPDEOperator}
\calL_{\tilde{s},v}\left[\tilde{V}^E(t,\tilde{S},v)\right] 
	& = -\tilde{S}\left(\tilde{\lambda}_1\tilde{\kappa}_1+\tilde{\lambda}_2\tilde{\kappa}_2^-\right)\pder[\tilde{V}^E]{\tilde{s}}+\left[\xi\eta-(\xi+\Lambda)v\right]\pder[\tilde{V}^E]{v}\\
	& \qquad + \frac{1}{2}\sigma^2 v\tilde{S}^2\pder[^2\tilde{V}^E]{\tilde{s}^2}+\frac{1}{2}\omega^2 v\pder[^2\tilde{V}^E]{v^2}+\omega(\sigma_1 \rho_1-\sigma_2\rho_2)v\tilde{S}\pder[^2\tilde{V}^E]{\tilde{s}\partial v}\\
	& \qquad + \tilde{\lambda_1}\E_{\hat{\Q}}^{Y_1}\left[\tilde{V}^E\left(t,\tilde{S}e^{Y_1},v\right)-\tilde{V}^E(t,\tilde{S},v)\right]\\
	& \qquad + \tilde{\lambda}_2\E_{\hat{\Q}}^{Y_2}\left[\tilde{V}^E\left(t,\tilde{S}e^{-Y_2},v\right)-\tilde{V}^E(t,\tilde{S},v)\right],
\end{split}
\end{align}
where $\E_{\hat{\Q}}^{Y_i}$ is the expectation with respect to the r.v. $Y_i$ ($i=1,2$) under the measure $\hat{\Q}$. Note that all partial derivatives are evaluated at $(t,\tilde{S}(t),v(t))$.
\end{prop}

Let $C^A(t,S_1(t),S_2(t),v(t))$ be the price at time $t$ of an American exchange option written on $S_1$ and $S_2$. After a rearrangement of terms, standard theory on American option pricing \citep[see e.g.][]{Myneni-1992} dictates that the discounted American exchange option price $\tilde{V}^A(t,\tilde{S}(t),v(t))$ is given by
\begin{align}
\begin{split}
\tilde{V}^A(t,\tilde{S}(t),v(t)) 
	& \equiv \frac{C^A(t,S_1(t),S_2(t),v(t)}{S_2(t)e^{q_2 t}}\\
	& = \esssup_{u\in[t,T]} e^{-q_1 u}\E_{\hat{\Q}}\left[\left.\left(\tilde{S}(u)-e^{(q_1-q_2)u}\right)^+\right|\calF_t\right],
\end{split}
\end{align}
where the supremum is taken over all $\hat{\Q}$-stopping times $u\in[t,T]$. \emph{From here, we see that the change of num\'eraire reduces the problem to pricing an American call option on the asset yield price ratio $\tilde{S}(t)$,} similar to our observation for the European exchange option. The price of the American exchange option also hedges against the exchange option payoff in the sense that $\tilde{V}^A(t,\tilde{S}(t),v(t)) \geq e^{-q_1 t}(\tilde{S}(t)-e^{(q_1-q_2)t})^+$ for all $t\in [0,T)$ and $\tilde{V}^A(T,\tilde{S}(T),v(T)) = e^{-q_1 T}(\tilde{S}(T)-e^{(q_1-q_2)T})^+.$

Before prescribing additional boundary conditions to IPDE \eqref{eqn-PutCall-IPDE-Vtilde} for the American exchange option, we first define the continuation and stopping regions, denoted by $\calC$ and $\calS$, respectively, that divide the domain $[0,T]\times\mathbb{R}_+^2$ of IPDE \eqref{eqn-PutCall-IPDE-Vtilde}. These regions are given by
\begin{align}
\begin{split}
\label{eqn-PutCall-StoppingContinuationRegions}
\calS & = \left\{(t,\tilde{S},v)\in[0,T]\times\mathbb{R}_+^2: \tilde{V}^A(t,\tilde{S},v) = e^{-q_1 t}\left(\tilde{S}-e^{(q_1-q_2)t}\right)^+\right\}\\
\calC & = \left\{(t,\tilde{S},v)\in[0,T]\times\mathbb{R}_+^2: \tilde{V}^A(t,\tilde{S},v) > e^{-q_1 t}\left(\tilde{S}-e^{(q_1-q_2)t}\right)^+\right\}.
\end{split}
\end{align}
Denote by $\calS(t)$ and $\calC(t)$ the stopping and continuation regions at a fixed $t\in[0,T]$.

For the American exchange option, there exists a critical stock price ratio $B(t,v)\geq 1$ such that the stopping and continuation regions can be written as
\begin{align}
\begin{split}
\label{eqn-PutCall-StoppingContinuationRegions2}
\calS & = \left\{(t,\tilde{S},v)\in[0,T]\times\mathbb{R}_+^2: \tilde{S}\geq B(t,v)e^{(q_1-q_2)t}\right\}\\
\calC & = \left\{(t,\tilde{S},v)\in[0,T]\times\mathbb{R}_+^2: \tilde{S}< B(t,v)e^{(q_1-q_2)t}\right\}
\end{split}
\end{align}
\citep{BroadieDetemple-1997, Touzi-1999}.\footnote{\citet{Mishura-2009} analyzed, in further detail, the properties of the exercise region of the finite-maturity American exchange option in a pure diffusion setting. In the same setting, \citet{Villeneuve-1999} established the nonemptiness of exercise regions of American rainbow options, which include spread and exchange options as special cases.} For a fixed $t\in[0,T]$ and $v\in(0,\infty)$, the early exercise boundary and the continuation and stopping regions are illustrated in Figure \ref{fig-EarlyExerciseBoundaryOptionPrice}. It is known that in the continuation region the American exchange option behaves like its live European counterpart, and so $\tilde{V}^A$ satisfies IPDE \eqref{eqn-PutCall-IPDE-Vtilde} for $(t,\tilde{S},v)\in\calC$.

\begin{figure}
\centering
\includegraphics[width = 0.5\textwidth]{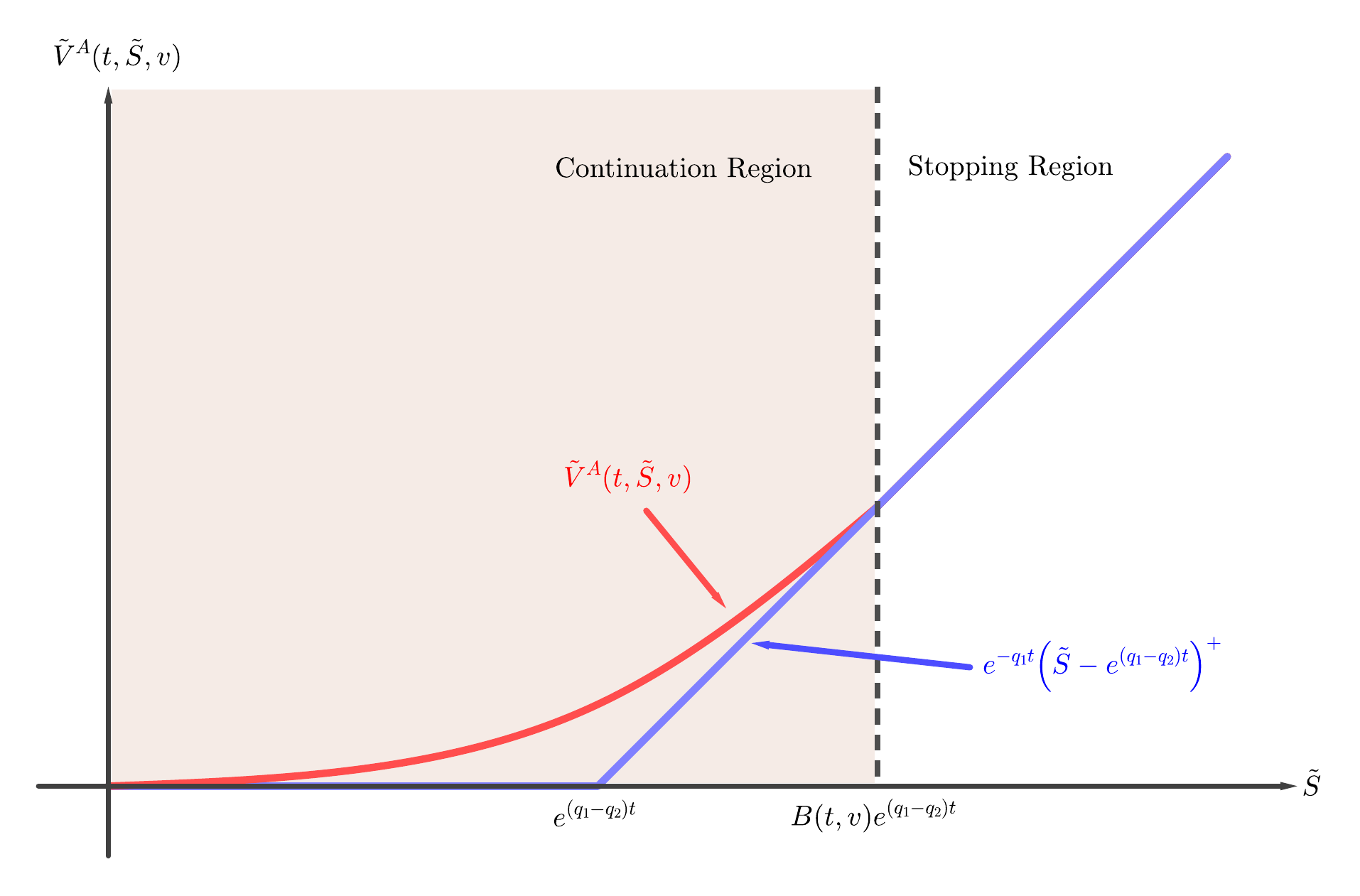}
\caption{The early exercise boundary and the continuation and stopping regions for the transformed American exchange option.}
\label{fig-EarlyExerciseBoundaryOptionPrice}
\end{figure}

We require value-matching and smooth-pasting conditions on IPDE \eqref{eqn-PutCall-IPDE-Vtilde} to elliminate arbitrage opportunities and to ensure that $\tilde{V}^A$ and $\partial\tilde{V}^A/\partial\tilde{s}$ are both continuous across the early exercise boundary $A(t,v)\equiv B(t,v)e^{(q_1-q_2)t}$. Specifically, the required value-matching condition is
\begin{equation}
\label{eqn-PutCall-ValueMatchingCondition-Vtilde}
\tilde{V}^A(t,A(t,v),v(t)) = e^{-q_1 t}\left(A(t,v)-e^{(q_1-q_2)t}\right),
\end{equation}
and the smooth-pasting conditions are
\begin{align}
\begin{split}
\label{eqn-PutCall-SmoothPastingCondition-Vtilde}
\lim_{\tilde{S}\to A(t,v)} \pder[\tilde{V}^A]{\tilde{s}}(t,\tilde{S}(t),v(t)) & = e^{-q_1 t}\\
\lim_{\tilde{S}\to A(t,v)} \pder[\tilde{V}^A]{v}(t,\tilde{S}(t),v(t)) & = 0\\
\lim_{\tilde{S}\to A(t,v)} \pder[\tilde{V}^A]{t}(t,\tilde{S}(t),v(t)) & = -q_1 e^{-q_1 t}\tilde{S}(t)+q_2 e^{-q_2 t}.
\end{split}
\end{align}
Therefore, $\tilde{V}^A(t,\tilde{S}(t),v(t))$ is a solution to IPDE \eqref{eqn-PutCall-IPDE-Vtilde} over the domain $0\leq t\leq T$, $0<\tilde{S}<A(t,v)$, $0<v<\infty$. The IPDE has terminal and boundary conditions
\begin{align}
\begin{split}
\label{eqn-PutCall-BoundaryConditions-Vtilde}
\tilde{V}(T,\tilde{S}(T),v(T)) & = e^{-q_1 T}\left(\tilde{S}(T)-e^{(q_1-q_2)T}\right)^+\\
\tilde{V}(t,0,v(t)) & = 0,
\end{split}
\end{align}
value-matching condition \eqref{eqn-PutCall-ValueMatchingCondition-Vtilde} and smooth-pasting condition \eqref{eqn-PutCall-SmoothPastingCondition-Vtilde}.

The American exchange option price $\tilde{V}^A(t,\tilde{S}(t),v(t))$ can be decomposed into the sum of the discounted European exchange option price $\tilde{V}(t,\tilde{S}(t),v(t))$ and an early exercise premium.

\begin{prop}
\label{prop-PutCall-EarlyExerciseRepresentation}
Suppose Assumption \ref{asp-PutCall-Differentiability} also holds for $\tilde{V}^A(t,\tilde{S},v)$. Assume further that the smooth pasting conditions \eqref{eqn-PutCall-SmoothPastingCondition-Vtilde} across the early exercise boundary hold. 
Then $\tilde{V}^A(t,\tilde{S}(t),v(t))$ can be expressed as
\begin{equation}
\label{eqn-PutCall-EarlyExerciseRepresentation}
\tilde{V}^A(t,\tilde{S}(t),v(t)) = \tilde{V}(t,\tilde{S}(t),v(t)) + \tilde{V}^P(t,\tilde{S}(t),v(t)),
\end{equation}
where $\tilde{V}(t,\tilde{S}(t),v(t))$ is the discounted European exchange option price given by equation \eqref{eqn-PutCall-DiscountedExcOpPrice2} and $\tilde{V}^P(t,\tilde{S}(t),v(t))$ is the early exercise premium given by 
\begin{align*}
& \tilde{V}^P(t,\tilde{S}(t),v(t))\\ 
& \qquad = \E_{\hat{\Q}}\left[\left.\int_t^T\left(q_1 e^{-q_1 s}\tilde{S}(s)-q_2 e^{-q_2 s}\right)\vm{1}(\calA(s))\dif s\right|\calF_t\right]\\
& \qquad \qquad -\tilde{\lambda}_1 \E_{\hat{\Q}}\left[\int_t^T \E_{\hat{\Q}}^{Y_1}\left[\left(\tilde{V}^A\left(s,\tilde{S}(s)e^{Y_1},v(s)\right)\right.\right.\right.\\
& \hspace{100pt} -\left.\left.\left.\left.\left(e^{-q_1 s}\tilde{S}(s)e^{Y_1}-e^{-q_2 s}\right)\right)\vm{1}(\calA_1(s))\right]\dif s\right|\calF_t\right]\\
& \qquad \qquad -\tilde{\lambda}_2 \E_{\hat{\Q}}\left[\int_t^T \E_{\hat{\Q}}^{Y_2}\left[\left(\tilde{V}^A\left(s,\tilde{S}(s)e^{-Y_2},v(s)\right)\right.\right.\right.\\
& \hspace{100pt} -\left.\left.\left.\left.\left(e^{-q_1 s}\tilde{S}(s)e^{-Y_2}-e^{-q_2 s}\right)\right)\vm{1}(\calA_2(s))\right]\dif s\right|\calF_t\right],
\end{align*}
where $\vm{1}(\cdot)$ is the indicator function, $\calA(s)$ is the event $\{(\tilde{S}(s),v(s))\in\calS(s)\}$ and
\begin{align*}
\calA_1(s)	& = \left\{B(s,v)e^{(q_1-q_2)s}\leq \tilde{S}(s) < B(s,v)e^{(q_1-q_2)s}e^{-Y_1}\right\}\\
\calA_2(s)	& = \left\{B(s,v)e^{(q_1-q_2)s}\leq \tilde{S}(s) < B(s,v)e^{(q_1-q_2)s}e^{Y_2}\right\}.
\end{align*}
Note that all partial derivatives in $\partial\tilde{V}^A/\partial t+\calL_{\tilde{s},v}[\tilde{V}^A(s,\tilde{S}(s),v(s))]$ are all evaluated at $(s,\tilde{S}(s),v(s))$.
\end{prop}

\begin{proof}
The proof follows the same process and techniques used in \citet[Proposition 6.1]{CheangGarces-2019}. We provide the proof for this specific case in \citet{GarcesCheang-2020}.
\end{proof}

The early exercise premium can be decomposed into a diffusion component (the positive term) and a jump component (the negative terms). However, unlike the early exercise premium derived by \citet{CheangChiarellaZiogas-2013} for an American call option under SVJD dynamics, our early exercise premium representation contains two jump terms. This is because $\tilde{S}(t)$ has two sources of jumps: the jumps in the price of the first asset and the jumps in the num\'eraire. Nonetheless, the interpretation remains the same: the diffusion term captures the discounted expected value of cash flows due to dividends when asset prices are in the stopping region and the jump terms capture the rebalancing costs incurred by the holder of the American exchange option when a jump instantaneously occurs in the price of either asset, causing $\tilde{S}(t)$ to jump back into the continuation region immediately after the option is exercised.\footnote{In this situation, the investor is unable to adjust the decision to exercise in response to the instantaneous jump in asset prices and is therefore vulnerable to the rebalancing cost described earlier. A similar phenomenon in the context of consumption-investment problems with transaction costs in a L\'evy-driven market is explored in greater technical detail by \citet{deValliere-2016}.} Figure \ref{fig-RebalancingCosts} illustrates the loss (captured by the difference in option value and the exercise value) incurred by the option holder when the asset yield ratio jumps back into the continuation region due to a jump, by a factor $e^{Y_1}$, in the price of the first asset. A similar graphical analysis holds if the price of the num\'eraire asset instantaneously jumps instead of the price of the first asset (in this case, the new asset yield ratio is $\tilde{S}(t)e^{-Y_2}$).

\begin{figure}
\centering
\includegraphics[width = 0.5\linewidth]{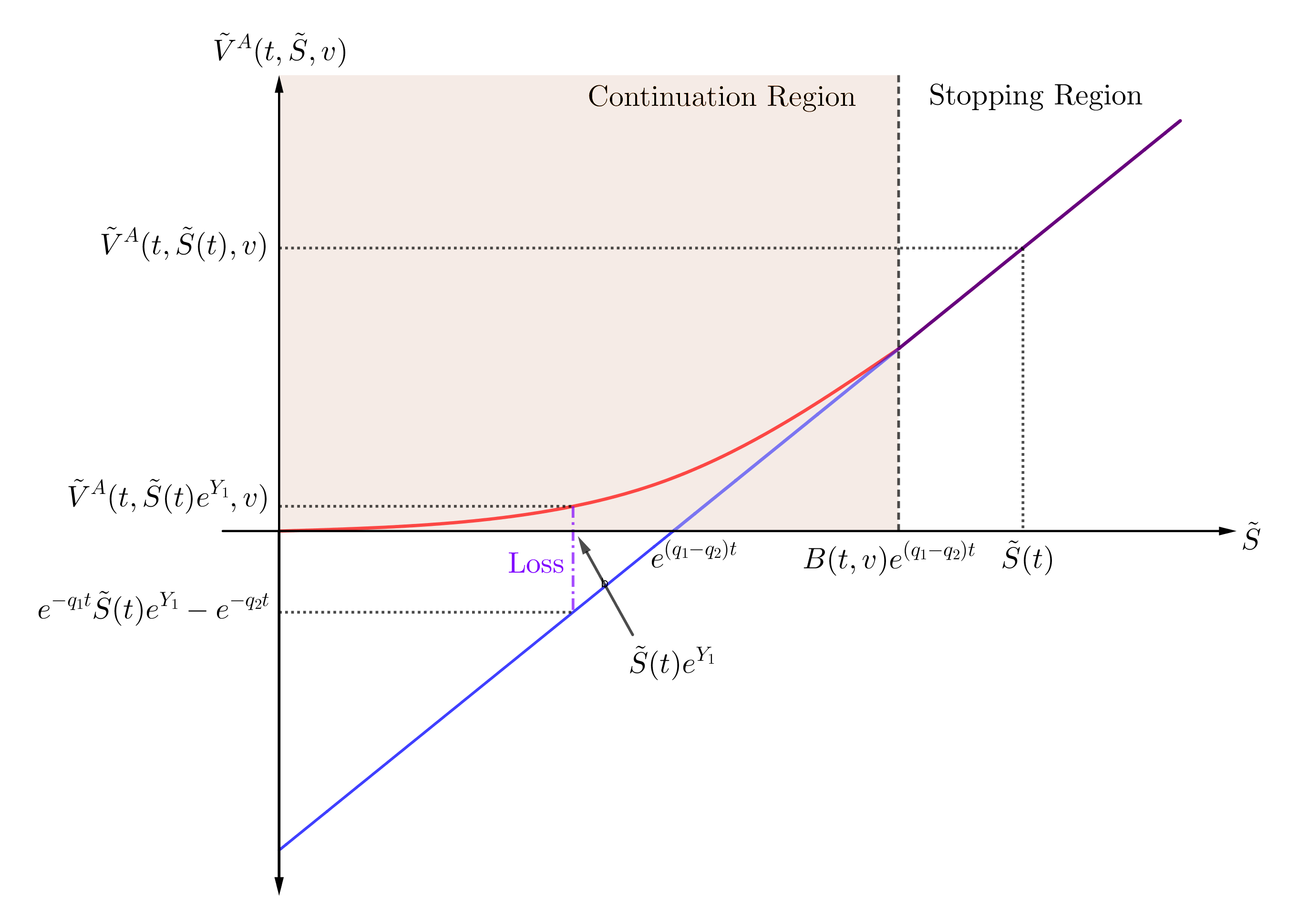}
\caption{Loss incurred by the option holder when the asset ratio instantaneously jumps from $\tilde{S}(t)$ in the stopping region to $\tilde{S}(t)e^{Y_1}$ back in the continuation region.}
\label{fig-RebalancingCosts}
\end{figure}

Recall that $\tilde{V}^A(t,\tilde{S}(t),v(t))$ is a solution of the homogeneous IPDE \eqref{eqn-PutCall-IPDE-Vtilde} over the \emph{restricted domain} $0\leq t\leq T$, $0<\tilde{S}(t)<B(t,v)e^{(q_1-q_2)t}$, and $0<v<\infty$ subject to the value-matching condition \eqref{eqn-PutCall-ValueMatchingCondition-Vtilde}, the smooth-pasting condition \eqref{eqn-PutCall-SmoothPastingCondition-Vtilde}, and boundary conditions \eqref{eqn-PutCall-BoundaryConditions-Vtilde}. Following \citet{Jamshidian-1992} and \citet{ChiarellaZiogas-2004}, the restriction on the domain can be removed by adding the appropriate inhomogeneous term to the IPDE such that the equation holds for all $\tilde{S}(t)>0$. The inhomogeneous IPDE corresponding to our analysis is presented in the following proposition. This analysis requires that $\tilde{V}^A(t,\tilde{S}(t),v(t))$ and its first-order partial derivative with respect to $\tilde{S}$ are continuous, but the value-matching and smooth-pasting conditions are sufficient to meet this requirement.

\begin{prop}
\label{prop-PutCall-InhomogeneousTerm}
The discounted American exchange option price $\tilde{V}^A(t,\tilde{S},v)$ is a solution to the inhomogeneous IPDE
\begin{equation}
\label{eqn-PutCall-InhomogeneousIPDE}
0 = \pder[\tilde{V}^A]{t}+\calL_{\tilde{s},v}\left[\tilde{V}^A(t,\tilde{S}(t),v(t))\right]+\Xi(t,\tilde{S}(t),v(t)),
\end{equation}
where the inhomogeneous term $\Xi$ is given by
{\small\begin{align}
\begin{split}
\label{eqn-PutCall-InhomogeneousTerm}
& \Xi(t,\tilde{S}(t),v(t))\\
	& = \left(q_1 e^{-q_1 t}\tilde{S}(t)-q_2 e^{-q_2 t}\right)\vm{1}(\calA(t))\\
	& \quad - \tilde{\lambda}_1\vm{1}(\calA(t))\int_{-\infty}^{b(t,\tilde{S}(t),v(t))}\left[\tilde{V}^A\left(t,\tilde{S}(t)e^{y},v(t)\right)-\left(e^{-q_1 t}\tilde{S}(t)e^y-e^{-q_2 t}\right)\right]G_1(y)\dif y\\
	& \quad - \tilde{\lambda}_2\vm{1}(\calA(t))\int_{-b(t,\tilde{S}(t),v(t))}^{\infty}\left[\tilde{V}^A\left(t,\tilde{S}(t)e^{-y},v(t)\right)-\left(e^{-q_1 t}\tilde{S}(t)e^{-y}-e^{-q_2 t}\right)\right]G_2(y)\dif y,
\end{split}
\end{align}}
where $G_1$ and $G_2$ are the pdfs of $Y_1$ and $Y_2$, respectively, under $\hat{\Q}$, and $b(t,\tilde{S}(t),v(t)) \equiv \ln[B(t,v(t))e^{(q_1-q_2)t}/\tilde{S}(t)]$. This equation is to be solved for $(t,\tilde{S}(t),v(t))\in[0,T]\times\mathbb{R}_+^2$, subject to terminal and boundary conditions \eqref{eqn-PutCall-BoundaryConditions-Vtilde}.
\end{prop}

\begin{proof}
Observe that for all $(t,\tilde{S}(t),v(t))\in[0,T]\times\mathbb{R}_+^2$, the equation $$0 = \pder[\tilde{V}^A]{t}+\calL_{\tilde{s},v}\left[\tilde{V}^A(t,\tilde{S}(t),v(t))\right] - \left(\pder[\tilde{V}^A]{t}+\calL_{\tilde{s},v}\left[\tilde{V}^A(t,\tilde{S}(t),v(t))\right]\right)\vm{1}(\calA(t))$$ holds. Equation \eqref{eqn-PutCall-InhomogeneousTerm} is obtained by expanding the negative term in the above equation with the functional form of $\tilde{V}^A$ in the stopping region and using the $G_1$ and $G_2$ to rewrite the expectations as integrals.
\end{proof}

\section{Limit of the Early Exercise Boundary at Maturity}
\sectionmark{Limit of the Early Exercise Boundary}
\label{sec-PutCall-EEBLimit}

When implementing numerical schemes to solve for the American exchange option price, it is important to know the behavior of the critical asset price ratio $B(t,v)$ associated to the early exercise boundary of the American exchange option immediately before the option matures (i.e. as $t\to T^-$). Doing so provides a terminal condition $B(T,v)$ for the unknown early exercise boundary that can be used in conjunction with the terminal condition on the discounted American exchange option price. The next proposition presents the limit of the early exercise boundary, which we obtain following the method of \citet{ChiarellaZiogas-2009}.

\begin{prop}
\label{prop-PutCall-EEBLimit}
The limit $B(T^-,v)\equiv\lim_{t\to T^-}B(t,v)$ is a solution of the equation
\begin{equation}
\label{eqn-PutCall-EEBLimit}
B(T^-,v) = \max\left\{1,\frac{q_2+\tilde{\lambda}_1\int_{-\infty}^{-\ln B(T^-,v)} G_1(y)\dif y + \tilde{\lambda}_2\int_{\ln B(T^-,v)}^\infty G_2(y)\dif y}{q_1 +\tilde{\lambda}_1\int_{-\infty}^{-\ln B(T^-,v)} e^y G_1(y)\dif y + \tilde{\lambda}_2\int_{\ln B(T^-,v)}^\infty e^{-y} G_2(y)\dif y}\right\}.
\end{equation}
\end{prop}

\begin{proof}
See Appendix \ref{app-Proof-EEBLimit}
\end{proof}

Equation \eqref{eqn-PutCall-EEBLimit} must be solved implicitly for $B(T^-,v)$, which can be done using standard root-finding techniques. From our analysis, we find that the limit is dependent on the asset dividend yields $q_1$ and $q_2$, the jump intensities $\tilde{\lambda}_1$ and $\tilde{\lambda}_2$, and the jump size densities $G_1$ and $G_2$. These dependencies highlight the influence of jumps in asset prices on the limiting behavior of the early exercise boundary.\footnote{This is in contrast to the proposition of \citet{CarrHirsa-2003}, in their analysis of the one-asset American put option where the log-price is driven by a L\'evy process, that the limit of the early exercise boundary is only dependent on the dividend yield and the risk-free rate.} We note further that equation \eqref{eqn-PutCall-EEBLimit} does not depend on the instantaneous variance $v$ since the option payoff is independent of $v$. However, equation \eqref{eqn-PutCall-EEBLimit} is true for all $v\in(0,\infty)$. In fact, $B(T^-,v)$ is constant with respect to $v$.

In the absence of jumps (i.e. when $\tilde{\lambda}_1=\tilde{\lambda}_2=0$), the limit reduces to $\max\{1,q_2/q_1\}$. This is consistent with the result of \citet{BroadieDetemple-1997} for American exchange options in the pure diffusion case. In the pure diffusion case, $B(T^-,v) = q_2/q_1 > 1$ if $q_2>q_1$, implying that the early exercise boundary many not be continuous in $t$ at maturity. When jumps are present, the analysis of continuity becomes more complicated, as shown below.

First, we present some conditions under which equation \eqref{eqn-PutCall-EEBLimit} has a solution.

\begin{prop}
\label{prop-PutCall-EEBLimit-Existence}
Suppose $q_1,q_2\geq 0$ and $\tilde{\lambda}_1,\tilde{\lambda}_2>0$ are given and let $G_1$ and $G_2$ be continious probability density functions. The equation
\begin{equation}
\label{eqn-PutCall-EEBLimit-ImplicitComponent}
x = \frac{q_2+\tilde{\lambda}_1\int_{-\infty}^{-\ln x} G_1(y)\dif y + \tilde{\lambda}_2\int_{\ln x}^\infty G_2(y)\dif y}{q_1 +\tilde{\lambda}_1\int_{-\infty}^{-\ln x} e^y G_1(y)\dif y + \tilde{\lambda}_2\int_{\ln x}^\infty e^{-y} G_2(y)\dif y}
\end{equation}
has a unique solution $x^*\in(0,\infty)$ if $q_1>0$. Furthermore, $x^*>1$ if and only if $$q_2-q_1+\tilde{\lambda}_1\int_{-\infty}^0 (1-e^y)G_1(y)\dif y+\tilde{\lambda}_2\int_0^\infty (1-e^{-y})G_2(y)\dif y > 0,$$ and the limit of the early exercise boundary is $B(T^-,v)=\max\{1,x^*\}$.
\end{prop}

\begin{proof}
See Appendix \ref{app-Proof-EEBLimit-Existence}
\end{proof}

An immediate result from Proposition \ref{prop-PutCall-EEBLimit-Existence} is a condition for the continuity of $B(t,v)$ at maturity.

\begin{prop}
Suppose $q_1>0$. For any fixed $v\in(0,\infty)$, $B(t,v)$ is continuous at maturity $t=T$ if
\begin{equation}
\label{eqn-PutCall-EEBLimit-ContinuityCondition}
q_1-q_2\geq \tilde{\lambda}_1\int_{-\infty}^0 (1-e^y)G_1(y)\dif y+\tilde{\lambda}_2\int_0^\infty (1-e^{-y})G_2(y)\dif y.
\end{equation}
\end{prop}

\begin{proof}
Suppose $q_1>0$ and condition \eqref{eqn-PutCall-EEBLimit-ContinuityCondition} holds. Then from the discussion at the end of the proof of Proposition \ref{prop-PutCall-EEBLimit-Existence}, the solution $x^*$ to equation \eqref{eqn-PutCall-EEBLimit-ImplicitComponent} lies in the interval $(0,1]$. It follows that $B(T^-,v)=1$, which is also the value of $B(T,v)$. Thus, $B(t,v)$ is continuous at the option maturity.
\end{proof}

In other words, if the dividend yield of the first asset exceeds that of the second asset by the amount given by the right-hand side of \eqref{eqn-PutCall-EEBLimit-ContinuityCondition}, then the early exercise boundary is continuous at maturity. 

We briefly discuss the behavior of $B(T^-,v)$ with respect to changes in $q_1$. Note that $\partial f/\partial q_1 = -x<0$, so when $q_1$ decreases, $f(x)$ increases. In particular, for a given $q_1>0$, there exists $x^*\in(0,\infty)$ such that $f(x^*)=0$. If $q_1$ decreases, then $f(x^*)$ increases away from zero, thereby moving the unique zero of $f$ to some other number $x'\in(0,\infty)$ such that $x'>x^*$. In other words, the solution $x^*$ of equation \eqref{eqn-PutCall-EEBLimit-ImplicitComponent} increases without bound, and consequently $B(T^-,v)\to\infty$, as $q_1\to 0^+$. \emph{Thus, when the first asset bears no dividend yield, it is not optimal to exercise the American exchange option early or at least immediately prior to the option maturity.}

\begin{figure}
\centering
\includegraphics[width = 0.5\textwidth]{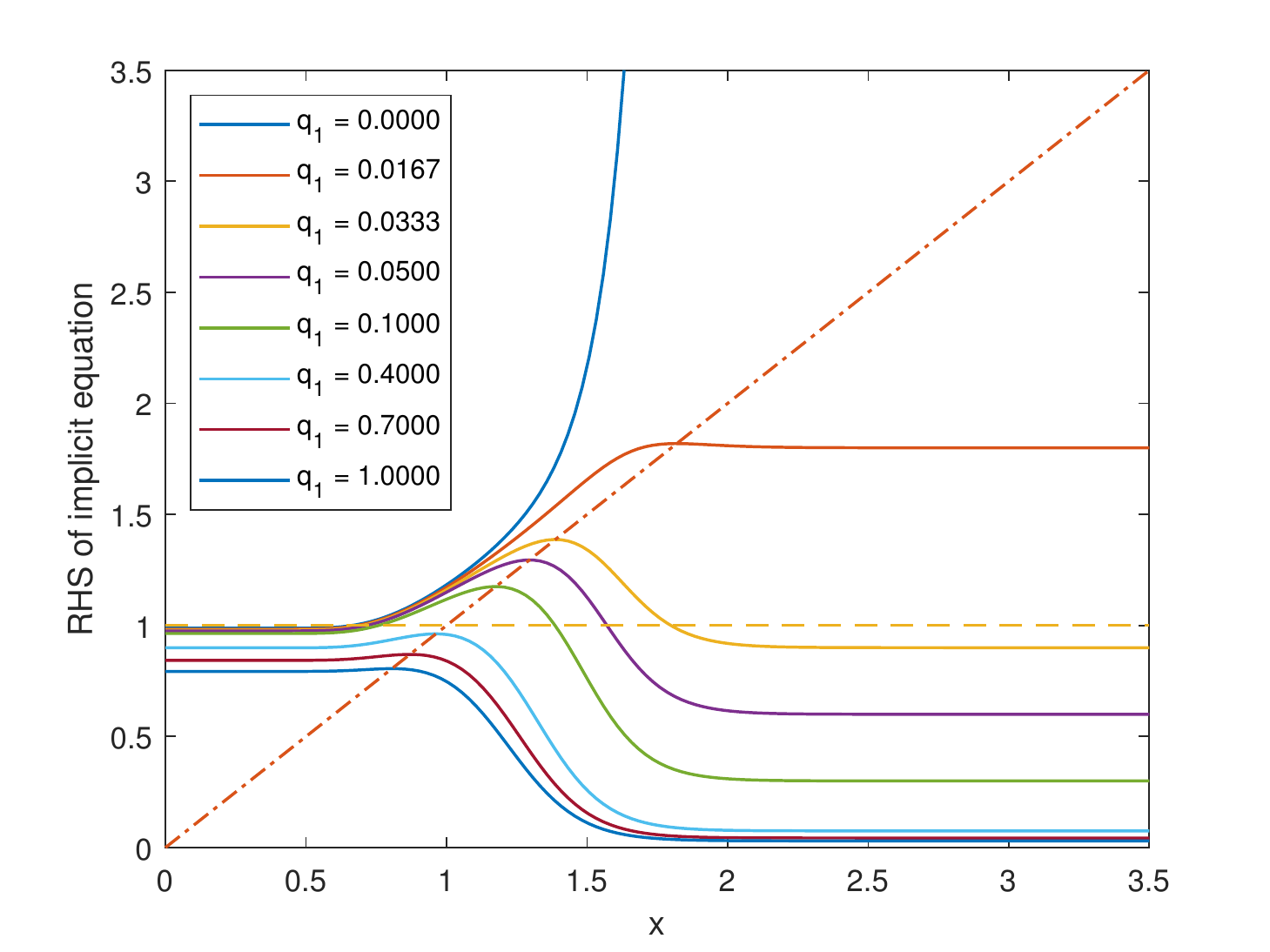}
\caption{Behavior of equation \eqref{eqn-PutCall-EEBLimit-ImplicitComponent} with respect to $q_1$ when jumps are normally distributed. Solid lines represent the right-hand side of \eqref{eqn-PutCall-EEBLimit-ImplicitComponent} and their intersection with the dash-dotted $45^\circ$ line represents the solution $x^*$ of \eqref{eqn-PutCall-EEBLimit-ImplicitComponent}. The horizontal dashed line indicates the position $x^*$ relative to unity.}
\label{fig-EEBLimit}
\end{figure} 

We give a numerical illustration of the aforementioned properties in the case of when the jump size random variables $Y_1$ and $Y_2$ have normal distributions $N(\alpha_{j_1},\beta_{j_1}^2)$ and $N(\alpha_{j_2},\beta_{j_2}^2)$, respectively, under $\hat{\Q}$. The integrals in equation \eqref{eqn-PutCall-EEBLimit-ImplicitComponent} can then be expressed in terms of the standard normal distribution. 
Figure \ref{fig-EEBLimit} shows the behavior of the right-hand side of equation \eqref{eqn-PutCall-EEBLimit-ImplicitComponent}, which we denote by $R(x)$, for increasing values of $q_1$ given this specific distribution of jump sizes and fixed parameters $q_2 = 0.03$, $\tilde{\lambda}_1 = \tilde{\lambda}_2 = 2$, $\alpha_{j_1}=\alpha_{j_2} = 0.02$, and $\beta_{j_1}=\beta_{j_2} = 0.2$. We observe that, except for when $q_1 = 0$, the behavior of $R(x)$ is generally the same for any value $q_1$. That is, it increases up until some value of $x$, then eventually decreases to $q_2/q_1$. The graph of $R(x)$ eventually crosses the $45^\circ$ line, indicating the solution $x^*$ of equation \eqref{eqn-PutCall-EEBLimit-ImplicitComponent}. However, if $q_1 = 0$, then the graph of $R(x)$ never crosses the line, implying that there is no solution exists for \eqref{eqn-PutCall-EEBLimit-ImplicitComponent} and that the equation will only be satisfied by taking the limit as $x\to 0$, giving an infinite early exercise boundary $B(T^-,v)$. We also note that as $q_1$ increases farther from $q_2$, $x^*$ eventually falls below unity. One can also formulate the difference $q_1-q_2$ in \eqref{eqn-PutCall-EEBLimit-ContinuityCondition} in terms of the standard normal distribution 

\begin{figure}
\centering
\includegraphics[width = 0.5\textwidth]{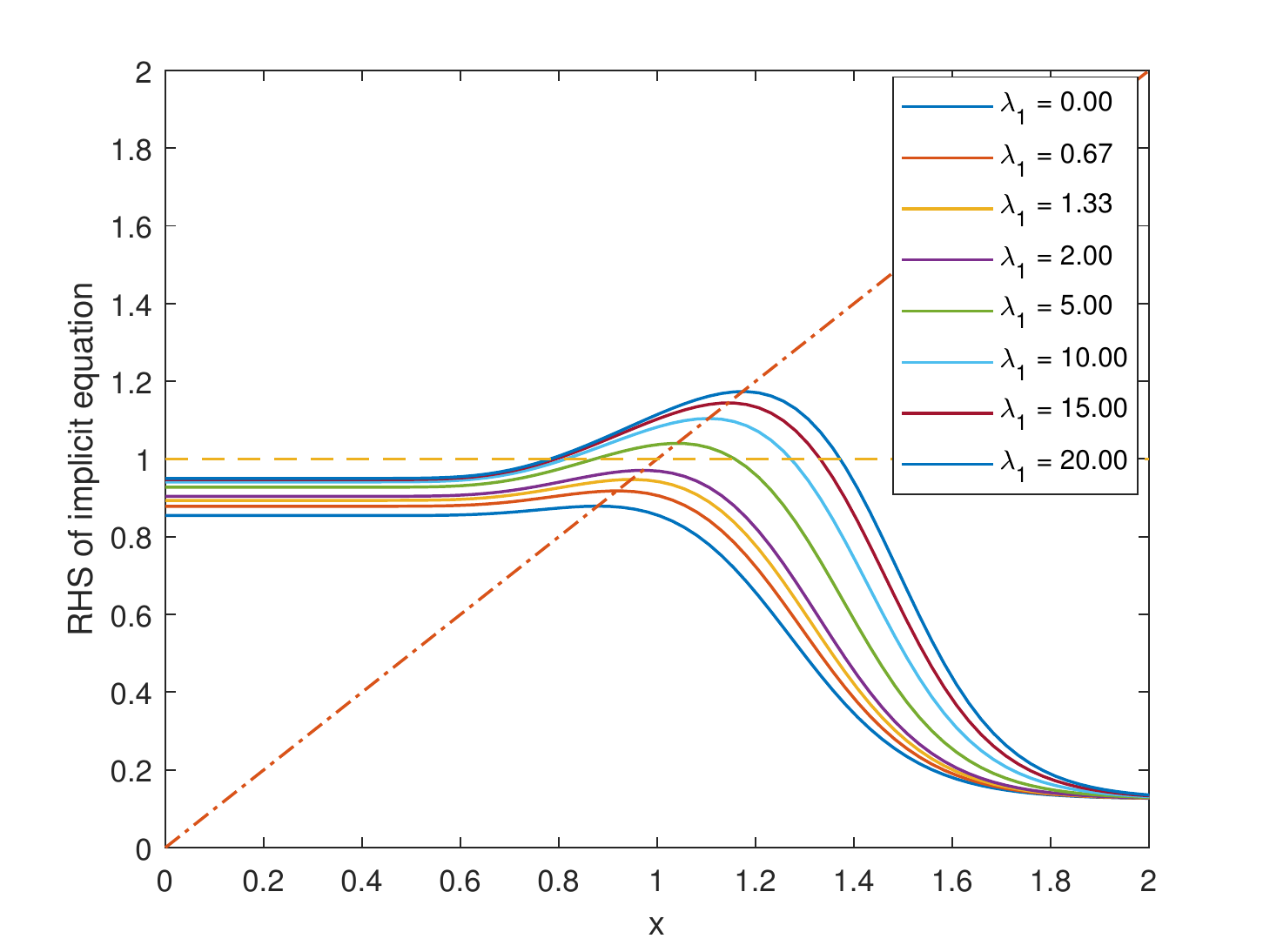}
\caption{Behavior of equation \eqref{eqn-PutCall-EEBLimit-ImplicitComponent} with respect to $\tilde{\lambda}_1$ when jumps are normally distributed. The dashed and dash-dotted lines function in a way similar to that for Figure \ref{fig-EEBLimit}.}
\label{fig-EEBLimit-Lambda}
\end{figure} 

A comparative static analysis of $B(T^-,v)$ with respect to the jump intensities $\tilde{\lambda}_1$ and $\tilde{\lambda}_2$ can also be conducted. We note that
\begin{align*}
\pder[f(x)]{\tilde{\lambda}_1} & = \int_{-\infty}^{-\ln x}G_1(y) \dif y - x\int_{-\infty}^{-\ln x}e^y G_1(y)\dif y \geq 0\\
\pder[f(x)]{\tilde{\lambda}_2} & = \int_{\ln x}^\infty G_2(y)\dif y - x\int_{\ln x}^\infty\ e^{-y}G_2(y)\dif y \geq 0,
\end{align*}
as was established in the proof of Proposition \ref{prop-PutCall-EEBLimit-Existence}. With the same arguments as above, these inequalities imply that $B(T^-,v)$ are non-decreasing with $\tilde{\lambda}_1$ and $\tilde{\lambda}_2$. This is graphically illustrated by Figure \ref{fig-EEBLimit-Lambda}, where we once again plot the right-hand side of equation \eqref{eqn-PutCall-EEBLimit-ImplicitComponent} for various values of $\tilde{\lambda}_1$. In this example, we set $q_1 = 0.40$ and $q_2 = 0.05$ (chosen such that $x^*$ below and above unity are illustrated), with the value of all other jump parameters chosen similarly to those for Figure \ref{fig-EEBLimit}. A similar illustration can also be generated for a comparative static analysis with respect to $\tilde{\lambda}_2$.

The response of the early exercise boundary limit at maturity with respect to the jump intensities is indeed expected. As discussed in the previous section, there is a nonzero probability that asset prices may jump back into the continuation region immediately after exercise, resulting to losses for the option holder. The option holder is thus more conservative in opting for early exercise and is willing to exercise early only when the asset yield ratios are higher compared to the no-jump case. This is so that even if the asset yield jumps downward, it is less likely to move all the way back in to the continuation region thereby reducing risk on the investor's end.

\section{Method of Lines Implementation}
\label{sec-MOL}

Here, we discuss the MOL algorithm to numerically solve the IPDEs for the discounted European and American exchange option prices and the joint transition density function (tdf) of the asset yield ratio and the variance process under $\hat{\Q}$. For convenience, the terminal value problem is transformed into an initial value problem involving the time to maturity $\tau=T-t$. We also denote the discounted exchange option prices by $V(\tau,s,v)$, where $s=\tilde{S}(T-\tau)$ and $v=v(T-\tau)$. We first discuss the MOL for the European exchange option (which is also applicable to the joint tdf) and present the required adjustments for the approximation of the unknown exercise boundary for the American exchange option. We conclude this section with a discussion of alternative boundary conditions at the far boundary of $v$.

It has been observed that MOL performs as competitively, if not more efficiently, compared to other methods such as componentwise splitting methods, numerical integration with the joint tdf, Monte Carlo techniques, and full finite difference schemes \citep[among others]{Chiarella-2009, ChiarellaZiveyi-2013, Kang-2014}. This is mainly due to the ability of the MOL to calculate the option delta and gamma to any desired accuracy with no considerable additional computational effort. To illustrate the accuracy and efficiency of the MOL, we compare the MOL American exchange option prices to prices generated by the \citet{LongstaffSchwartz-2001} least-squares Monte Carlo (LSMC) algorithm, a simple and popular simulation-based method for valuing early exercise options.\footnote{The LSMC algorithm only simulates the optimal exercise strategy and is unable to estimate the early exercise boundary. For a simulation-based method for approximating the early exercise frontier, see \citet{Ibanez-2004}. More recently, \citet{Bayer-2020} proposed a new simulation-based method that estimates exercise rates of randomized exercise strategies, an optimization problem which they show is equivalent to the original optimal stopping formulation.}

\subsection{MOL Approximation of the IPDE and the Riccati Transform Method}
\label{sec-MOL-Algorithm}

With the new notation introduced in this section, we find that the discounted European exchange option price satisfies
\begin{align}
\begin{split}
\label{eqn-MOLEu-IPDE}
0	& = \frac{1}{2}\sigma^2 v s^2\pder[^2V]{s^2}+\frac{1}{2}\omega^2 v\pder[^2V]{v^2} + \omega(\sigma_1\rho_1-\sigma_2\rho_2)vs\pder[^2V]{s\partial v}\\
	& \qquad - \left(\tilde{\lambda}_1\tilde{\kappa}_1+\tilde{\lambda}_2\tilde{\kappa}_2^-\right)s\pder[V]{s} + \left[\xi\eta-(\xi+\Lambda)v\right]\pder[V]{v} - (\tilde{\lambda}_1+\tilde{\lambda}_2)V\\
	& \qquad - \pder[V]{\tau} + \tilde{\lambda}_1\int_{\mathbb{R}} V(\tau,se^y,v)G_1(y)\dif y + \tilde{\lambda}_2\int_{\mathbb{R}}V(\tau,se^{-y},v)G_2(y)\dif y.
\end{split}
\end{align}
This equation is to be solved for $s\in(0,\infty)$, $v\in(0,\infty)$, and $\tau\in[0,T]$ subject to the initial condition $$V(0,s,v) = e^{-q_1 T}\left(s-e^{(q_1-q_2)T}\right)^+.$$ We also impose the following (asymptotic) boundary conditions on \eqref{eqn-MOLEu-IPDE}: 
\begin{equation}
\label{eqn-MOLEu-BC}
\lim_{s\to 0^+}V(\tau,s,v) = 0, \quad \lim_{s\to\infty}\pder[^2V]{s^2} = 0, \quad \lim_{v\to\infty}\pder[V]{v} = 0.
\end{equation}
The first boundary condition implies that the option becomes worthless for very small values of $s$ (i.e. when the price of the first asset dominates that of the second asset). The second boundary condition implies that the option delta becomes insensitive to changes in the asset yield ratio when it is sufficiently large. The third boundary condition means that the option price is insensitive to changes in the instantaneous variance for sufficiently large values of $v$.

For computational purposes, the infinite domains for $s$ and $v$ are truncated to $[0,s_J]$ and $[0,v_M]$, respectively, with some suitably chosen $s_J$ and $v_M$. We consider the partition $0=\tau_0<\tau_1<\dots<\tau_N = T$ of $[0,T]$ with uniform width $\Delta\tau$ and the partition $0=v_0<v_1<\dots<v_M$ of $[0,v_M]$, also with uniform size $\Delta v$. Denote by $V_{n,m}(s)$ the approximate solution of \eqref{eqn-MOLEu-IPDE} at the time line $\tau=\tau_n$ and variance line $v=v_m$, and let $\scrV_{n,m}(s)=V_{n,m}'(s)$ be the approximation of the option delta.

We then use finite difference quotients to discretize equation \eqref{eqn-MOLEu-IPDE} in all variables but $s$. For the time derivative, a backward difference approximation is used following \citet{vanderHoek-1997},
\begin{equation*}
\pder[V]{\tau} \approx \begin{cases}
\frac{1}{\Delta\tau}(V_{n,m}(s)-V_{n-1,m}(s)) & \text{if $n=1,2$}\\
\frac{3}{2\Delta\tau}(V_{n,m}(s)-V_{n-1,m}(s))-\frac{1}{2\Delta\tau}(V_{n-1,m}(s)-V_{n-2,m}(s)) & \text{if $n\geq 3$}.
\end{cases}
\end{equation*}
The approximation is only first-order accurate for the first two time steps, but for all subsequent time steps, a second-order approximation is used. Using such a scheme improves the accuracy of the method of lines algorithm. The second-order partial derivatives in $v$ are approximated using central difference quotients
\begin{align*}
\pder[^2V]{v^2} & \approx \frac{V_{n,m+1}(s)-2V_{n,m}(s)+V_{n,m-1}(s)}{(\Delta v)^2}\\
\pder[^2V]{s\partial v} & \approx \frac{\scrV_{n,m+1}(s)-\scrV_{n,m-1}(s)}{2\Delta v}.
\end{align*}
We use an upwinding difference approximation for $\partial V/\partial v$,
\begin{align*}
\left[\xi\eta-(\xi+\Lambda)v\right]\pder[V]{v}
	& \approx \max\left\{\xi\eta-(\xi+\Lambda)v_m,0\right\}\left[\frac{V_{n,m+1}(s)-V_{n,m}(s)}{\Delta v}\right]\\
	& \qquad + \min\left\{\xi\eta-(\xi+\Lambda)v_m,0\right\}\left[\frac{V_{n,m}(s)-V_{n,m-1}(s)}{\Delta v}\right].
\end{align*}
This approximation is only first-order accurate, but it helps keep the finite difference approximation in $v$ stable as the coefficients of the second-order partial derivatives vanish as $v\to 0^+$ and it improves the performance of the iterative method used in the method of lines algorithm \citep{Chiarella-2009, Meyer-2015}.

At the last variance line $v=v_M$, we invoke the asymptotic boundary condition \eqref{eqn-MOLEu-BC} in $v$. The boundary condition can be approximated by 
\begin{equation}
\label{eqn-MOLEu-vBC}
\pder[V(\tau_n,s,v_M)]{v} = 0,
\end{equation}
and a forward difference approximation of this implies that $V_{n,M+1}(s) = V_{n,m}(s)$. Differentiating with respect to $s$, we also find that $\scrV_{n,M+1}(s) = \scrV_{n,M}(s)$.

The integral terms must also be approximated at each time and variance line. Assuming that $Y_1\sim N(\alpha_{j_1},\beta_{j_1}^2)$ and $Y_2\sim N(\alpha_{j_2},\beta_{j_2}^2)$ under $\hat{\Q}$, the integrals, which we denote by $I_1(\tau_n,x,v_m)$ and $I_2(\tau_n,x,v_m)$, can be approximated using Gauss-Hermite quadrature,
\begin{align*}
I_1(\tau_n,s,v_m) & = -\tilde{\lambda}_1\int_{\mathbb{R}}V(\tau_n,se^{y},v_m)G_1(y)\dif y \approx -\frac{\tilde{\lambda}_1}{\sqrt{\pi}}\sum_{l=1}^L \varrho_l^H V_{n,m}(se^{\sqrt{2}\beta_{j_1}z_l^H+\alpha_{j_1}})\\
I_2(\tau_n,s,v_m) & = -\tilde{\lambda}_2\int_{\mathbb{R}}V(\tau_n,se^{-y},v_m)G_2(y)\dif y \approx -\frac{\tilde{\lambda}_2}{\sqrt{\pi}}\sum_{l=1}^L \varrho_l^H V_{n,m}(se^{-\sqrt{2}\beta_{j_2}z_l^H-\alpha_{j_2}}).
\end{align*}
where $\varrho_l^H$ and $z_l^H$ are the weights and abscissas of the Gauss-Hermite quadrature scheme with $L$ integration points. The parameters $\tilde{\kappa}_1$ and $\tilde{\kappa}_2^-$ are also approximated using Gauss-Hermite quadrature.

The MOL approximation of equation \eqref{eqn-MOLEu-IPDE} at $\tau=\tau_n$ and $v=v_m$ is therefore
\begin{align}
\begin{split}
\label{eqn-MOLEu-IPDE-Approx2}
& a(s,v_m) V_{n,m}''(s) + b(s,v_m) V_{n,m}'(s) - c(\tau_n,s,v_m) V_{n,m}(s)\\
& \qquad = F(\tau_n,s,v_m)+I_1(\tau_n,s,v_m)+I_2(\tau_n,s,v_m),
\end{split}
\end{align}
where the coefficients are obtained by comparing \eqref{eqn-MOLEu-IPDE-Approx2} to \eqref{eqn-MOLEu-IPDE} after substituting the finite difference and integral approximations (see Appendix \ref{app-MOL-Coeff}).

As the source term of equation \eqref{eqn-MOLEu-IPDE-Approx2} at each $m=1,\dots,M-1$ depends on $V_{n,m+1}(s)$, $\scrV_{n,m}(s)$, and $\scrV_{n,m+1}(s)$, none of which are available at the current variance line, and involves integrals based on $V_{n,m}(s)$, the sequence of equations \eqref{eqn-MOLEu-IPDE-Approx2} represents a system of coupled second-order integral-ordinary differential equations. To this end, we employ a nested two-level iterative scheme in conjunction with the Riccati transform approach of \citet{vanderHoek-1997}, as was done by \citet{Chiarella-2009}. In the first iteration level, referred to as \emph{integral term iterations}, the integral terms are approximated using the latest available values of the option price. This enables us to treat \eqref{eqn-MOLEu-IPDE-Approx2} as a second-order ODE which can then be solved using the Riccati transform method in the second iteration level, referred to as the \emph{variance line iterations}.

The two-stage iterative process and the Riccati transform solution of equation \eqref{eqn-MOLEu-IPDE-Approx2} is shown in Algorithm \ref{pseudo-MOL-EuExcOp}. The succeeding discussions aim to clarify certain details of the algorithm.

\begin{algorithm}
\scriptsize
\caption{MOL for the discounted European exchange option price}
\label{pseudo-MOL-EuExcOp}
\SetAlgoLined
\KwResult{Option price $V$ and delta $\scrV$}
Inputs: SVJD model parameters, jump distribution parameters, partitions for $[0,T]$, $[0,s_J]$, and $[0,v_M]$\;
Compute Gauss-Hermite parameters $\varrho_l^H$ and $z_l^H$, $l=1,\dots,L$ and compute $\kappa$ using Gauss-Hermite quadrature\;
Initialize $V_{0,m}(s_j) = e^{-q_1 T}(s_j-e^{(q_1-q_2)T})^+$, $\scrV_{0,m}(s_j) = e^{-q_1 T}\text{Heaviside}(s_j-e^{(q_1-q_2)T})$ for all $j,m$\;
\For{$n=1$ \KwTo $N$}{
	Set $V_{n,m}^{k'-1}(s_j) = V_{n-1,m}(s_j)$ and $\scrV_{n,m}^{k'-1}(s_j) = \scrV_{n-1,m}(s_j)$\;
	Initialize $\epsilon_\text{integ} = 1$ and $k'=1$\;
	\tcc{Commence integral term iterations}
	\While{$\epsilon_\text{integ}> 10^{-8}$ \& $k'<k'_{max}$}{
		Compute $I_1(\tau_n,s_j,v_m)$, $I_2(\tau_n,s_j,v_m)$ using $V_{n,m}^{k-1}(s_j)$ for all $j,m$\;
		Set $V_{n,m}^k(s_j) = V_{n,m}^{k'}(s_j)$ and $\scrV_{n,m}^k(s_j) = \scrV_{n,m}^{k'}(s_j)$\;
		Initialize $\epsilon_\text{var} = 1$, $k=1$, and set $k_{max} = 50$\;
		\tcc{Commence variance line iterations}
		\While{$\epsilon_\text{var}>10^{-8}$ \& $k<k_{max}$}{
			\For{$m=0$ \KwTo $M$}{
				\If{$m=1$}{
					\tcc{Quadratic extrapolation formulas}
					Compute $V_{n,0}^k$ and $\scrV_{n,0}^k$ using equations \eqref{eqn-MOL-QuadExtrap}\;
					}
				\ElseIf{$1\leq m\leq M-1$}{
					\tcc{Forward sweep}
					\For{$j=1$ \KwTo $J$}{
						Compute $R(s_j)$ and $w(s_j)$ using the trapezoidal rule on \eqref{eqn-MOLEu-RiccatiSystem}\;
					}
					\tcc{Reverse sweep over $[0,s_J]$}
					Compute $\scrV_{n,m}^k(s_J)$\;
					\For{$j=J$ \KwTo $1$}{
						Compute $\scrV_{n,m}^k(s_j)$ using the trapezoidal rule\;
						Compute $V_{n,m}^k(s_j)$ using the Riccati transform \eqref{eqn-MOLEu-RiccatiTransform}\;
					}
				}
				\Else{ 
					Recalculate the coefficients of \eqref{eqn-MOLEu-RiccatiSystem} with \eqref{eqn-MOLEu-IPDE-Approx2-Coeff-vBC}\tcp{At $v=v_M$}
					Perform the forward sweep and reverse sweep as in when $1\leq m\leq M-1$ to compute $V_{n,M}^k(s_j)$ and $\scrV_{n,M}^k(s_j)$ for all $j$\;
				}
			}
			Update $\epsilon_\text{var} = \max_m\{\max_j\{|V_{n,M}^k(s_j)-V_{n,m}^{k-1}(s_j)|\}$\;
			Update $V_{n,m}^{k-1} = V_{n,m}^{k}$ and $\scrV_{n,m}^{k-1} = \scrV_{n,m}^k$\;
			Update $k=k+1$\;
		}
		Update $V_{n,m}^{k'} = V_{n,m}^k$ and $\scrV_{n,m}^{k'} = \scrV_{n,m}^k$\;
		Update $\epsilon_\text{integ} = \max_m\{\max_j\{|V_{n,M}^{k'}(s_j)-V_{n,m}^{k'-1}(s_j)|\}$\;
		Update $V_{n,m}^{k'-1} = V_{n,m}^{k'}$ and $\scrV_{n,m}^{k'-1} = \scrV_{n,m}^{k'}$\;
		Update $k'=k'+1$
	}
	Update $V_{n,m} = V_{n,m}^{k'}$ and $\scrV_{n,m} = \scrV_{n,m}^{k'}$\;
}
\end{algorithm}

Let $k'$ and $k$ be the iteration counters for the integral term and variance line iterations, respectively, and let $V_{n,m}^k(s)$ and $\scrV_{n,m}^k(s)$ be the $k$th iterates of $V$ and $\scrV$, replacing $k$ with $k'$ when we refer to their counterparts in the integral term iterations. ``Previous'' iterates will be denoted with a superscript $k-1$ or $k'-1$. The Riccati transform method requires that the computational $s$-domain be partitioned with a mesh $0=s_0<s_1<\dots<s_J$. Note that this mesh may not necessarily have a uniform width. In view of boundary conditions \eqref{eqn-MOLEu-BC}, it is known that $V_{n,m}^k(s_0) = 0$ and $\scrV_{n,m}'^k(s_J) = 0$.

In each integral term iteration, $I_1$ and $I_2$ are computed using previous iteration values of $V$. But since $V_{n,m}^{k-1}(s)$ is available only at the mesh points for $s$, the summands required for $I_1$ and $I_2$ are linearly interpolated or extrapolated. The integrals are computed only once per integral term iteration and are not recomputed during the variance line iterations.

We now discuss what transpires in the general $k$th variance line iteration. 

At $m=0$, equation \eqref{eqn-MOLEu-IPDE-Approx2} is not solved directly. Instead, we apply the quadratic extrapolation formulas
\begin{align}
\begin{split}
\label{eqn-MOL-QuadExtrap}
V_{n,0}^k(s_j) & \approx 3V_{n,1}^{k-1}(s_j)-3V_{n,2}^{k-1}(s_j)+V_{n,3}^{k-1}(s_j)\\
\scrV_{n,0}^k(s_j) & \approx 3\scrV_{n,1}^{k-1}(s_j)-3\scrV_{n,2}^{k-1}(s_j)+\scrV_{n,3}^{k-1}(s_j)
\end{split}
\end{align}
to approximate the option price and delta at the first variance line.\footnote{Provided that the Feller condition is satisfied, the quadratic extrapolation combined with the MOL approach produces a consistent approximation of the pricing equation as $v\to 0$ \citep[Appendix]{Chiarella-2009}.} 

For $m=1,\dots,M-1$, equation \eqref{eqn-MOLEu-IPDE-Approx2} is transformed into a system of first-order ODEs
\begin{align}
\begin{split}
\label{eqn-MOLEu-FirstOrderSystem}
V_{n,m}'^k(s) & = \scrV_{n,m}^k(s)\\
\scrV_{n,m}'^k(s) & = C(\tau_n,s,v_m) V_{n,m}^k(s) + D(\tau_n,s,v_m) \scrV_{n,m}^k(s) + g(\tau_n,s,v_m),
\end{split}
\end{align}
where $C=c/a$, $D=-b/a$, and $g = (F+I_1+I_2)/a$. The corresponding Riccati transformation is given by
\begin{equation}
\label{eqn-MOLEu-RiccatiTransform}
V_{n,m}^k(s) = R(s)\scrV_{n,m}^k(s)+w(s),
\end{equation}
where $R$ and $w$ are solutions of the initial value problems
\begin{align}
\begin{split}
\label{eqn-MOLEu-RiccatiSystem}
R'(s) & = 1-D(\tau_n,s,v_m)R(s)-C(\tau_n,s,v_m)R^2(s)\\
w'(s) & = -C(\tau_n,s,v_m)R(s)w(s) - g(\tau_n,s,v_m)R(s),
\end{split}
\end{align}
with $R(0)=w(0)=0$. System \eqref{eqn-MOLEu-RiccatiSystem} is integrated using the trapezoidal rule starting at $s_0=0$ until $s_J$. This process is referred to as the \emph{forward sweep}. Once $R$ and $w$ have been determined, the equation
\begin{align*}
\scrV_{n,m}'^k(s) 
	& = \left[C(\tau_n,s,v_m)R(s)+D(\tau_n,s,v_m)\right]\scrV_{n,m}^k(s)\\
	& \qquad + C(\tau_n,s,v_m)w(s) + g(\tau_n,s,v_m)
\end{align*}
is integrated using the trapezoidal rule starting at $s=s_J$. The starting point $\scrV_{n,m}^k(s_J)$ is calculated using the boundary condition $\scrV_{n,m}'^k(s_J) = 0$ and the above equation. The option price $V_{n,m}^k(s)$ can be recovered using the Riccati transform \eqref{eqn-MOLEu-RiccatiTransform}. The \emph{reverse sweep} concludes after calculating $V_{n,m}^k(s)$ and $\scrV_{n,m}^k(s)$ for all $s$.

The process at the last variance line $m=M$ is similar to that for $1\leq m\leq M-1$, except that the coefficients $C$, $D$, and $g$ must be adjusted in view of equations \eqref{eqn-MOLEu-IPDE-Approx2-Coeff-vBC}.

The variance line iterations terminate once the convergence criterion
\begin{equation}
\label{eqn-ConvergenceCriterion}
\max_{0\leq m\leq M}\left\{\max_{0\leq j\leq J}\left|V_{n,m}^k(s_j) - V_{n,m}^{k-1}(s_j)\right|\right\}<10^{-8}
\end{equation} 
is satisfied. Otherwise, the current iterates are stored as previous iterates and the next variance line iteration commences. If \eqref{eqn-ConvergenceCriterion} is met, then $V_{n,m}^k(s)$ and $\scrV_{n,m}^k(s)$ are stored as $V_{n,m}^{k'}(s)$ and $\scrV_{n,m}^{k'}(s)$ and the current integral term iteration continues. The convergence criterion may also include the option delta, if a more stringent criterion is desired.

Now back at the current integral term iteration, criterion \eqref{eqn-ConvergenceCriterion} is checked again using $V_{n,m}^{k'}(s)$ and $\scrV_{n,m}^{k'}(s)$. If it is not satisfied, then the current iterates are stored as previous iterates and the next integral term iteration commences. At this point, $I_1$ and $I_2$ are recalculated and another pass of variance line iterations commences. Otherwise, the integral term iterations terminate and the latest iterates $V_{n,m}^{k'}(s)$ and $\scrV_{n,m}^{k'}(s)$ are taken to be the solutions $V_{n,m}(s)$ and $\scrV_{n,m}(s)$ at the $n$th time step. The algorithm then proceeds to the next time step.

As seen above, the MOL algorithm naturally computes the option price $V_{n,m}(s)$, the option delta $\scrV_{n,m}(s)$, and the option gamma $\scrV_{n,m}'(s)$. The option gamma profile can be stored by computing $\scrV_{n,m}'^k(s_j)$ with the second equation in \eqref{eqn-MOLEu-FirstOrderSystem} after computing $V_{n,m}^k(s_j)$ and $\scrV_{n,m}^k(s_j)$ in the reverse sweep.

The MOL algorithm discussed above can also be used to approximate the joint transition density function of the log-asset yield ratio $x=\ln s$ and the variance $v$. The joint tdf $H(\tau,x,v;0,x_0,v_0) = \hat{\Q}(\tilde{S}(0)=e^{x_0}, v(0)=v_0|\tilde{S}(\tau) = e^x, v(\tau) = v)$ is known to satisfy the backward Kolmogorov equation
\begin{align}
\begin{split}
\label{eqn-MOLTDE-IPDE}
\pder[H]{\tau} 
	& = \frac{1}{2}\sigma^2 v\pder[^2H]{x^2} + \frac{1}{2}\omega^2 v\pder[^2H]{v^2} + \omega(\sigma_1\rho_1-\sigma_2\rho_2) v\pder[^2H]{x\partial v}\\
	& \qquad - \left(\tilde{\lambda}_1\tilde{\kappa}_1+\tilde{\lambda}_2\tilde{\kappa}_2^- +\frac{1}{2}\sigma^2 v\right)\pder[H]{x} + \left[\xi\eta-(\xi+\Lambda)v\right]\pder[H]{v} - (\tilde{\lambda}_1+\tilde{\lambda}_2)H\\
	& \qquad + \tilde{\lambda}_1\int_\mathbb{R} H(\tau,x+y,v)G_1(y)\dif y + \int_{\mathbb{R}}H(\tau,x-y,v)G_2(y)\dif y,
\end{split}
\end{align}
for $-\infty<x<\infty$, $0<v<\infty$, and $0\leq\tau\leq T$. Equation \eqref{eqn-MOLTDE-IPDE} is to be solved subject to the initial condition $H(0,x,v;0,x_0,v_0)=\delta(x-x_0)\delta(v-v_0)$, where $\delta(\cdot)$ is the Dirac delta function, and boundary conditions $\lim_{x\to\pm\infty}H(\tau,x,v) = 0$, $\lim_{v\to\infty}H(\tau,x,v)=0$, and $\lim_{v\to 0^+}H(\tau,x,v)=0$. It is assumed here that the values of the state variables at maturity are constant.

\subsection{Adjustments for the American Exchange Option Price}
\label{sec-MOL-AmExcOp}

The MOL algorithm discussed above has to be slightly adjusted to determine the unknown early exercise boundary during the forward sweep. For convenience, we continue to denote the discounted American exchange option price by $V(\tau,s,v)$. Then $V$ also satisfies \eqref{eqn-MOLEu-IPDE}, subject to the same initial condition. However, this equation must be solved over the restricted domain $0<v<\infty$, $0\leq\tau\leq T$, and $0<s<A(\tau,v) \equiv B(\tau,v)e^{(q_1-q_2)(T-\tau)}$, where $B(\tau,v)$ us the unknown early exercise boundary. For $s\geq A(\tau,v)$, the option price is known analytically as $V(\tau,s,v) = e^{-q_1(T-\tau)}(s-e^{(q_1-q_2)(T-\tau)})$. In addition to the usual boundary conditions $\lim_{s\to 0^+}V(\tau,s,v) = 0$ and $\lim_{v\to\infty}\partial V/\partial v = 0,$ we also impose smooth-pasting condition $\lim_{s\to A(\tau,v)}\partial V/\partial s = e^{-q_1(T-t)}$. A suitable initial condition for the early exercise boundary is $A(0^+,v) = B(0^+,v)e^{(q_1-q_2)T}$, where $B(0^+,v)$ is a solution of equation \eqref{eqn-PutCall-EEBLimit}, and let $A_{n,m}=A(\tau_n,v_m)$ be the approximation of the boundary at $\tau=\tau_n$ and $v=v_m$.

The MOL approximation of the American exchange option at $\tau=\tau_n$ and $v=v_m$ is also given by equation \eqref{eqn-MOLEu-IPDE-Approx2} with coefficients \eqref{eqn-MOLEu-IPDE-Approx2-Coeff} for $m=1,\dots,M-1$ and \eqref{eqn-MOLEu-IPDE-Approx2-Coeff-vBC} for $m=M$. Thus, in this section, we highlight how the Riccati transform method in the previous section is modified to solve for the early exercise boundary. We assume that $s_J$ is chosen such that $A_{n,m}<s_J$ for all $n,m$. Algorithm \ref{pseudo-MOL-AmExcOp} shows the modified MOL algorithm.

\begin{algorithm}
\scriptsize
\caption{MOL for the discounted American exchange option price}
\label{pseudo-MOL-AmExcOp}
\SetAlgoLined
\KwResult{Option price $V$ and delta $\scrV$}
Inputs: SVJD model parameters, jump distribution parameters, partitions for $[0,T]$, $[0,s_J]$, and $[0,v_M]$\;
Compute Gauss-Hermite parameters $\varrho_l^H$ and $z_l^H$, $l=1,\dots,L$ and compute $\kappa$ using Gauss-Hermite quadrature\;
Compute $A_{0,m} = B(0^+,v_m)e^{(q_1-q_2)T}$ using \eqref{eqn-PutCall-EEBLimit}\;
Initialize $V_{0,m}(s_j) = e^{-q_1 T}(s_j-e^{(q_1-q_2)T})^+$, $\scrV_{0,m}(s_j) = e^{-q_1 T}\text{Heaviside}(s_j-e^{(q_1-q_2)T})$ for all $j,m$\;
\For{$n=1$ \KwTo $N$}{
	Set $V_{n,m}^{k'-1}(s_j) = V_{n-1,m}(s_j)$ and $\scrV_{n,m}^{k'-1}(s_j) = \scrV_{n-1,m}(s_j)$\;
	Initialize $\epsilon_\text{integ} = 1$ and $k'=1$\;
	\tcc{Commence integral term iterations}
	\While{$\epsilon_\text{integ}> 10^{-8}$ \& $k'<k'_{max}$}{
		Compute $I_1(\tau_n,s_j,v_m)$, $I_2(\tau_n,s_j,v_m)$ using $V_{n,m}^{k-1}(s_j)$ for all $j,m$\;
		Set $V_{n,m}^k(s_j) = V_{n,m}^{k'}(s_j)$ and $\scrV_{n,m}^k(s_j) = \scrV_{n,m}^{k'}(s_j)$\;
		Initialize $\epsilon_\text{var} = 1$, $k=1$, and set $k_{max} = 50$\;
		\tcc{Commence variance line iterations}
		\While{$\epsilon_\text{var}>10^{-8}$ \& $k<k_{max}$}{
			\For{$m=0$ \KwTo $M$}{
				\If{$m=1$}{
					Compute $V_{n,0}^k$ and $\scrV_{n,0}^k$ using equations \eqref{eqn-MOL-QuadExtrap}\;
					}
				\ElseIf{$1\leq m\leq M-1$}{
					\tcc{Forward sweep and computation of $A_{n,m}^k$}
					Initialize forward sweep while-loop parameters: $i=1$, $\phi(s_0)=e^{-q_2(T-\tau_n}$, $\Pi = 1$\;
					\While{$\Pi>0$}{
						Compute $R(s_i)$, $w(s_i)$, and $\phi(s_i)$\;
						Update $\Pi = \phi(s_i)\phi(s_{i-1})$ and assign $j^*=i$\;
						Update $i=i+1$\;
					}
					Compute $\phi(s_{j^*+1})$ and fit a cubic spline through $\{(s_i,\phi(s_i))\}_{i=j^*-2}^{j^*+1}$\;
					Take $A_{n,m}^k$ as the zero of the cubic spline in $(s_{j^*-1},s_{j^*})$\;
					\tcc{Reverse sweep over $[s_{j^*-1},A_{n,m}^k]$}
					Set $\scrV(A_{n,m}^k) = e^{-q_1(T-\tau_n)}$ and linearly interpolate values of $C$, $D$, $g$, $R$, and $w$ at $s=A_{n,m}^k$\;
					Calculate $\scrV_{n,m}^k(s_{j^*-1})$ using the trapezoidal rule and $V_{n,m}^k(s_{j^*-1})$ using \eqref{eqn-MOLEu-RiccatiTransform}\;
					\tcc{Reverse sweep over $[0,s_{j^*-1}]$}
					\For{$j=j^*-2$ \KwTo $1$}{
						Compute $\scrV_{n,m}^k(s_j)$ using the trapezoidal rule\;
						Compute $V_{n,m}^k(s_j)$ using the Riccati transform \eqref{eqn-MOLEu-RiccatiTransform}\;
					}
					Calculate $V_{n,m}^k(s_j)$ and $\scrV_{n,m}^k(s_j)$ using the analytical form in the stopping region, $j=j^*,\dots,J$\;
				}
				\Else{ 
					Recalculate the coefficients of \eqref{eqn-MOLEu-RiccatiSystem} with \eqref{eqn-MOLEu-IPDE-Approx2-Coeff-vBC}\tcp{At $v=v_M$}
					Perform the forward sweep and reverse sweep as in when $1\leq m\leq M-1$\;
				}
			}
			Update $\epsilon_\text{var} = \max_m\{\max_j\{|V_{n,M}^k(s_j)-V_{n,m}^{k-1}(s_j)|\}$\;
			Update $V_{n,m}^{k-1} = V_{n,m}^{k}$ and $\scrV_{n,m}^{k-1} = \scrV_{n,m}^k$\;
			Update $k=k+1$\;
		}
		Update $V_{n,m}^{k'} = V_{n,m}^k$ and $\scrV_{n,m}^{k'} = \scrV_{n,m}^k$\;
		Update $\epsilon_\text{integ} = \max_m\{\max_j\{|V_{n,M}^{k'}(s_j)-V_{n,m}^{k'-1}(s_j)|\}$\;
		Update $V_{n,m}^{k'-1} = V_{n,m}^{k'}$ and $\scrV_{n,m}^{k'-1} = \scrV_{n,m}^{k'}$\;
		Update $k'=k'+1$
	}
	Update $V_{n,m} = V_{n,m}^{k'}$ and $\scrV_{n,m} = \scrV_{n,m}^{k'}$\;
}
\end{algorithm}

In each of the $k$ variance line iterations, the $V_{n,0}^k(s)$ and $\scrV_{n,0}^k(s)$ are estimated using the quadratic extrapolation formulas \eqref{eqn-MOL-QuadExtrap}. For each $m=1,\dots,M-1$, the forward sweep of system \eqref{eqn-MOLEu-RiccatiSystem} via the trapezoidal rule starts at $s_0$ in the direction of increasing $s$. The main difference is that for the American exchange option, the sign of the function
\begin{equation}
\label{eqn-MOLAm-SignCheck}
\phi(s_j) = R(s_j) e^{-q_1(T-\tau_n)}+w(s_j) - e^{-q_1(T-\tau_n)}\left(s_j-e^{(q_1-q_2)(T-\tau_n)}\right)
\end{equation}
is monitored at each $j$. Equation \eqref{eqn-MOLAm-SignCheck} arises from combining the smooth-pasting condition and the Riccati transform equation. The initial value is $\phi(s_0) = e^{-q_2(T-\tau_n)}>0$. The forward sweep stops at the index $j^*$ at which $\phi$ first becomes negative; i.e. $j^*$ is the first index such that $\phi(s_{j^*-1})\phi(s_{j^*})<0$. The approximation of the early exercise boundary at this iteration, which we denote by $A^k_{n,m}$, is estimated to be the zero of the cubic spline interpolant through the points $\{(s_i,\phi(s_i))\}_{i=j^*-2}^{j^*+1}$ that occurs in between $s_{j^*-1}$ and $s_{j^*}$. This zero is numerically solved using any standard root finding algorithm.

Once $A^k_{n,m}$ has been determined, the reverse sweep for $\scrV^k_{n,m}(s)$ is then carried out over the interval $[0,A^k_{n,m}]$. Since $A^k_{n,m}$ is not part of the regular mesh, the trapezoidal rule is first applied over $[s_{j^*-1},A^k_{n,m}]$, where the required values of $C$, $D$, $g$, $R$, and $w$ are linearly interpolated from values at $s_{j^*-1}$ and $s_{j^*}$, which were determined as part of the forward sweep. Once $\scrV^k_{n,m}(s_{j^*-1})$ has been calculated, the reverse sweep can continue over $[0,s_{j^*-1}]$ along the regular mesh. For $j\geq j^*$, it is known that $\scrV^k_{n,m}(s_j) = e^{-q_1(T-\tau_n)}$. The value of $V^k_{n,m}(s_j)$ can be calculated using the Riccati transform equation.

The same process applies at the last variance line, except that a recalculation of the coefficients of the Riccati system according to equation \eqref{eqn-MOLEu-IPDE-Approx2-Coeff-vBC} is required before commencing the forward sweep.

Once the convergence criterion is satisfied at both iteration levels, the final iterates are then stored as the solution at the $n$th time step.\footnote{The convergence criterion can also include the early exercise boundary.} The algorithm then proceeds to the next time step.

\subsection{Venttsel Boundary Conditions at $v=v_M$}
\label{sec-MOL-Venttsel}

The choice of boundary conditions for equation \eqref{eqn-MOLEu-IPDE} affects the quality of the MOL approximation obtained using the algorithms discussed above. In most financial applications, the initial condition is usually dictated by the type of instrument being priced. However, articulating boundary conditions in the spatial variables is not as straightforward since the equation may be degenerate at certain points of the boundary or the domain may be unbounded. The problem considered in this paper struggles with both issues. 

The pricing equation \eqref{eqn-MOLEu-IPDE} is degenerate when $s=0$ or $v=0$. As such, the structure of the IPDE determines whether conditions at these boundaries should be independently imposed or if the equation itself should naturally hold at these boundaries. \citet{Chiarella-2009} consider this issue in more detail, with the aid of the Fichera function for degenerate equations, for the closely related American call option under SVJD dynamics and justify the use of quadratic extrapolation for the solution of the pricing equation along the first variance line $v=v_0$.\footnote{Such an analysis was also done by \citet{Kang-2014} for American calls under stochastic volatility and stochastic interest rates. \citet{Meyer-2015} provides a discussion of the Fichera theory applied to common financial problems.} They showed that equation \eqref{eqn-MOLEu-IPDE} is expected to hold when $v=0$ provided that the Feller condition is satisfied. This implies that no independent condition is required at that point. On the other hand, we allow the dynamics of $\tilde{S}(t)$ to dictate the boundary condition at $s=0$. From \eqref{eqn-PutCall-YieldRatioSDE-Q}, if the process starts at zero, then it will stay at zero, implying that the option becomes worthless throughout its life. This motivates the choice of the first boundary condition in \eqref{eqn-MOLEu-BC}.

For the MOL algorithm, we truncated the infinite $s$ and $v$ domains. As such, boundaries must be prescribed at the far boundaries $s=s_J$ and $v=v_M$. For the Riccati transform method to work in the European case, the boundary condition at $s_J$ must be equivalent to a scalar equation in $V$ and $\scrV$, and this is attained through the second condition in \eqref{eqn-MOLEu-BC} when approximated as $\scrV_{n,m}'(s_J) = 0$ (although there may be alternatives to this). In the American case, no such condition is required since we already have the smooth-pasting and value-matching conditions at the free boundary, which is assumed to be less than $s_J$. Hence, we only need to consider boundary conditions at $v_M$.

As an alternative to the boundary condition $\partial V(\tau,s,v_M)/\partial v = 0$, we consider Venttsel boundary conditions derived from the pricing equation.\footnote{We do not seek to formally define Venttsel boundary conditions in this paper. The reader is referred to \citet[Section 1.2.3]{Meyer-2015} for a discussion of Venttsel boundary conditions in the context of financial pricing problems.} Following the theory of Venttsel boundary conditions in \citet{Meyer-2015}, we find that the equation
\begin{align}
\begin{split}
\label{eqn-Fichera-VBC1}
0 & =\frac{1}{2}\sigma^2 v_M s^2\pder[^2V]{s^2}-\left(\tilde{\lambda}_1\tilde{\kappa}_1+\tilde{\lambda}_2\tilde{\kappa}_2^-\right)s\pder[V]{s}+\left[\xi\eta-(\xi+\Lambda)v_M\right]\pder[V]{v}\\
	& \qquad - (\tilde{\lambda}_1+\tilde{\lambda}_2)V - \pder[V]{\tau} - I_1(\tau,s,v_M) - I_2(\tau,s,v_M)
\end{split}
\end{align}
is an admissible Venttsel condition at $v=v_M$ provided that $(\xi+\Lambda)v_M\geq \xi\eta$. Equation \eqref{eqn-Fichera-VBC1} is equivalent to the assumption that the option delta $\partial V/\partial s$ and vega $\partial V/\partial v$ are insensitive to the instantaneous variance for when the variance level is sufficiently large. Another Venttsel condition is given by
\begin{align}
\begin{split}
\label{eqn-Fichera-VBC2}
0 & = \frac{1}{2}\sigma^2 v_M s^2\pder[^2V]{s^2}-\left(\tilde{\lambda}_1\tilde{\kappa}_1+\tilde{\lambda}_2\tilde{\kappa}_2^-\right)s\pder[V]{s}-(\tilde{\lambda}_1+\tilde{\lambda}_2)V - \pder[V]{\tau}\\
	& \qquad - I_1(\tau,s,v_M) - I_2(\tau,s,v_M),
\end{split}
\end{align}
which resembles the pricing equation when volatility is constant at $\sigma\sqrt{v_M}$. Whichever Venttsel boundary condition is chosen, the equation must be solved subject to boundary conditions in $\tau$ and $s$.

In the MOL implementation, the adoption of either Venttsel condition implies a change in the coefficients of equation \eqref{eqn-MOLEu-IPDE-Approx2} at $v=v_M$. For equation \eqref{eqn-Fichera-VBC1}, we employ a backward difference approximation of $\partial V/\partial v$. This is appropriate in view of the upwinding difference approximation for the other variance lines as this Venttsel condition requires that $(\xi+\Lambda)v_M\geq \xi\eta$. The usual discretization is employed for $\partial V/\partial \tau$. Once the coefficients of equation \eqref{eqn-MOLEu-IPDE-Approx2} have been updated, the Riccati solution for the last variance line can proceed as outlined in Section \ref{sec-MOL-Algorithm}.

The impact of these Venttsel conditions on the performance of the MOL algorithm is explored in the succeeding sections.

\section{Numerical Results and Discussion}
\label{sec-NumericalResults}

\begin{table}[]
\caption{Parameter values used for the method of lines (MOL) approximation. For the solution of the transition density function, a value of $x_{\max}=5$ is used for the partition of the computational domain in $x$. For the approximation of the European and American exchange option prices, a value of $s_J=4$ is used for the mesh for $s$. All MOL implementations share all other parameter values.}
\label{tab-ParameterValues}
\centering
\begin{tabular}{@{}cccccccc@{}}
\toprule
\multicolumn{2}{c}{\textbf{Asset Price}} & \multicolumn{2}{c}{\textbf{Stoch. Vol.}} & \multicolumn{2}{c}{\textbf{Jumps}} & \multicolumn{2}{c}{\textbf{Mesh Sizes}} \\ \midrule
$T$ & 0.50 & $\xi$ & 2.00 & $\tilde{\lambda}_1$ & 5.00 & $x_{\max}$ & 5.00 \\
$q_1$ & 0.05 & $\eta$ & 0.56 & $\beta_{j_1}$ & 0.20 & $s_J$ & 4.00 \\
$q_2$ & 0.03 & $\Lambda$ & 0.00 & $\alpha_{j_1}$ & $0.00$ & $v_M$ & 2.00 \\
$\sigma_1$ & 0.50 & $\omega$ & 0.40 & $\tilde{\lambda}_2$ & 2.00 & $N$ & 30 \\
$\sigma_2$ & 0.50 & $\rho_1$ & 0.50 & $\beta_{j_2}$ & 0.20 & $M$ & 25 \\
$\rho_w$ & 0.50 & $\rho_2$ & 0.05 & $\alpha_{j_2}$ & 0.00 & $J$ & 140 \\
 &  &  &  & $L$ & 20 &  &  \\ \bottomrule
\end{tabular}
\end{table}

In this section, we present numerical approximations of the discounted European and American exchange option prices, the early exercise boundary, the discounted early exercise premium, and the joint tdf generated using the MOL algorithm discussed in the previous section. These results were obtained using the parameter values enumerated in Table \ref{tab-ParameterValues}, although in Section \ref{sec-CompStat}, we investigate how the option prices and the early exercise boundary change in response to alternative model parameter values.\footnote{Model parameter values used in this paper are similar to those used by \citet{Chiarella-2009}, although additional parameters are used to accommodate the second asset.} The values of $q_1$ and $q_2$ were chosen so that the exercise boundary is discontinuous at maturity, thereby highlighting the impact of jumps in asset prices. The chosen means of the jump size variables $Y_1$ and $Y_2$ are zero, indicating that upward and downward jumps in asset prices are equally likely to occur. We also assume that they have equal standard deviations.\footnote{Less emphasis on the jump size distribution concentrates the analysis on the effect of the jump intensities. Alternative jump size distributions may be considered as well, although the form of the density function dictates the quadrature formula for approximating the integral terms.} Correlations among Wiener processes are initially assumed to be positive, although the effect of negative correlations is also be explored in Section \ref{sec-CompStat}.

All source codes were implemented using MATLAB$^\text{\textregistered}$ on an x64-based personal computer with an Intel$^\text{\textregistered}$ Core\texttrademark i7-10710U CPU with 1.10GHz, 1608 MHz, 6 cores, 12 logical processors, and 8GB RAM. Due to hardware constraints, modest mesh sizes for $\tau$, $v$, and $s$ (or $x$) are used, although as seen in Table \ref{tab-LSMonteCarlo} below, the MOL does not require too many time steps to converge. The computational domain $[0,4]$ for $s$ was divided into three intervals, $[0,0.5]$, $[0.5,3]$, and $[3,4]$, which was then subdivided using 20, 80, and 40 mesh points, respectively. This allows for a more precise approximation of the early exercise boundary as it tends to occur near $e^{(q_1-q_2)T}$. Despite the use of relatively few mesh points, convergence was attained for all MOL implementations with parameters reported in Table \ref{tab-ParameterValues} and the alternative values explored in Section \ref{sec-CompStat}. 

All option prices, option deltas, and price differences presented in this section are in their \emph{discounted form}, i.e. they are expressed in units of the second asset yield. As such, prices and price differences may be more pronounced if the exchange option is written on highly priced underlying assets.  

\subsection{MOL Approximation of Exchange Option Prices and the Joint TDF}
\label{sec-NumericalPrices}

\begin{figure}
\centering
\includegraphics[width = 0.4\linewidth]{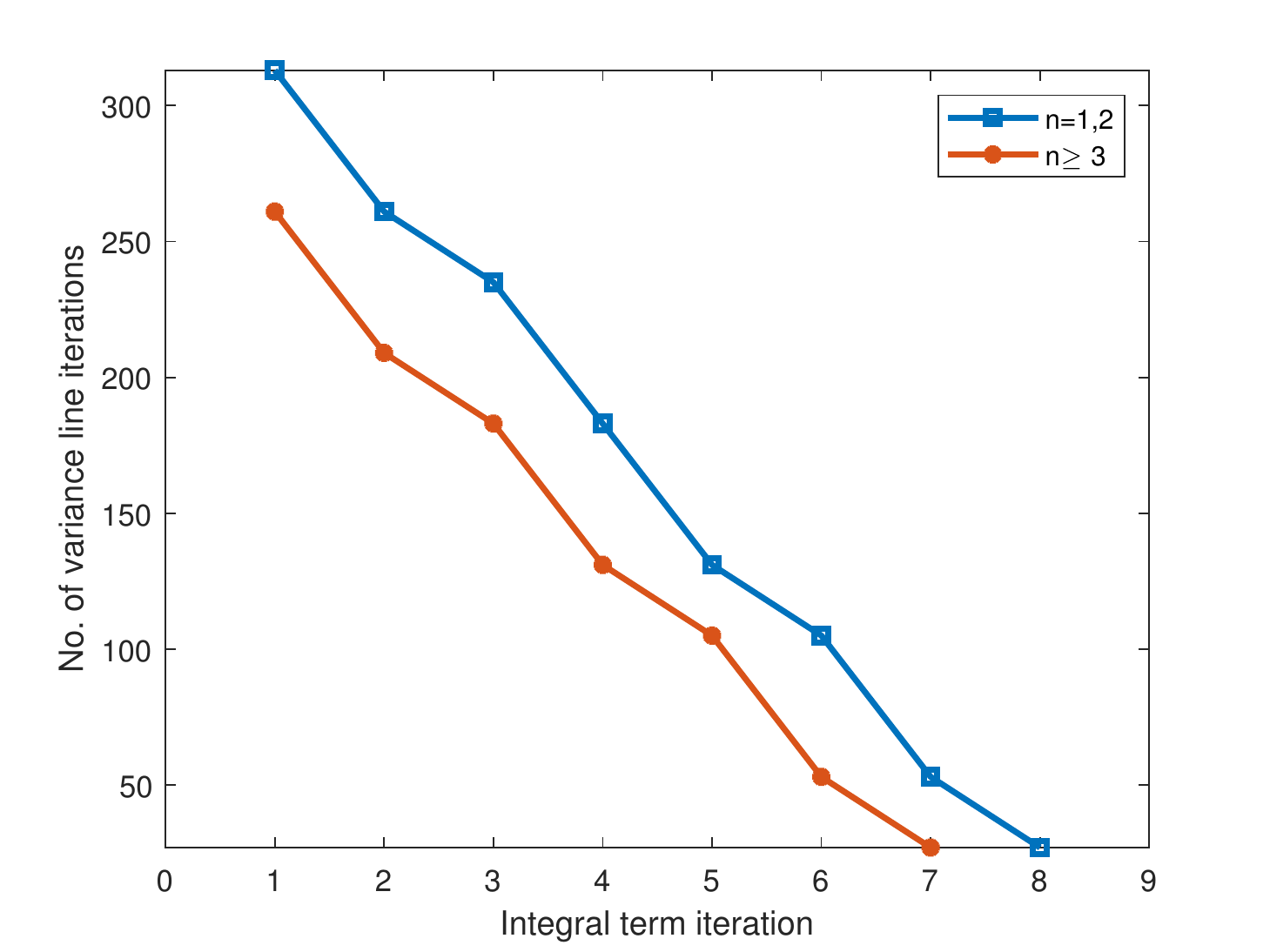}
\caption{Typical number of variance line iterations per integral term iteration at the $n$th time step for the MOL for European and American exchange options.}
\label{fig-Convergence}
\end{figure}

Figure \ref{fig-Convergence} shows the typical convergence pattern of the integral term and variance line iterations for the MOL approximation of the European and American exchange option in the $n$th time step. It has been observed that the convergence behavior is the same for both types of exchange options. However, we note a decrease in the number of integral term iterations required on and after the third time step, signifying the effect of adopting a second-order backward difference approximation for the time derivative for $n\geq 3$ on the efficiency of the computation. The number of variance line iterations generally decreases as the integral line iterations continue. \citet{Chiarella-2009} reported that the number of variance line iterations required drastically increases as the number of mesh points in $v$ increases, but we were unable to replicate this phenomenon due to hardware constraints. While the convergence behavior is the same for both European and American cases, a typical MOL pass yields a computing time for the American exchange option that is almost double that required for the European case. In a sample implementation, the European and American cases were completed in 74.4492 and 152.5094 seconds, respectively.

\begin{figure}
\centering
\subfloat[Discounted European exchange option price]{
	\includegraphics[width = 0.4\linewidth]{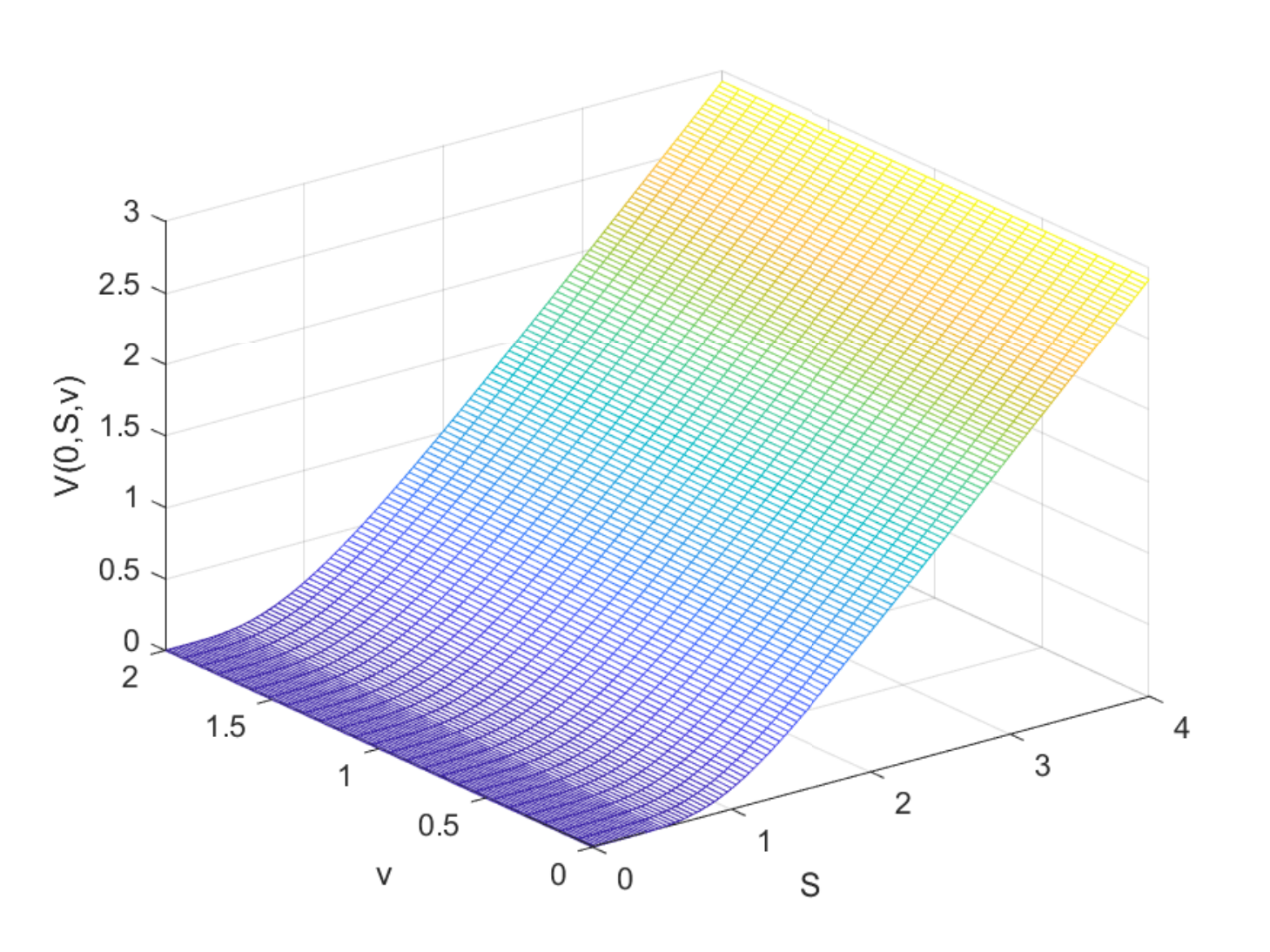}}
\subfloat[Discounted European delta]{
	\includegraphics[width = 0.4\linewidth]{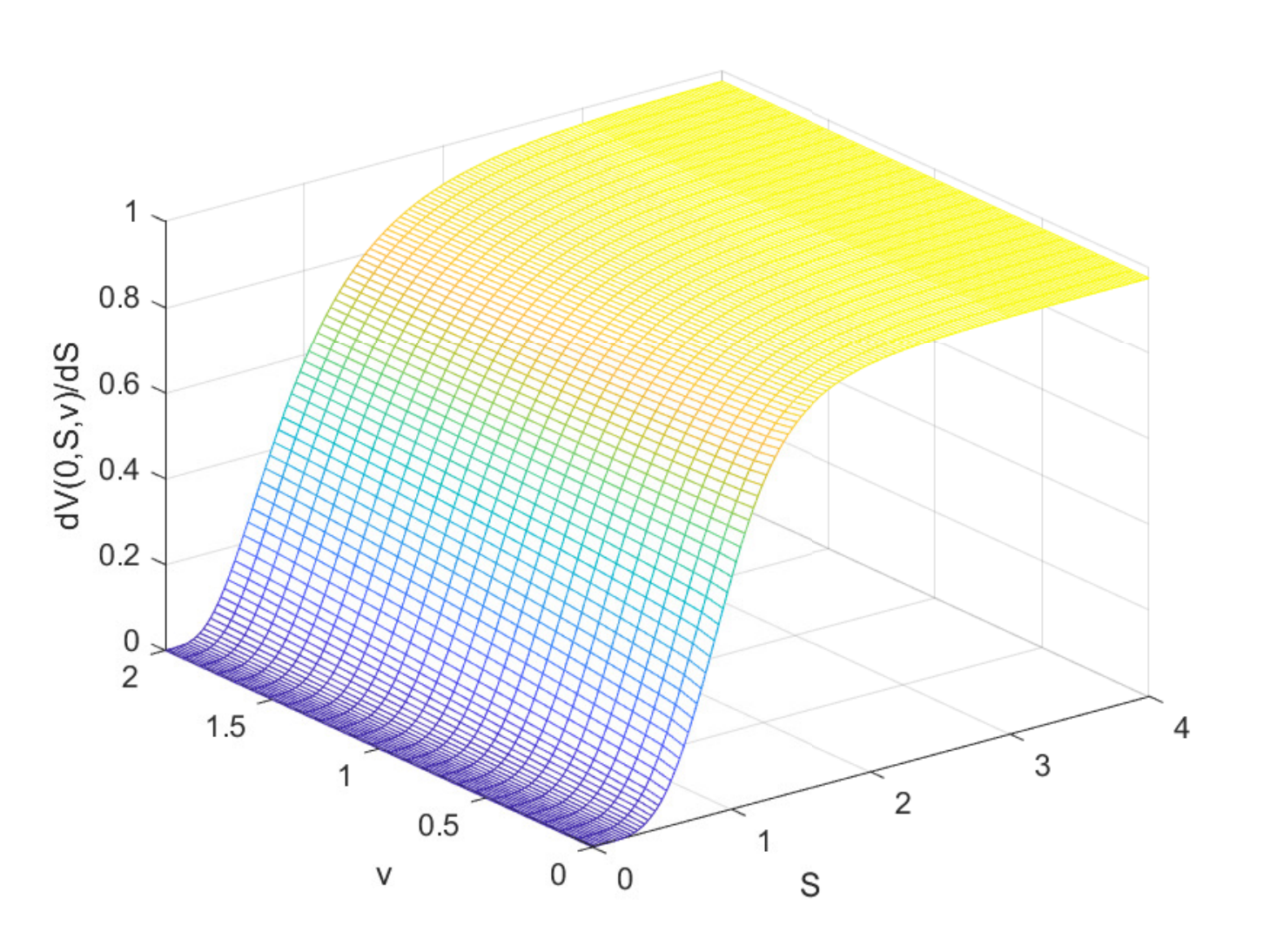}}
\caption{Method of lines approximation of the discounted European exchange option price $V(\tau,s,v)$ and delta $\partial V(\tau,s,v)/\partial s$ at $\tau=T$.}
\label{fig-EuExcOp}
\end{figure} 

Figure \ref{fig-EuExcOp} shows the approximation of the discounted European exchange option price and delta against $s$ and $v$ at $\tau=T$ (or $t=0$), the start of the life of the option. Similar to ordinary European call options, the delta of the European exchange option is steepest when the asset yield ratio is close to $e^{(q_1-q_2)T}$.

\begin{figure}
\centering
\includegraphics[width = 0.4\linewidth]{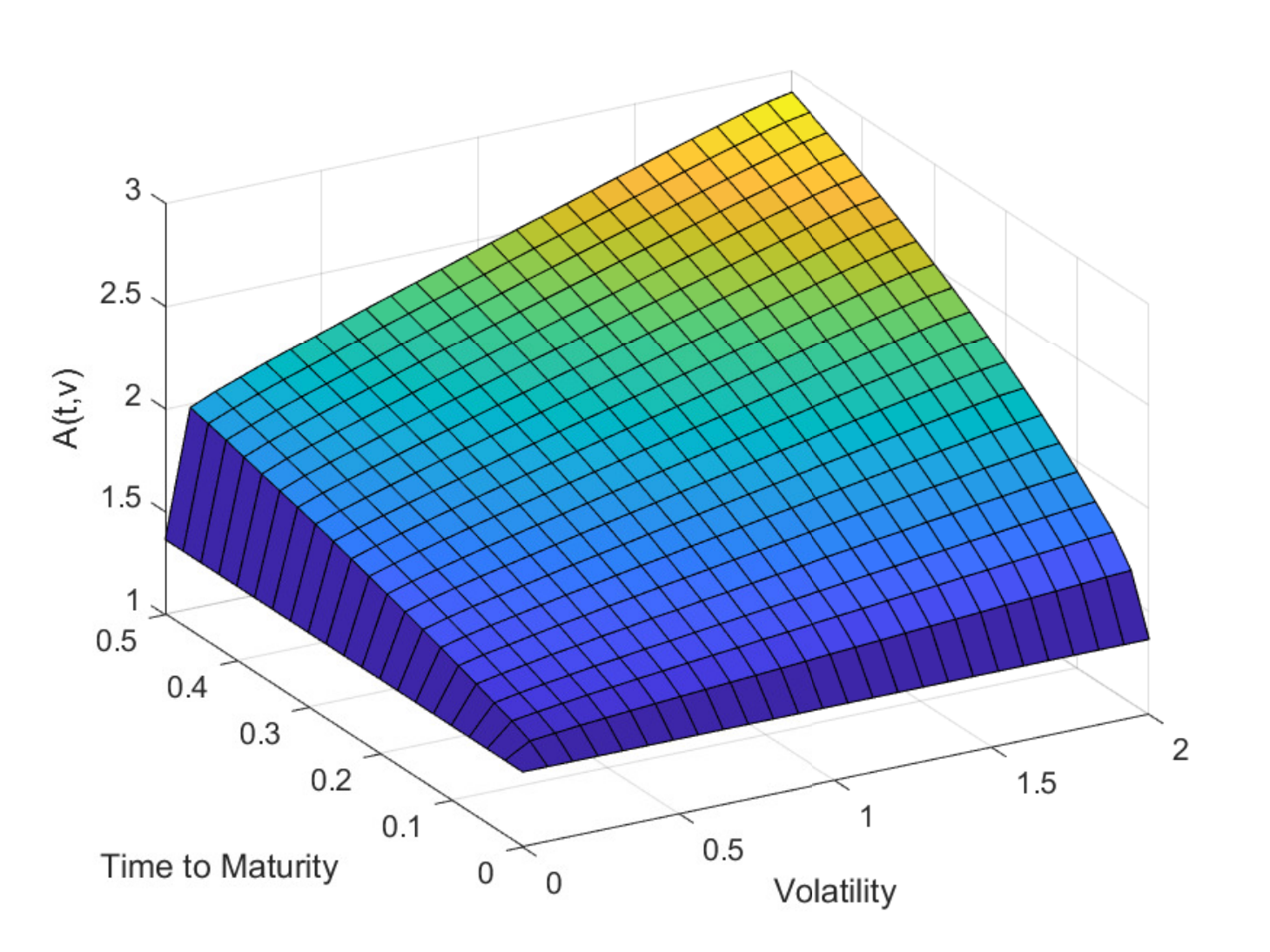}
\caption{Method of lines approximation of the early exercise boundary surface $A(\tau,v)$.}
\label{fig-AmExcOp-EEB}
\end{figure}

\begin{figure}
\centering
\subfloat[Discounted American exchange option price]{
	\includegraphics[width = 0.4\linewidth]{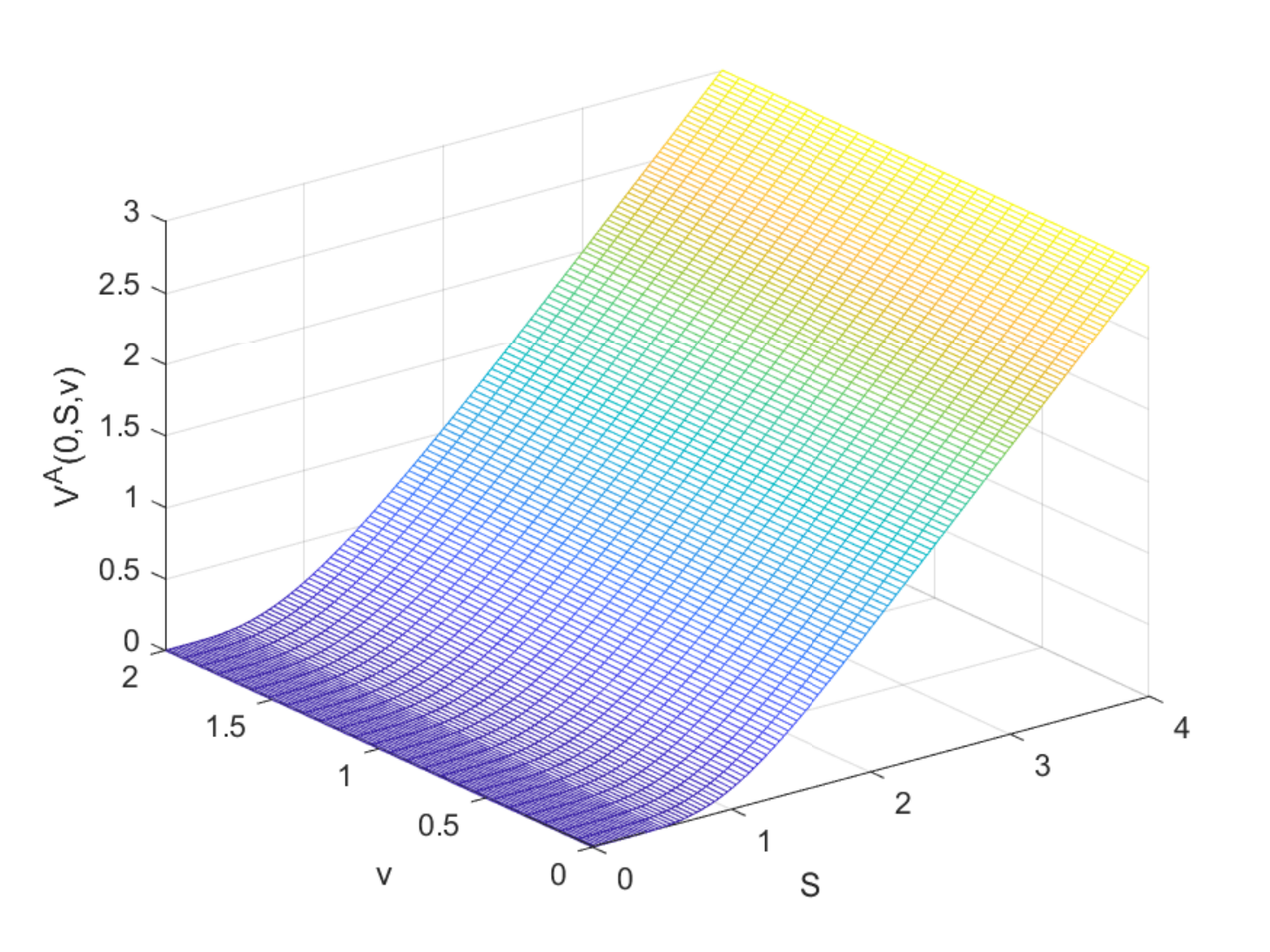}}
\subfloat[Discounted American delta]{
	\includegraphics[width = 0.4\linewidth]{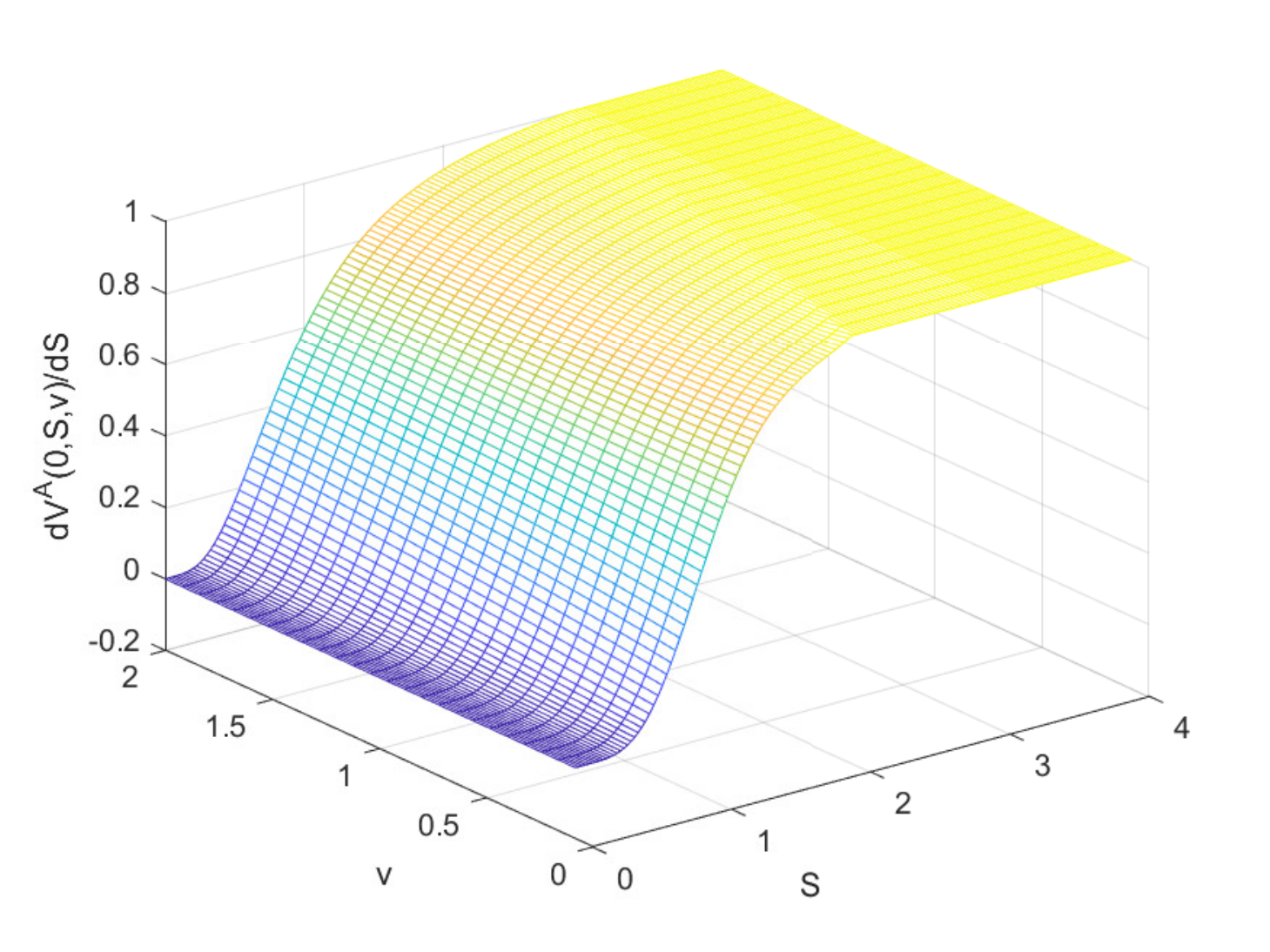}}
\caption{Method of lines approximation of the discounted American exchange option price $V^A(\tau,s,v)$ and delta $\partial V^A(\tau,s,v)/\partial s$ at $\tau=T$.}
\label{fig-AmExcOp}
\end{figure}

Figures \ref{fig-AmExcOp-EEB} and \ref{fig-AmExcOp} show the MOL approximations of the early exercise surface, the discounted American exchange option price, and the American exchange option delta at time-to-maturity $\tau = T$ (or time $t=0$). Figure \ref{fig-AmExcOp-EEB} shows that the early exercise boundary at $\tau=0$ is constant with respect to $v$ since the option payoff is independent of $v$ as discussed in Section \ref{sec-PutCall-EEBLimit}. The figure also illustrates that the $A(t,v)$ is an increasing function of $v$. As expected from the put-call transformation technique, the American and European exchange option price profiles, when expressed in terms of the asset yield ratio, behave in a similar manner to their ordinary call option counterparts.

\begin{figure}
\centering
\subfloat[Discounted early exercise premium]{
	\includegraphics[width = 0.4\linewidth]{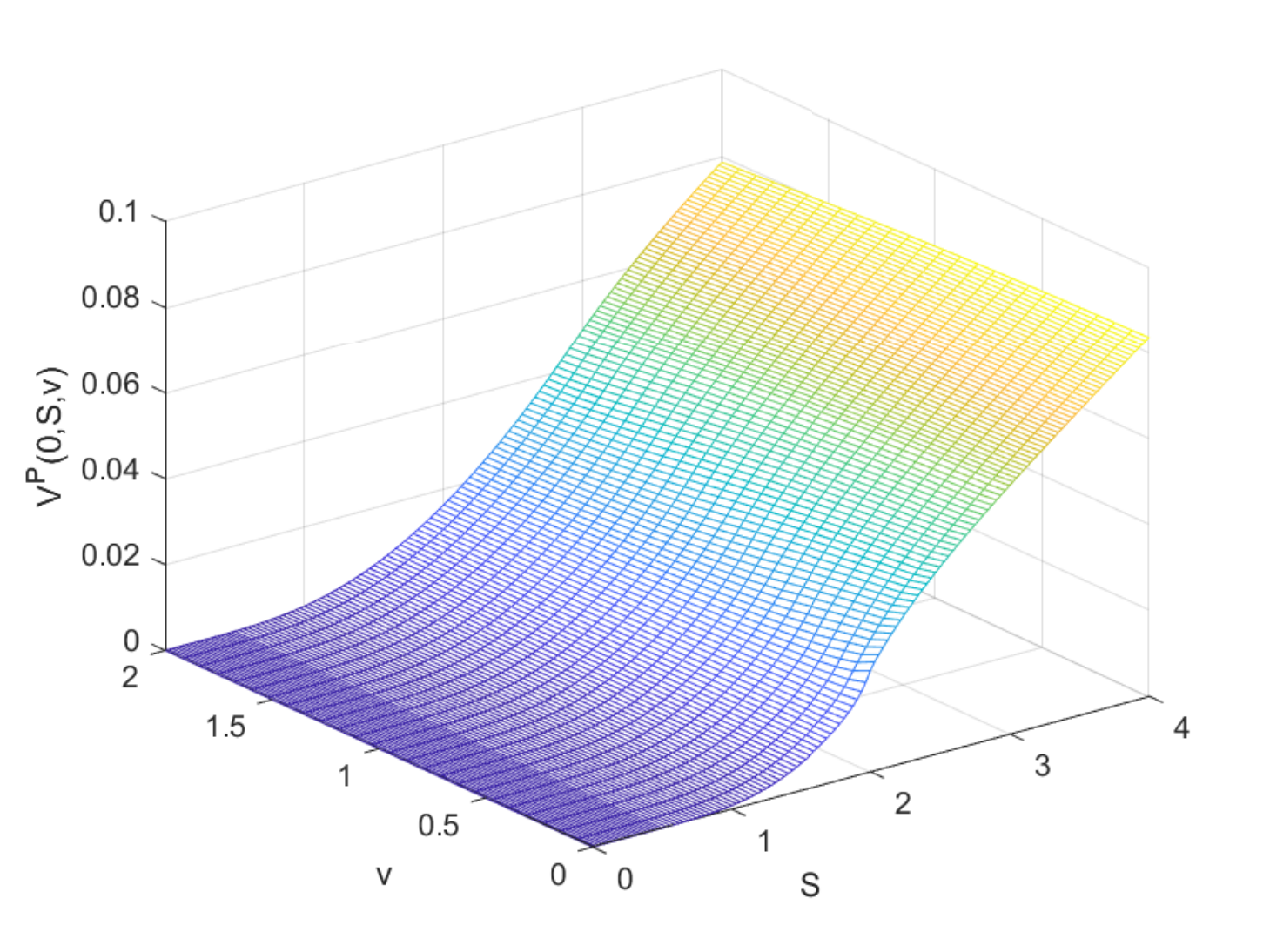}}
\subfloat[Difference between American and European deltas]{
	\includegraphics[width = 0.4\linewidth]{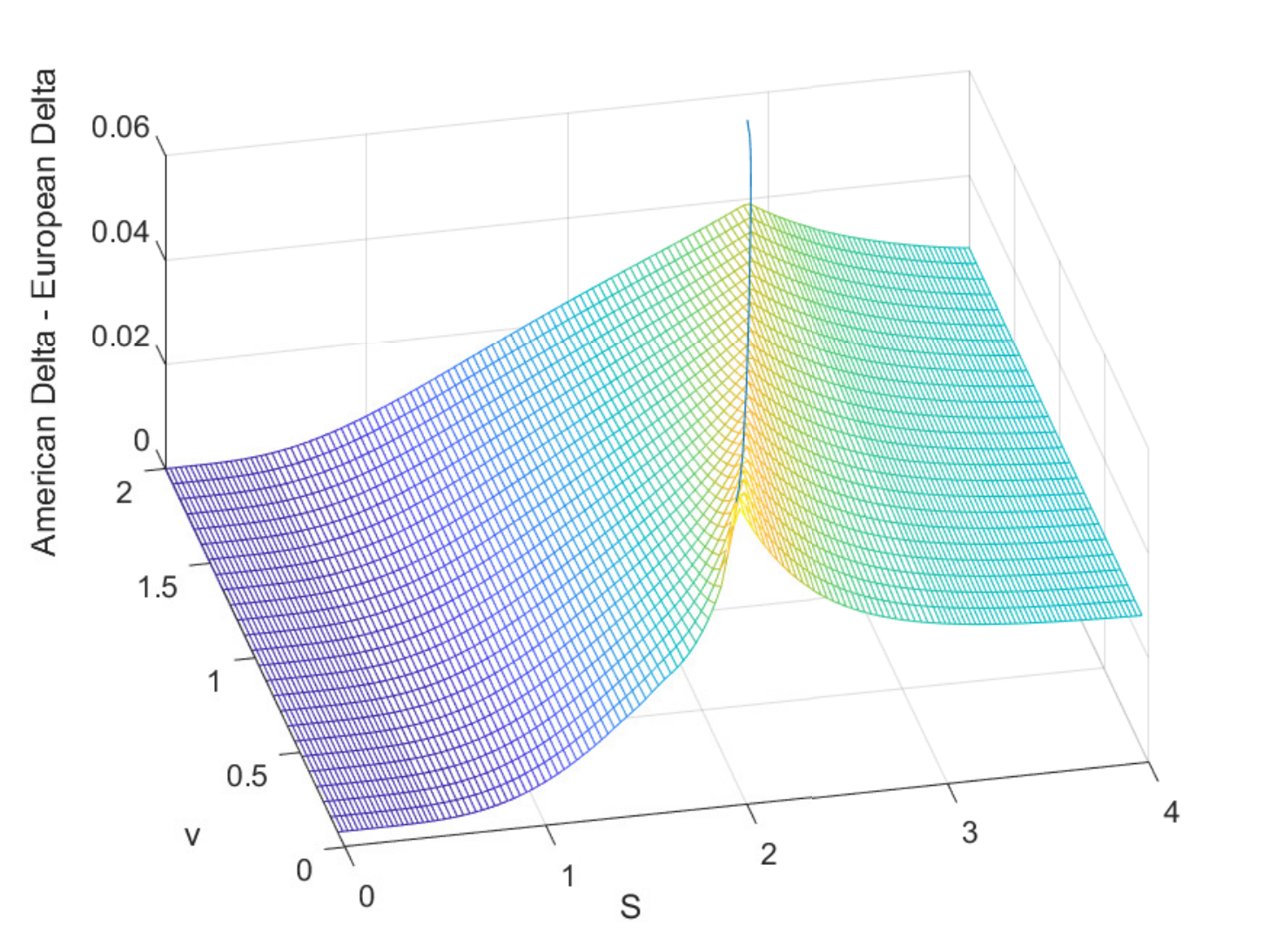}}
\caption{Method of lines approximation of the discounted early exercise premium $V^P(\tau,s,v)$ and the difference between the deltas of the American and European exchange options at time $\tau=T$.	The blue curve in (b) indicates the early exercise boundary at $\tau=T$.}
\label{fig-EEPremium}
\end{figure}

The price and delta profiles of the European and American exchange options appear to behave similarly in $s$ and $v$, but the difference is emphasized when we compute the early exercise premium $V^P$. Figure \ref{fig-EEPremium} shows the discounted early exercise premium surface in the asset yield ratio and variance at $\tau=T$. For a fixed $v$, the premium is increasing in $s$. On closer inspection, the early exercise premium has an inflection which can be confirmed to occur at the early exercise boundary $A(T,v)$. This is also where the difference in American and European deltas peaks. This implies that the premium increases the fastest as the asset yield ratio approaches the exercise boundary from the left.

\begin{table}[]
\caption{Comparison of American exchange option prices at (calendar) time $t=0$ generated by the method of lines and the least-squares Monte Carlo (LSMC) algorithm of \citet{LongstaffSchwartz-2001} for various values of the asset yield ratio and fixed spot variance $v(0)=0.56$.}
\label{tab-LSMonteCarlo}
\centering
\begin{tabular}{@{}crrrrrr@{}}
\toprule
\textbf{} & \multicolumn{3}{c}{\textbf{Method of Lines}} & \multicolumn{3}{c}{\textbf{[LS] Monte Carlo}} \\ \midrule
\textbf{$\tilde{S}(0)$} & \multicolumn{1}{c}{\textbf{$N = 20$}} & \multicolumn{1}{c}{\textbf{$N = 50$}} & \multicolumn{1}{c}{\textbf{$N = 100$}} & \multicolumn{1}{c}{\textbf{Price}} & \multicolumn{2}{c}{\textbf{95\% CI}} \\ \midrule
0.500 & 0.015447 & 0.015441 & 0.015440 & 0.000015 & $-$0.000012 & 0.000032\\
0.625 & 0.039076 & 0.039108 & 0.039112 & 0.000622 & 0.004355 & 0.008093 \\
0.750 & 0.077450 & 0.077521 & 0.077530 & 0.004084 & 0.003621 & 0.004546 \\
0.875 & 0.130851 & 0.130953 & 0.130966 & 0.018030 & 0.017121 & 0.018939 \\
1.000 & 0.198052 & 0.198179 & 0.198194 & 0.058294 & 0.056785 & 0.059803 \\
1.500 & 0.565205 & 0.565367 & 0.565382 & 0.499905 & 0.499602 & 0.500208 \\
2.000 & 1.016076 & 1.016183 & 1.016189 & 0.999884 & 0.999481 & 1.000286 \\
2.500 & 1.500852 & 1.500875 & 1.500869 & 1.500207 & 1.499520 & 1.500893 \\
3.000 & 2.000000 & 2.000000 & 2.000000 & 1.999870 & 1.999284 & 2.000457 \\
3.500 & 2.500000 & 2.500000 & 2.500000 & 2.499871 & 2.499129 & 2.500612 \\
4.000 & 3.000000 & 3.000000 & 3.000000 & 2.999690 & 2.998935 & 3.000445 \\ \midrule
$A(0,v(0))$ & 2.6622 & 2.6626 & 2.6605 & - & - & - \\
Comp Time (s) & 194.5725 & 332.1941 & 522.4031 & $5.24\times 10^5$ & - & - \\ \bottomrule
\end{tabular}
\end{table}

In Table \ref{tab-LSMonteCarlo}, we compare MOL prices for American exchange options to those generated by the \citet{LongstaffSchwartz-2001} least-squares Monte Carlo algorithm. The LSMC algorithm was implemented using $N=1000$ time steps and 10,000 scenarios (half of which are antithetic variates). For the simulation approach, first-order Euler-Maruyama discretizations of the asset yield and instantaneous variance processes were used. The LSMC algorithm requires knowledge of the second asset price process since it serves as a (stochastic) discount factor for pricing options under $\hat{\Q}$, so $S_2(t)$ was also simulated with initial value $S_2(0)=1$.\footnote{By assuming that $S_2(0)=1$, the MOL prices are then expressed in monetary units rather than in units of the second asset yield process.} 

The exercise policy generated by the LSMC is sub-optimal as it approximates the American option by its Bermudan counterpart, and so the prices produced by the algorithm are lower bounds of the ``true'' American option price \citep{LongstaffSchwartz-2001}. We are able to verify this property for the parameters used, with MOL prices being consistently higher than point estimates for the price from the LSMC algorithm. We note however that the discrepancy is larger especially for when the option is deeply out-of-the-money, but this is most likely due to the slow convergence of the Monte Carlo method, a phenomenon that was also observed in a similar analysis by \citet{ChiarellaZiveyi-2014}. For smaller values of $\tilde{S}(0)$, the regression step is also ill-conditioned, which most likely contributed to the discrepancy as well. MOL and LSMC prices for deeply out-of-the-money options are nonetheless consistent with \eqref{eqn-MOLEu-BC}. For higher values of $\tilde{S}(0)$, the MOL prices fall within the 95\% confidence interval calculated using the LSMC approach.

Also, it takes substantially longer to estimate a complete profile of American option prices using the LSMC algorithm compared to the MOL. Having to simulate sample paths for the second asset price process means the appeal of the put-call transformation in reducing the dimensionality of the problem is lost in the simulation approach. The MOL is also far more efficient than the LSMC algorithm, since in one implementation of the MOL, we are able to estimate the option price, the delta, the gamma, and the early exercise boundary. The MOL also requires substantially fewer time steps to converge. By increasing the number of basis functions used in the regression and the number of simulations decreases the gap between the LSMC point estimate and the ``true'' option price, but doing so increases the computation time \citep{Stentoft-2004}. We report however that increasing the number of time steps does not drastically diminish the discrepancy between MOL and LSMC prices for lower asset yield ratios. In particular, when $N = 2000$ and $N = 5000$ the LSMC prices when $\tilde{S}(0) = 1$ are $0.057969$ and $0.058027$, respectively, which are still far below the MOL prices.\footnote{The LSMC algorithm was also implemented for a smaller number of exercise times, $N=50$ and $N=100$, but it produced prices and confidence intervals which fall completely below the MOL price for all values of $\tilde{S}(0)$ considered. Average computation times for $N=50$ and $N=100$ time steps in the LSM algorithm are 4,582.9s and 4,949.9s, respectively.}

A complete comparison of the relative efficiency and accuracy of these two methods remains to be seen, since both methods have numerous sources of error, including (but not limited to) the discretization scheme for the simulation of state variables or the discretization of partial derivatives, truncation of infinite domains, the choice of mesh sizes and partition points in all variables, boundary conditions for the IPDE, numerical integration scheme, and the choice of basis functions for regression. Nonetheless, for the parameter values and assumptions reported for this numerical experiment, the results we obtained are reasonably comparable.

\subsection{Numerical Comparative Statics}
\label{sec-CompStat}

In this section, we investigate how the discounted American and European exchange option prices and the early exercise boundary time profile change in response to changes in the model parameters. This analysis covers three components: (1) the effect of the correlations between the Wiener processes in the asset price and variance processes, (2) the effect of asset price jump intensities, and (3) the effect of the variance process parameters, namely the mean reversion rate $\xi$, volatility of volatility $\omega$, and the market price of volatility $\Lambda$. The analysis of option price differencesat $t=0$ (or $\tau = T$) are shown with respect to the asset yield ratio $\tilde{S}$ as this highlights how the early exercise boundary influences the comparative statics of the American exchange option. In all subsequent analyses, option prices and the early exercise boundary are shown for $v=0.56$ (the long-term variance $\eta$) and parameter values in Table \ref{tab-ParameterValues} except for those which were varied in the numerical experiment. Prices and the early exercise boundary were generated using the method of lines discussed in the previous sections. 

\begin{figure}
\centering
\subfloat[$\rho_1$ and $\rho_2$]{
	\includegraphics[width = 0.4\linewidth]{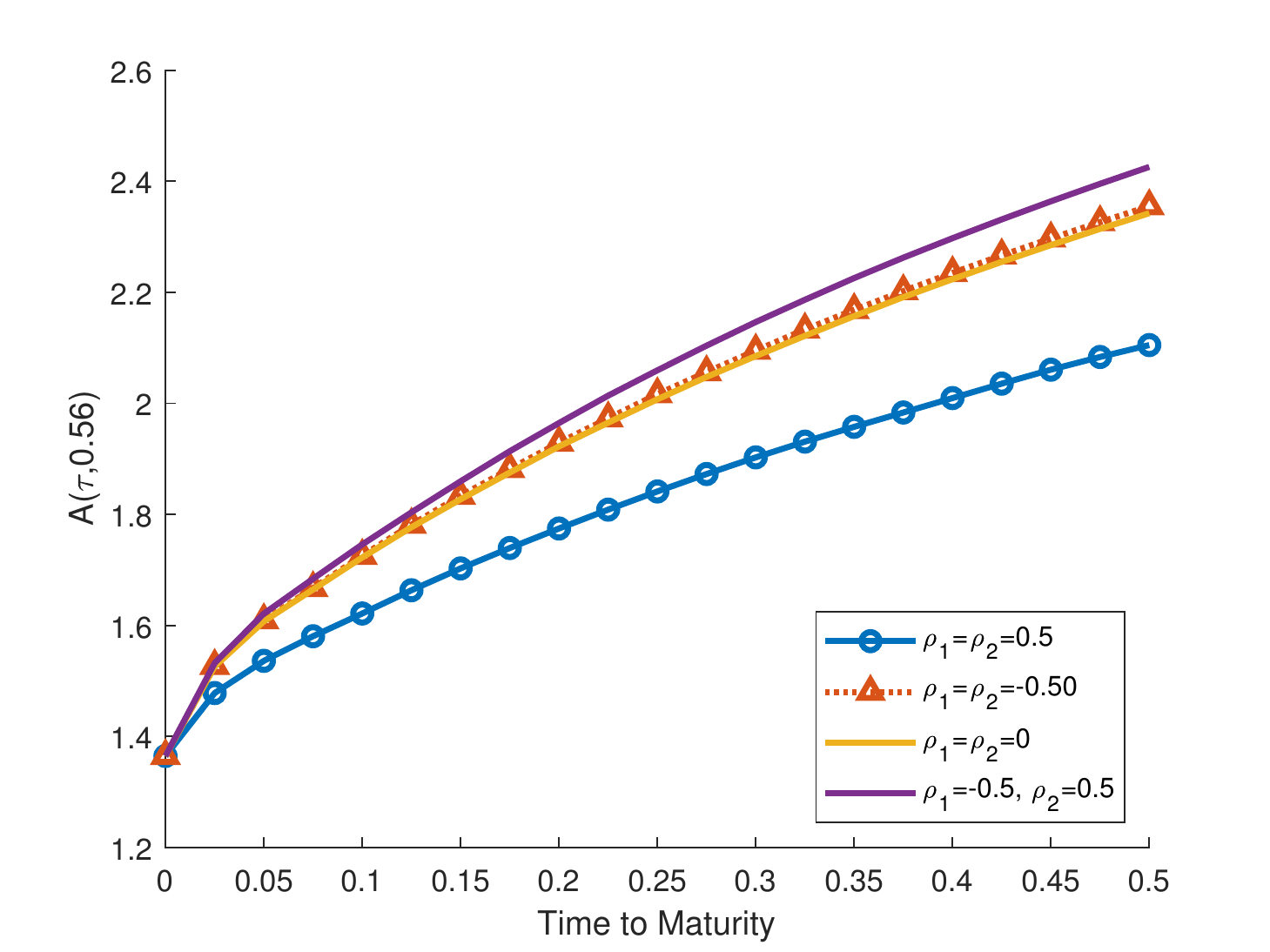}}
\subfloat[$\rho_w$]{
	\includegraphics[width = 0.4\linewidth]{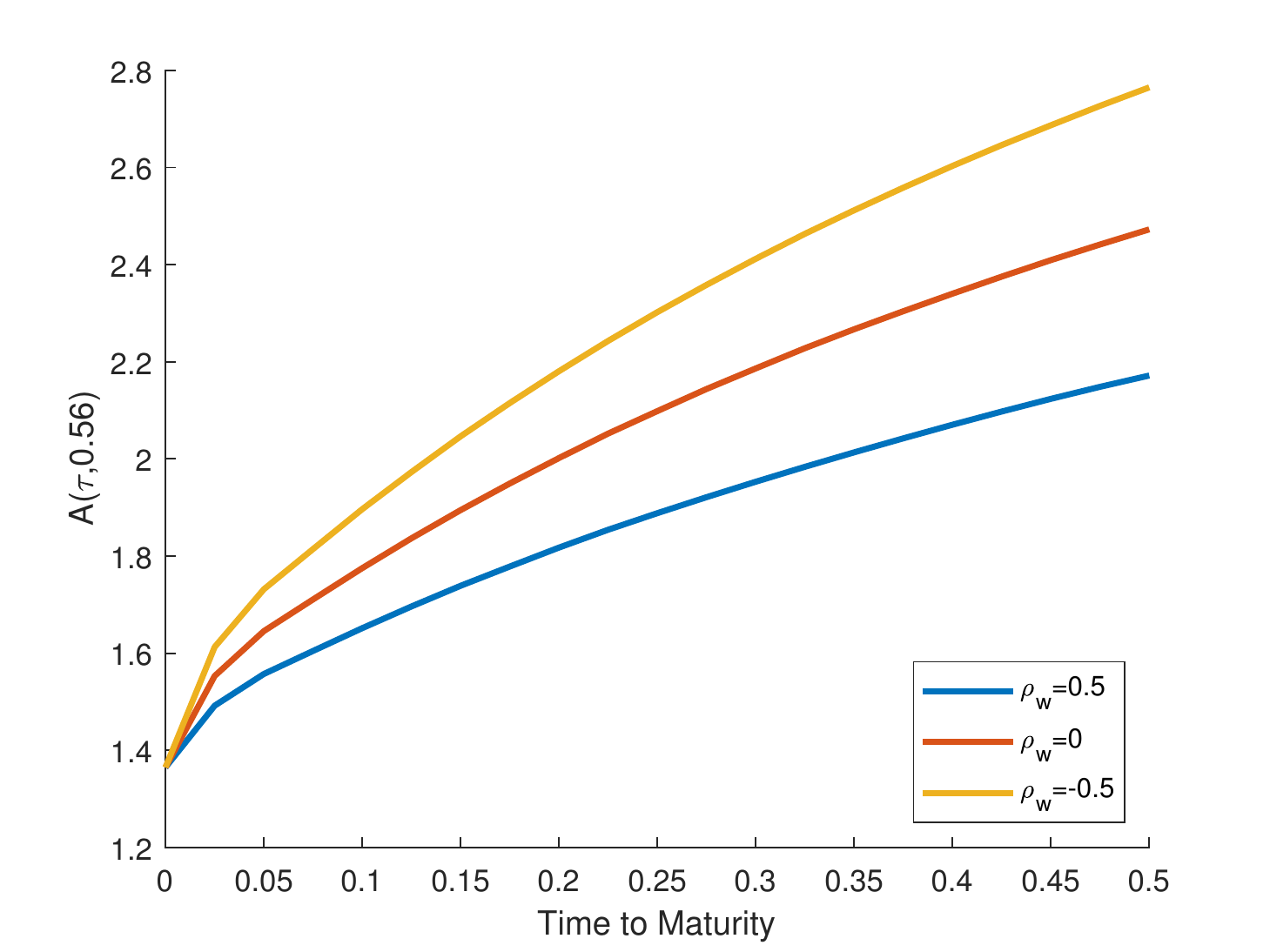}}
	
\subfloat[$\sigma_1$ and $\sigma_2$]{
	\includegraphics[width = 0.4\linewidth]{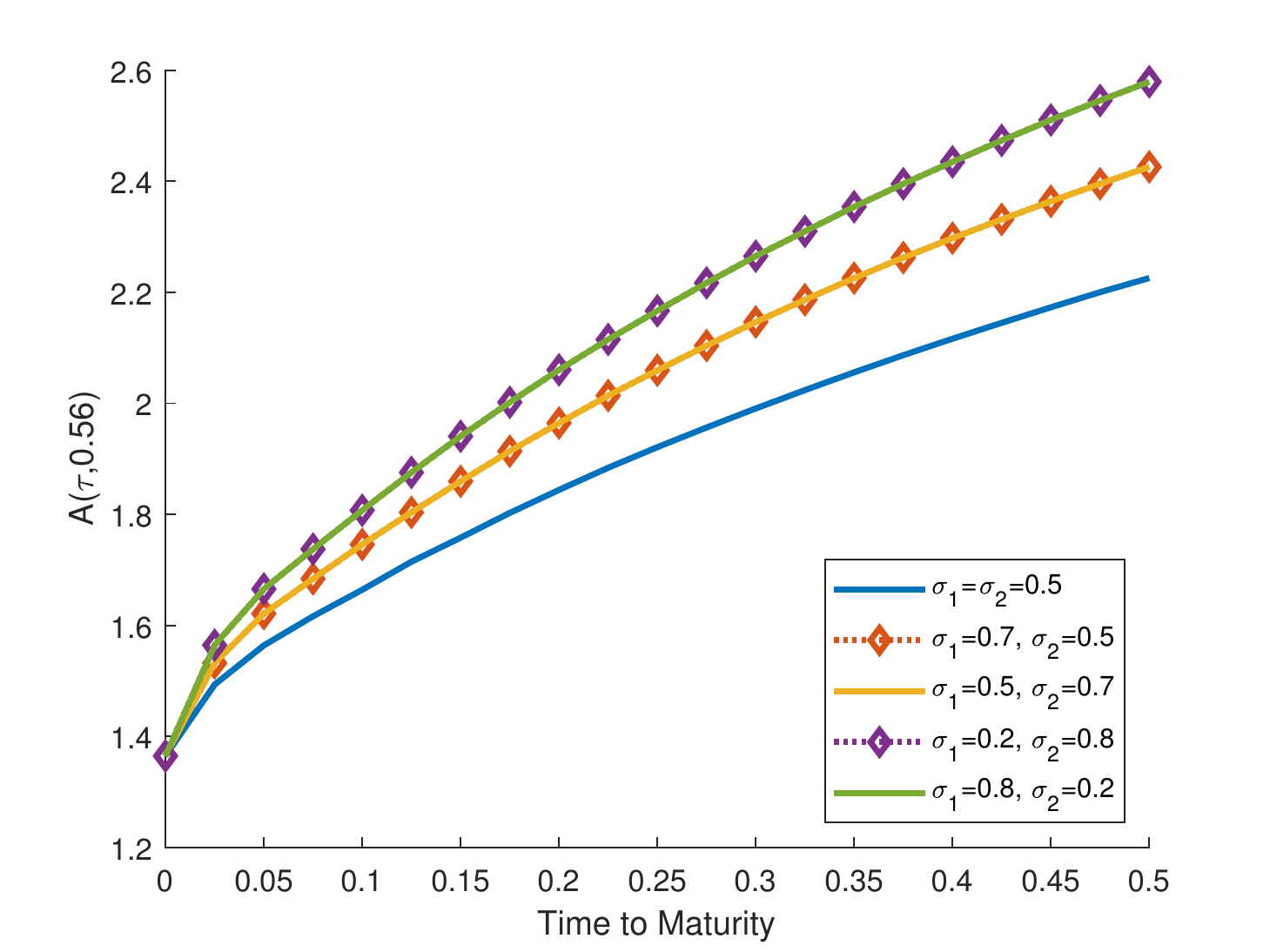}}
\subfloat[$\tilde{\lambda}_1$ and $\tilde{\lambda}_2$]{
	\includegraphics[width = 0.4\linewidth]{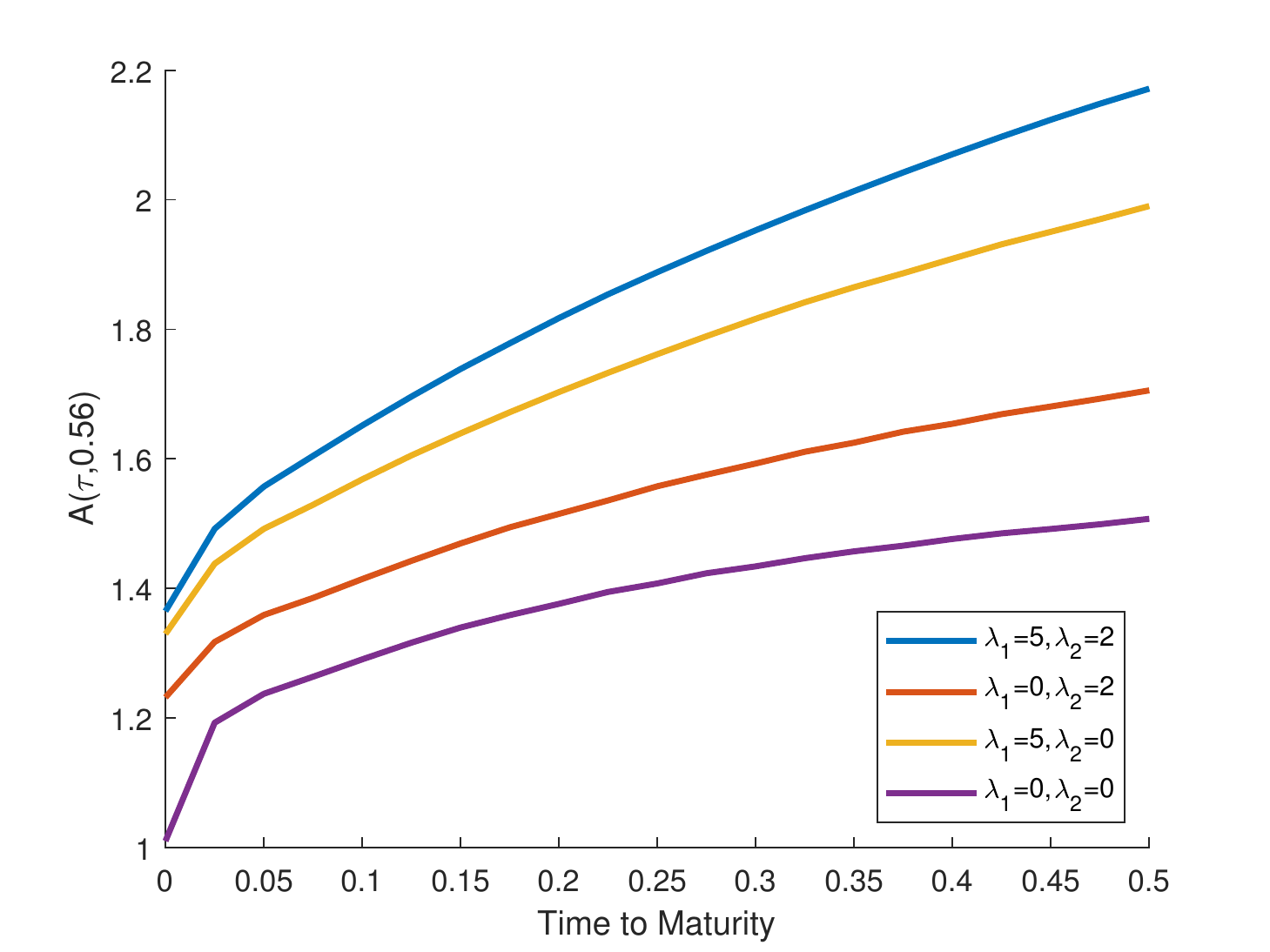}}

\subfloat[$\xi$]{
	\includegraphics[width = 0.4\linewidth]{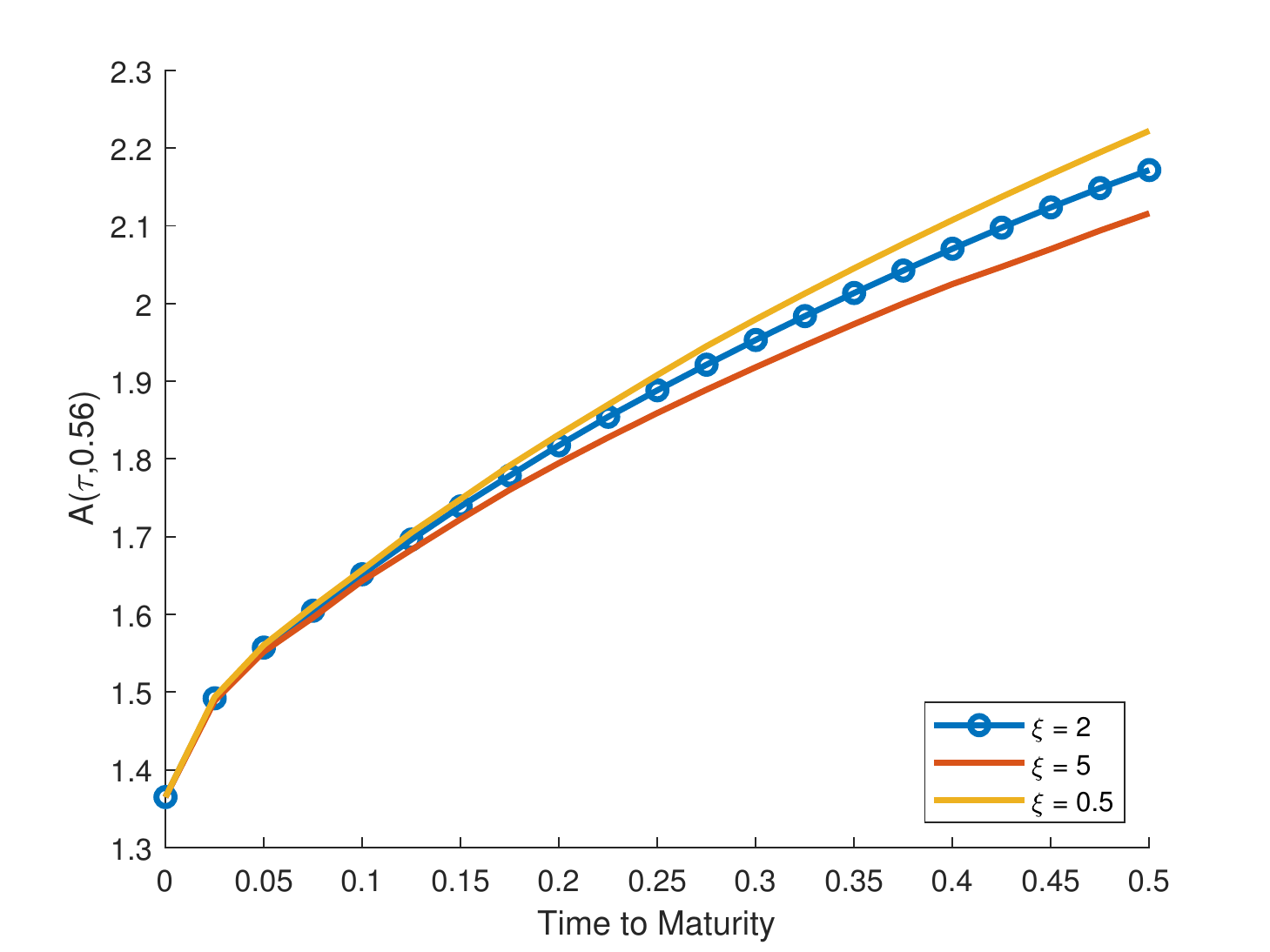}}
	\subfloat[$\Lambda$]{
	\includegraphics[width = 0.4\linewidth]{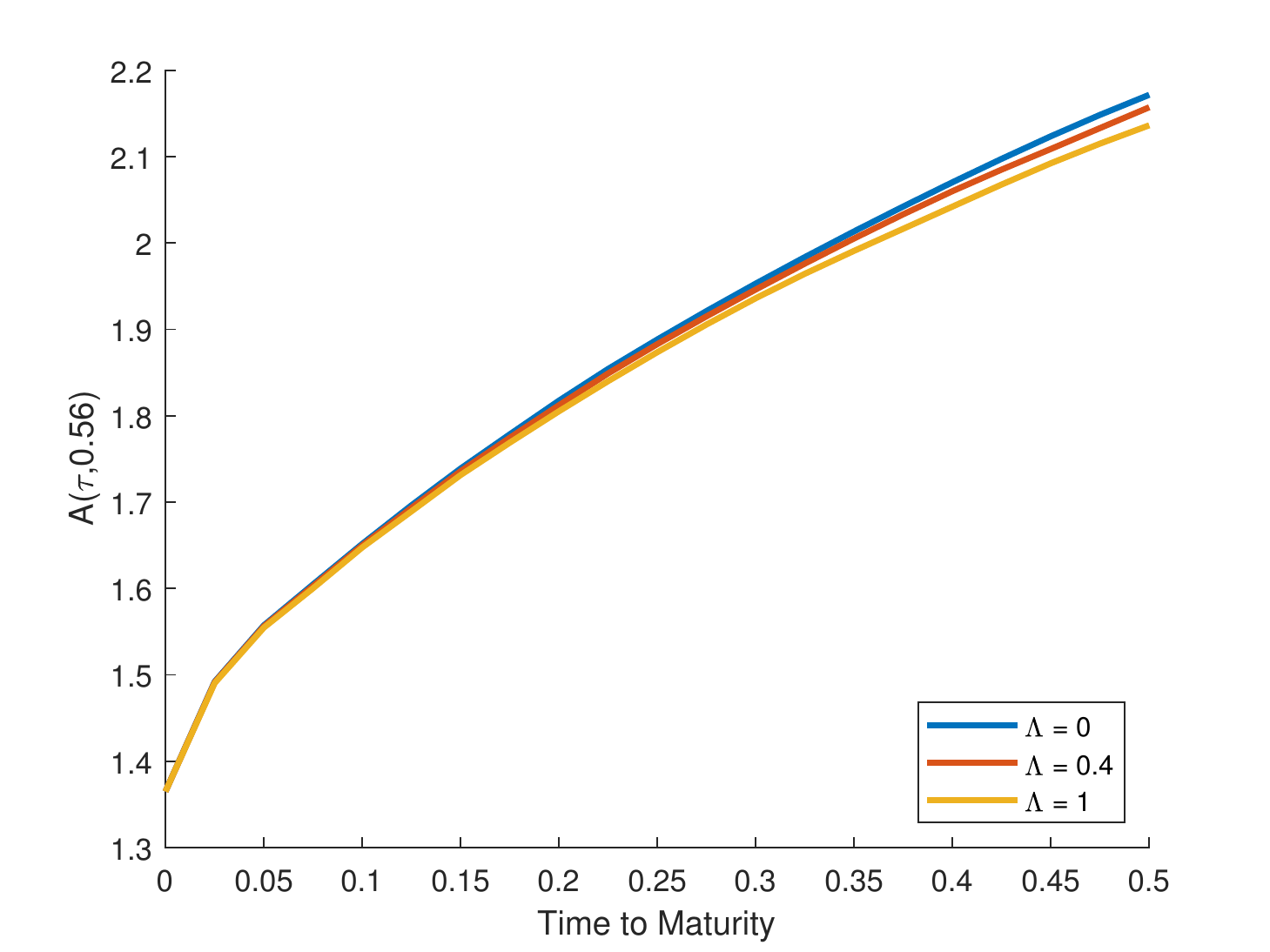}}


\caption{Effect of changes in model parameters on the early exercise boundary for $v=0.56$.}
\label{fig-ComStat-EEB}
\end{figure}

Figure \ref{fig-ComStat-EEB} shows the effect of the model parameters on the early exercise boundary. In all panels except for the fourth, we note that all boundary curves start at the limit at maturity ($\tau\to 0^+$) as calculated using the method in Section \ref{sec-PutCall-EEBLimit} ($A(0^+,0.56)\approx 1.4$ for the parameters in Table \ref{tab-ParameterValues}). This is expected as the limit depends only on the dividend yields and the jump parameters, and not on diffusion and variance parameters. The results of the numerical comparative statics are enumerated below.
\begin{itemize}
	\item The exercise boundary is lowest when both asset prices are positively correlated to the variance process, and is highest when the assets have opposing directions of correlation with the variance process. The boundaries generated when the correlations are zero or are both negative are similar, but are consistently higher than the positive correlation case. 
	\item A negative correlation between the asset prices also generates a consistently higher boundary curve than when $\rho_w=0$ or $\rho_w>0$, the latter generating the smallest values. 
	\item The boundary curve pivots upward when the gap between $\sigma_1$ and $\sigma_2$ increases, although the effect is symmetric whether it is asset 1 or 2 that has a greater proportionality coefficient.
	\item The early exercise boundary is not as sensitive to the rate of mean reversion $\xi$ and the market price of volatility $\Lambda$ compared to the diffusion coefficients. The differences are more pronounced when the option is far from maturity. 
	\item Changing the value of $\omega$ produces negligible changes in the boundary curve and is thus not shown here. This is somewhat contrary to the results of \citet{Chiarella-2009} who found that the impact of stochastic volatility on the free boundary for ordinary American call options is more pronounced for higher values of $\omega$. Our conclusions may be due to the way volatility affects both the numerator and denominator of the asset yield ratio. Thus, there is a possibility that an increase in one asset price due to more volatile volatility may be canceled out by the same phenomenon in the other asset price.
\end{itemize}

As discussed in Section \ref{sec-PutCall-EEBLimit}, increasing the jump intensities $\tilde{\lambda}_1$ and $\tilde{\lambda}_2$ introduces upward shifts in the early exercise boundary (as opposed to the upward pivots as seen in the other simulations). If no jumps occur in both asset prices (i.e. when both asset prices are modelled as correlated stochastic volatility processes), the boundary at maturity is $\max\left\{1,q_2/q_1\right\}e^{(q_1-q_2)(T-T)} = 1$ and the entire curve lies much lower than when at least one of the asset prices has jumps. While it is not shown in Figure \ref{fig-ComStat-EEB}, the impact of the jump intensities is symmetric in the sense that, for example, the same boundary curve is generated when $\tilde{\lambda}_1=5, \tilde{\lambda}_2=0$ and when $\tilde{\lambda}_1=0, \tilde{\lambda}_2=5$. It is also expected that changing the jump size density parameters (in our case, the mean and variance of the normal distribution governing $Y_1$ and $Y_2$) will also introduce upward or downward shifts in the boundary curve as these also directly affect the limit of the boundary at maturity. A similar effect may also be observed when assuming a different distribution for the jump sizes, e.g. the double exponential jumps assumed by \citet{Kou-2002}.

We now proceed to numerical comparative statics on the discounted European and American prices. Here, we display the difference when the price generated after modifying a given parameter is subtracted from the option price solved with the parameters in Table \ref{tab-ParameterValues}. A common observation is that price differences vanish for deeply in-the-money American exchange options because of the value-matching condition. The point at which the price differences meet the baseline is the exercise boundary, which is more clearly explicated in Figure \ref{fig-ComStat-EEB}. In contrast, the price differences for deeply in-the-money European options dissipate more slowly, if at all as can be seen in some simulations.

\begin{figure}
\centering
\subfloat[European exchange option]{
	\includegraphics[width = 0.4\linewidth]{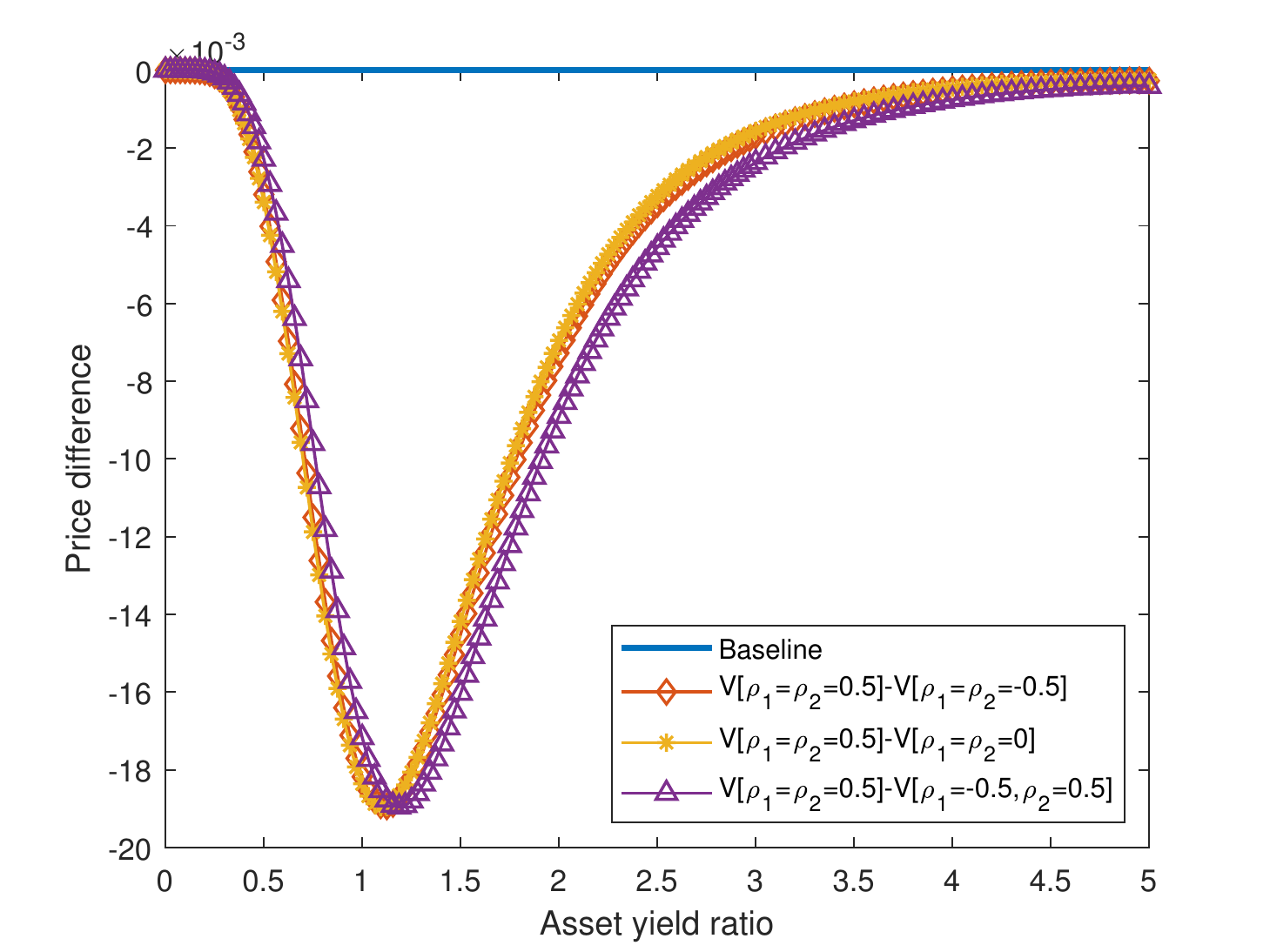}}
\subfloat[American exchange option]{
	\includegraphics[width = 0.4\linewidth]{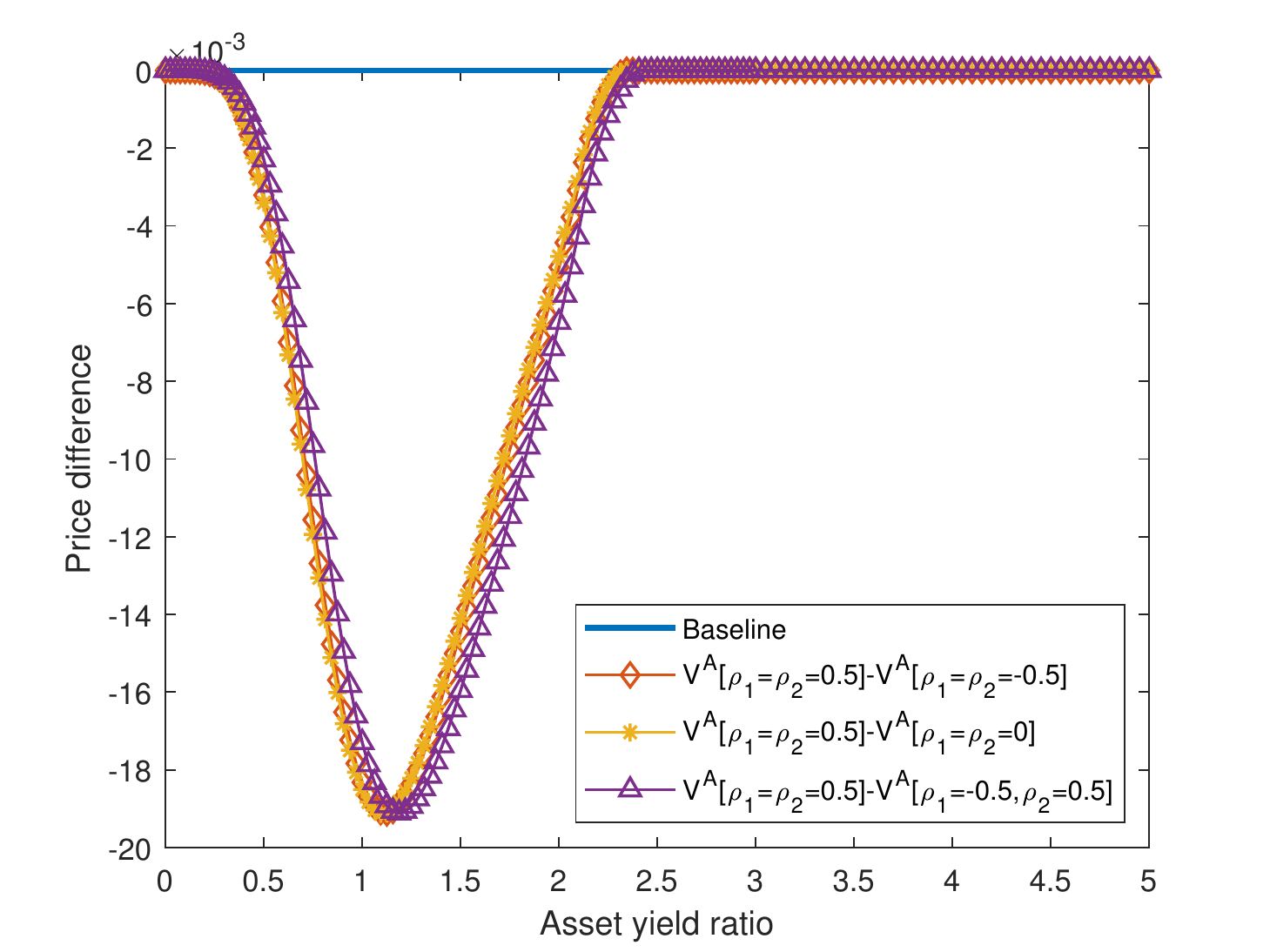}}
\caption{Price differences in European and American exchange option prices for various values of the correlations $\rho_1$ and $\rho_2$ between the asset price and the variance process at $t=0$ and $v=0.56$.}
\label{fig-ComStat-Rho1}
\end{figure}

\begin{figure}
\centering
\subfloat[European exchange option]{
	\includegraphics[width = 0.4\linewidth]{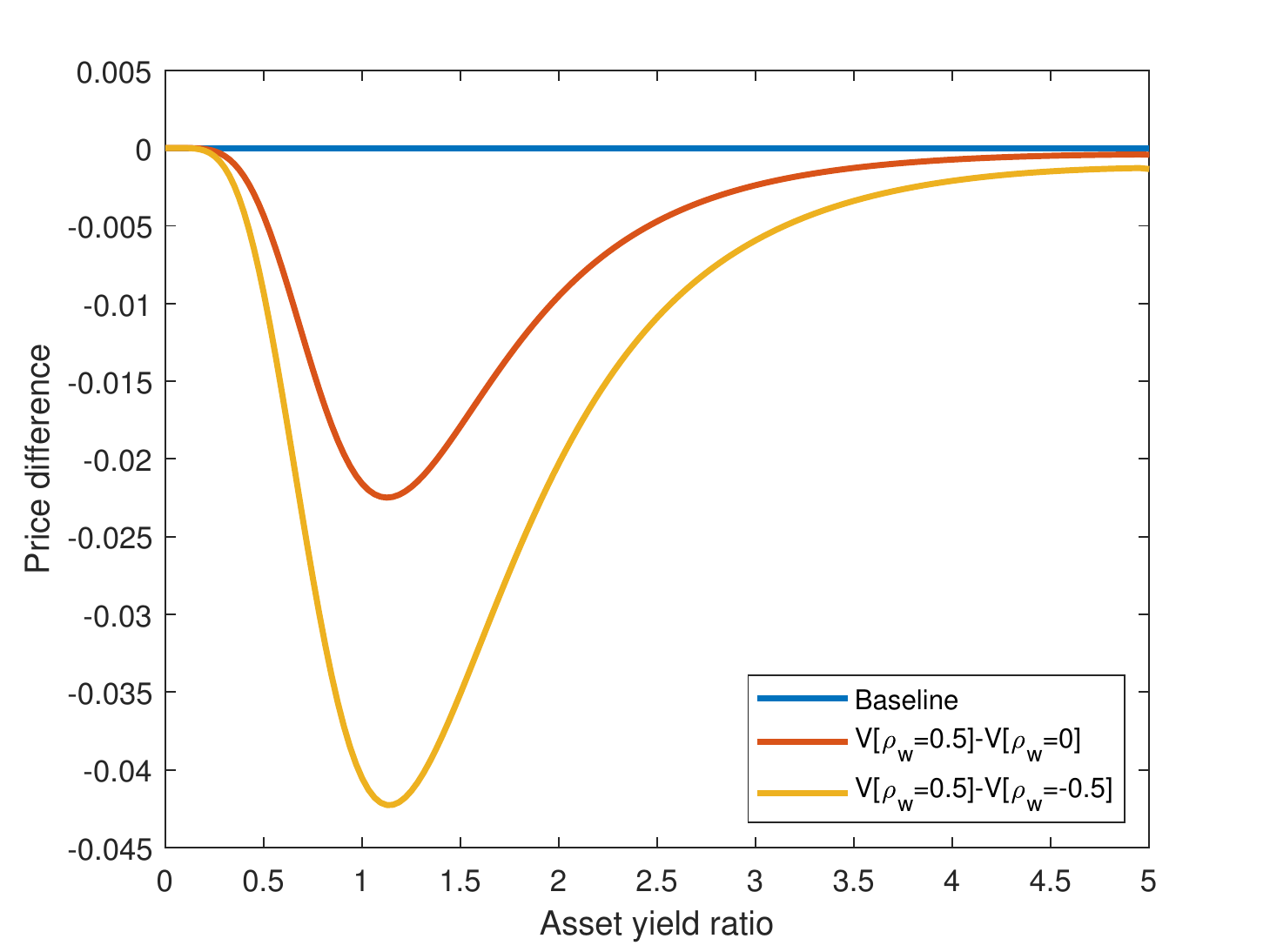}}
\subfloat[American exchange option]{
	\includegraphics[width = 0.4\linewidth]{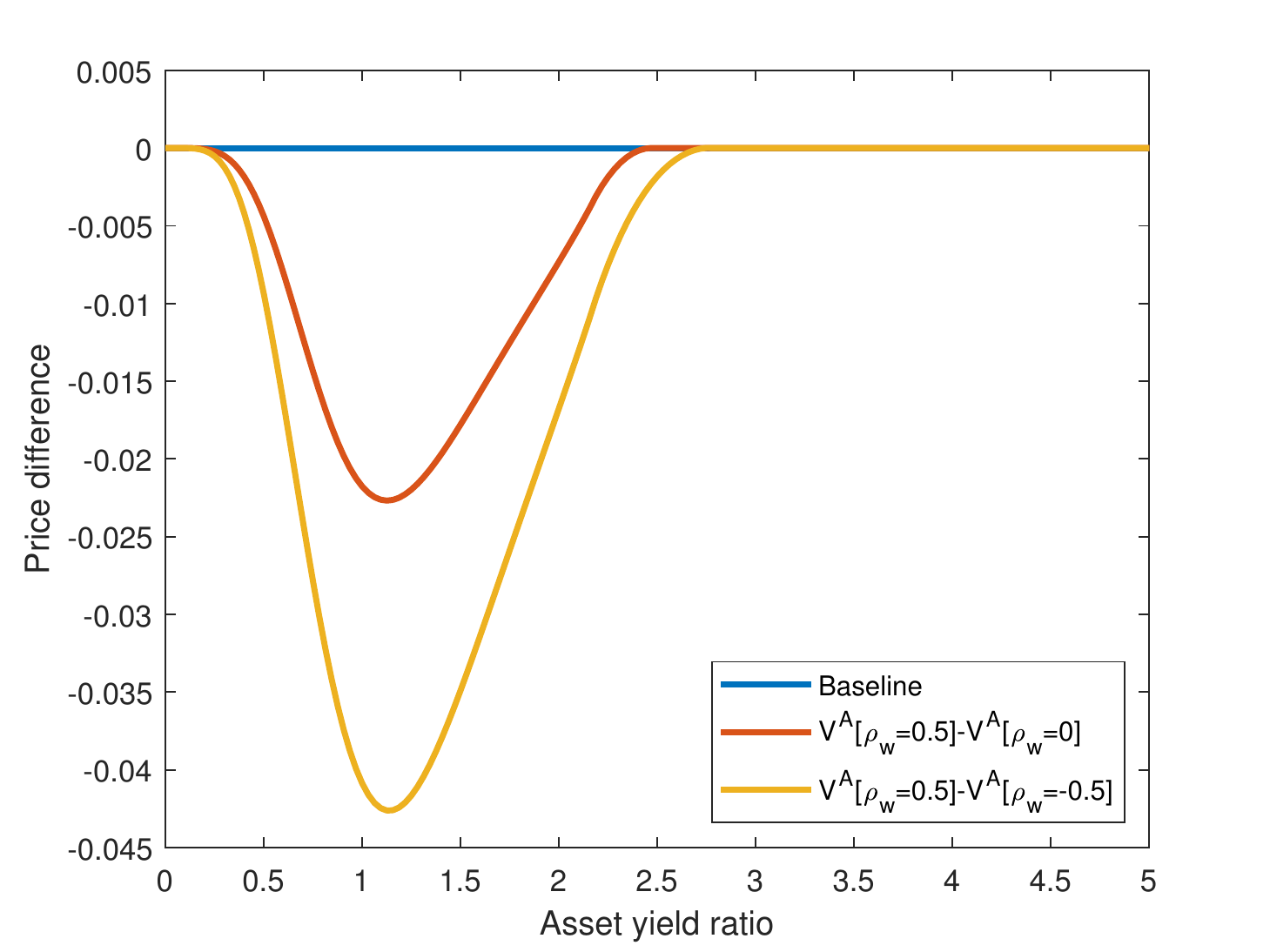}}
\caption{Price differences in European and American exchange option prices for various values of the correlation $\rho_w$ between the asset price processes at $t=0$ and $v=0.56$.}
\label{fig-ComStat-RhoW}
\end{figure}

Figures \ref{fig-ComStat-Rho1} and \ref{fig-ComStat-RhoW} show the effect of the correlation coefficients on the discounted exchange option prices. In these simulations, we find that deviating from the default correlations in Table \ref{tab-ParameterValues} resulted to higher MOL prices for both types of options. However, the price differences do not vary drastically between assumed alternative values for correlations with the Wiener process in the variance process. In contrast, price differences between cases are more pronounced when the correlation between the asset prices is changed. In particular, both types of options are more expensive when the asset prices are negatively correlated to each other, as there is a higher tendency for larger spreads between the prices of the two assets. From Figure \ref{fig-ComStat-RhoW}, option prices are lowest when the assets are positively correlated with each other. In both analyses, the price differences are maximal when the options are at-the-money.

\begin{figure}
\centering
\subfloat[European exchange option]{
	\includegraphics[width = 0.4\linewidth]{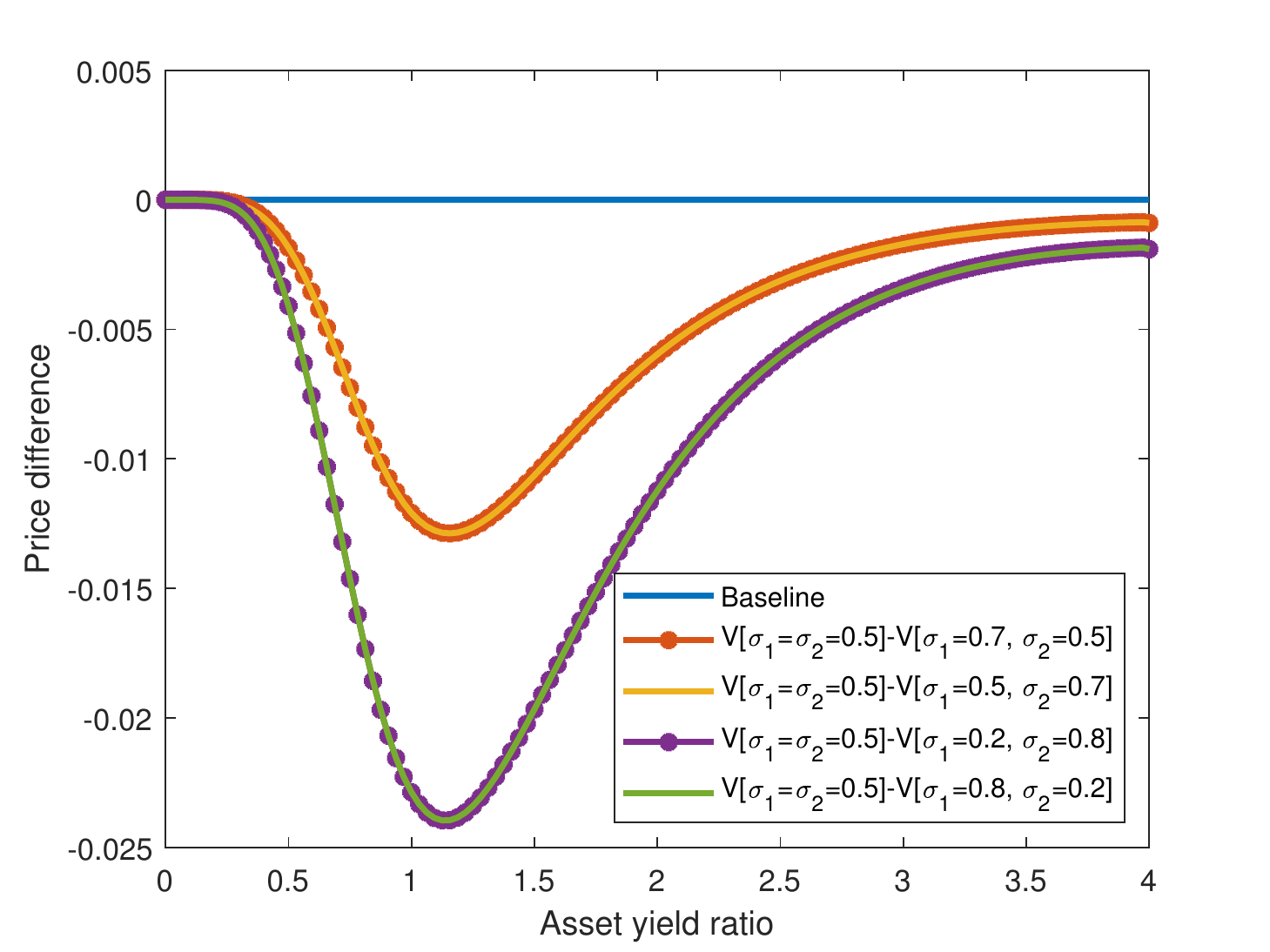}}
\subfloat[American exchange option]{
	\includegraphics[width = 0.4\linewidth]{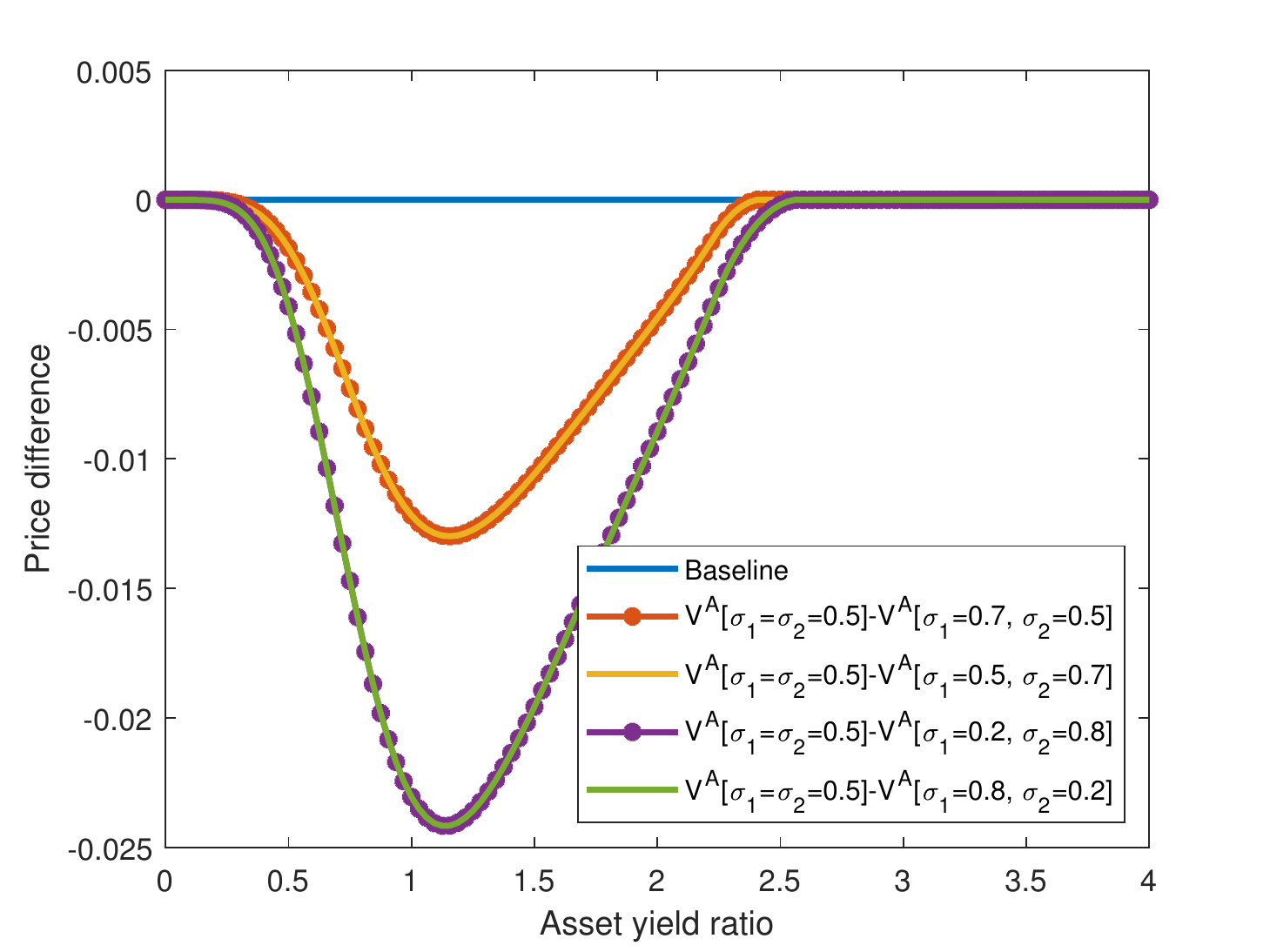}}
\caption{Price differences in European and American exchange option prices for various values of the diffusion coefficients $\sigma_1$ and $\sigma_2$ of the asset price processes at $t=0$ and $v=0.56$.}
\label{fig-ComStat-Sigma}
\end{figure}

Option prices also tend to increase when the gap between the proportionality coefficients $\sigma_1$ and $\sigma_2$ in the diffusion term of the asset price dynamics increases, as indicated by Figure \ref{fig-ComStat-Sigma}. Similar to the analysis on the early exercise boundary, it does not matter which asset is more sensitive to the instantaneous variance as the effect on option prices is symmetric. We note however that there is a noticeable price difference for deeply in-the-money European options.

\begin{figure}
\centering
\subfloat[European exchange option]{
	\includegraphics[width = 0.4\linewidth]{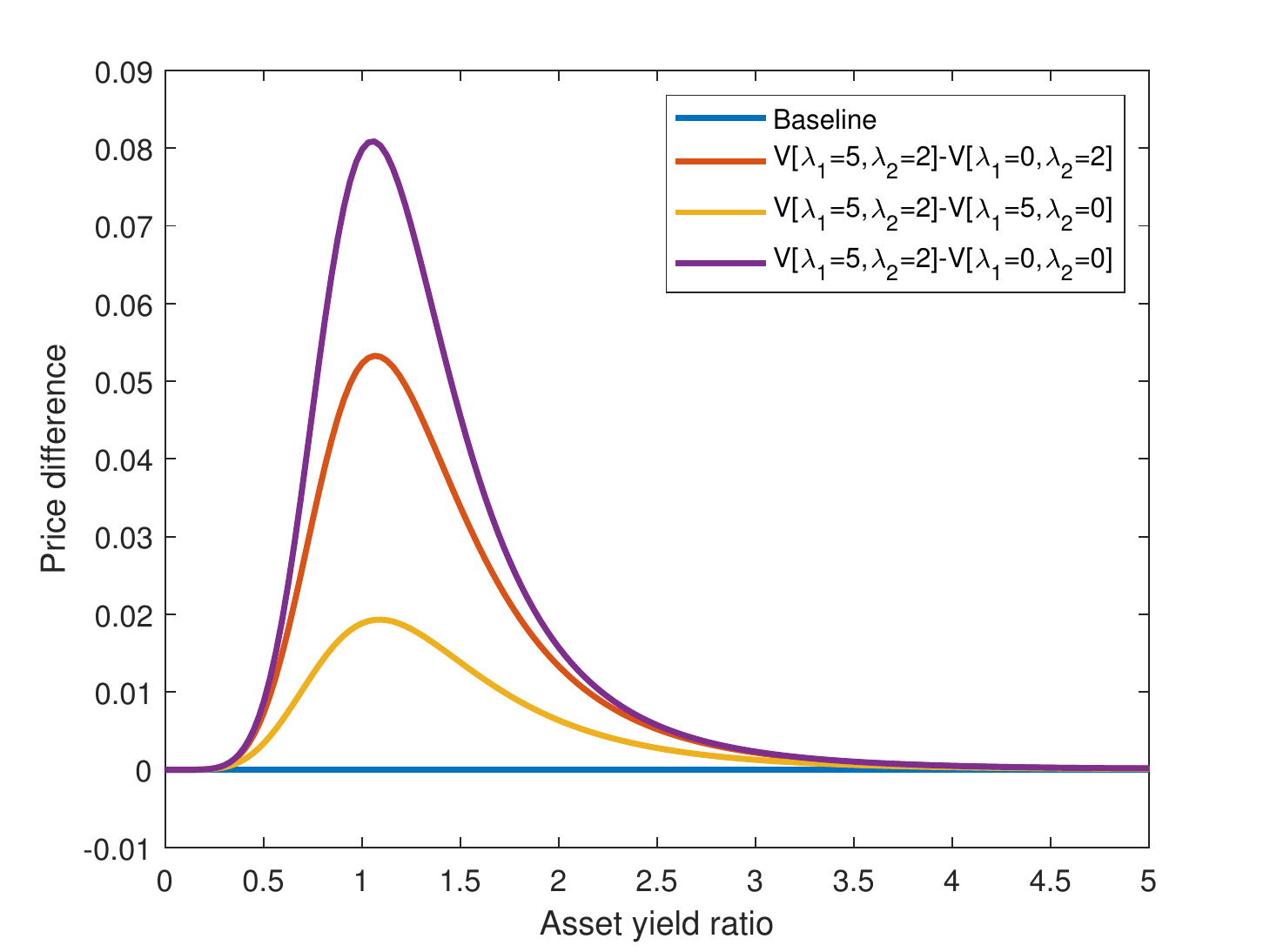}}
\subfloat[American exchange option]{
	\includegraphics[width = 0.4\linewidth]{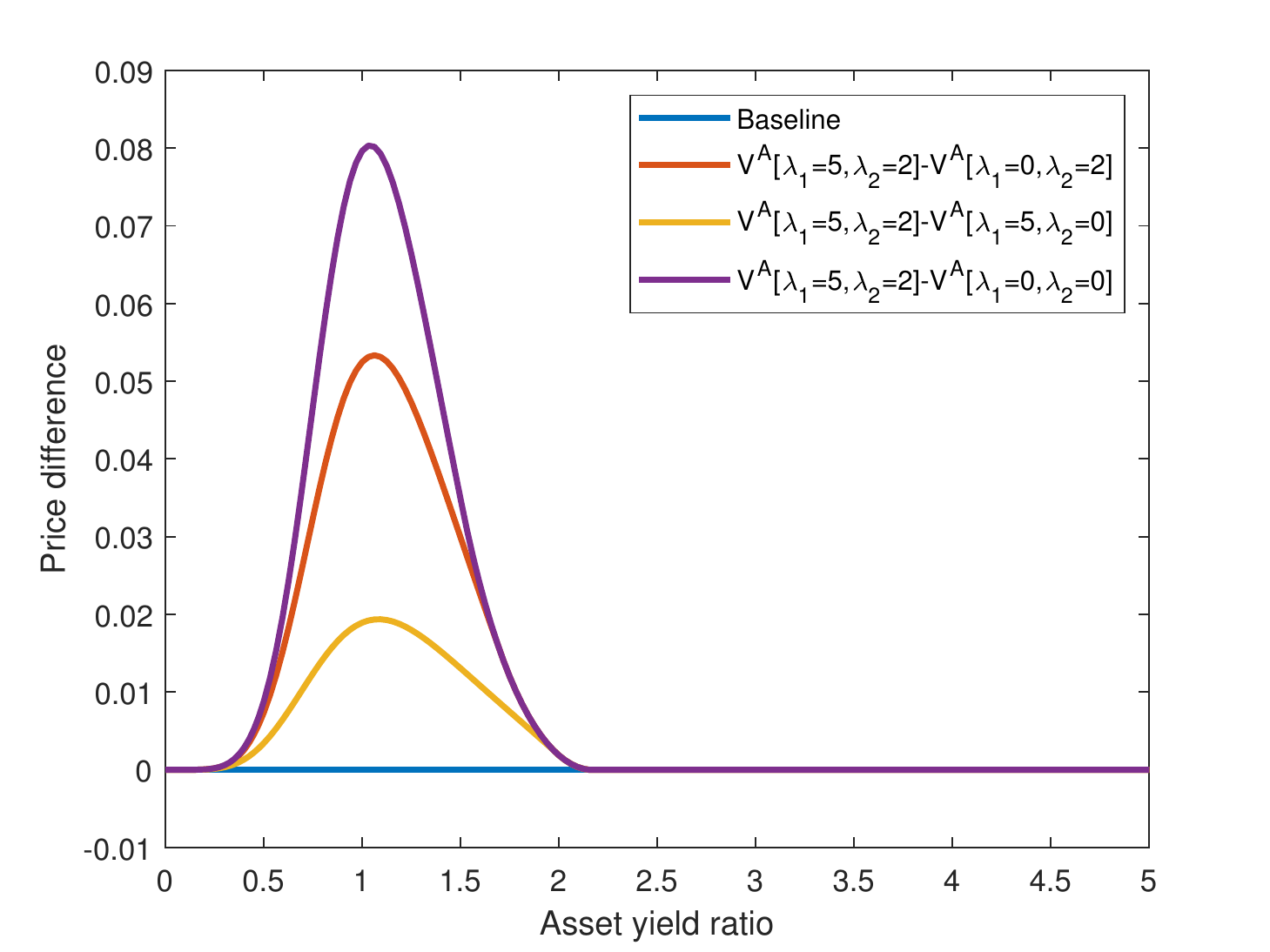}}
\caption{Price differences in European and American exchange option prices for various values of the jump intensity parameters $\tilde{\lambda}_1$ and $\tilde{\lambda}_2$ at $t=0$ and $v=0.56$.}
\label{fig-CompStat-Lambda}
\end{figure}

Figure \ref{fig-CompStat-Lambda} shows how the jump intensity rates affect the discounted exchange option prices. The baseline reference, where it is possible for both assets prices to jump, is priced the highest. The largest difference is observed between the baseline and the stochastic volatility case, where neither asset price jumps. While it is not shown here, we also observe the same symmetry in how $\tilde{\lambda}_1$ and $\tilde{\lambda}_2$ affect option prices, as was observed for the early exercise boundary. In the comparative static analysis presented in this section, changing the jump intensities result in the largest differences in option prices, reaching up to 0.08 when the option is at-the-money. The magnitude of differences may be affected, however, by the other choices for the jump size distribution and/or its parameters. 

\begin{figure}
\centering
\subfloat[European exchange option]{
	\includegraphics[width = 0.4\linewidth]{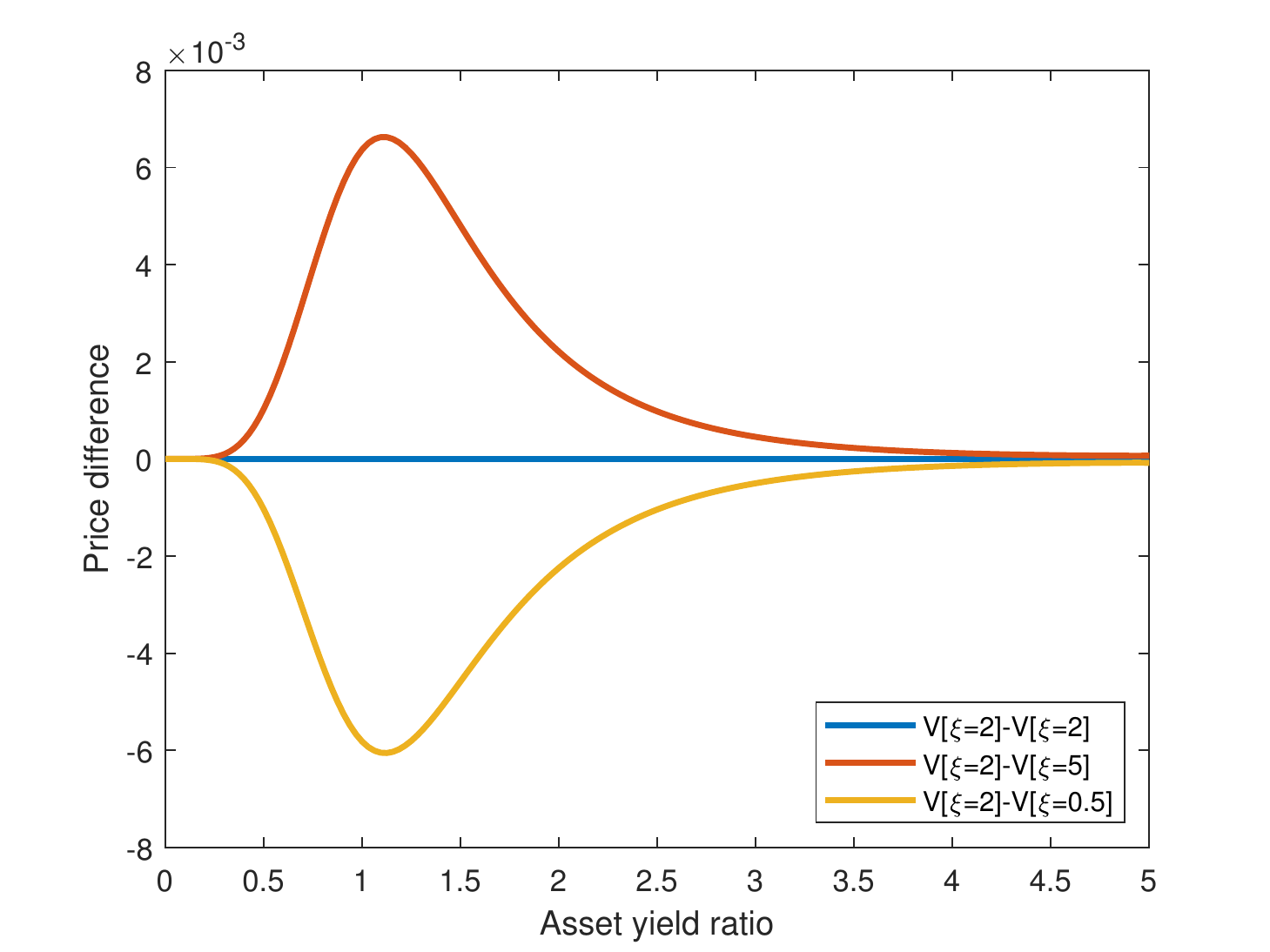}}
\subfloat[American exchange option]{
	\includegraphics[width = 0.4\linewidth]{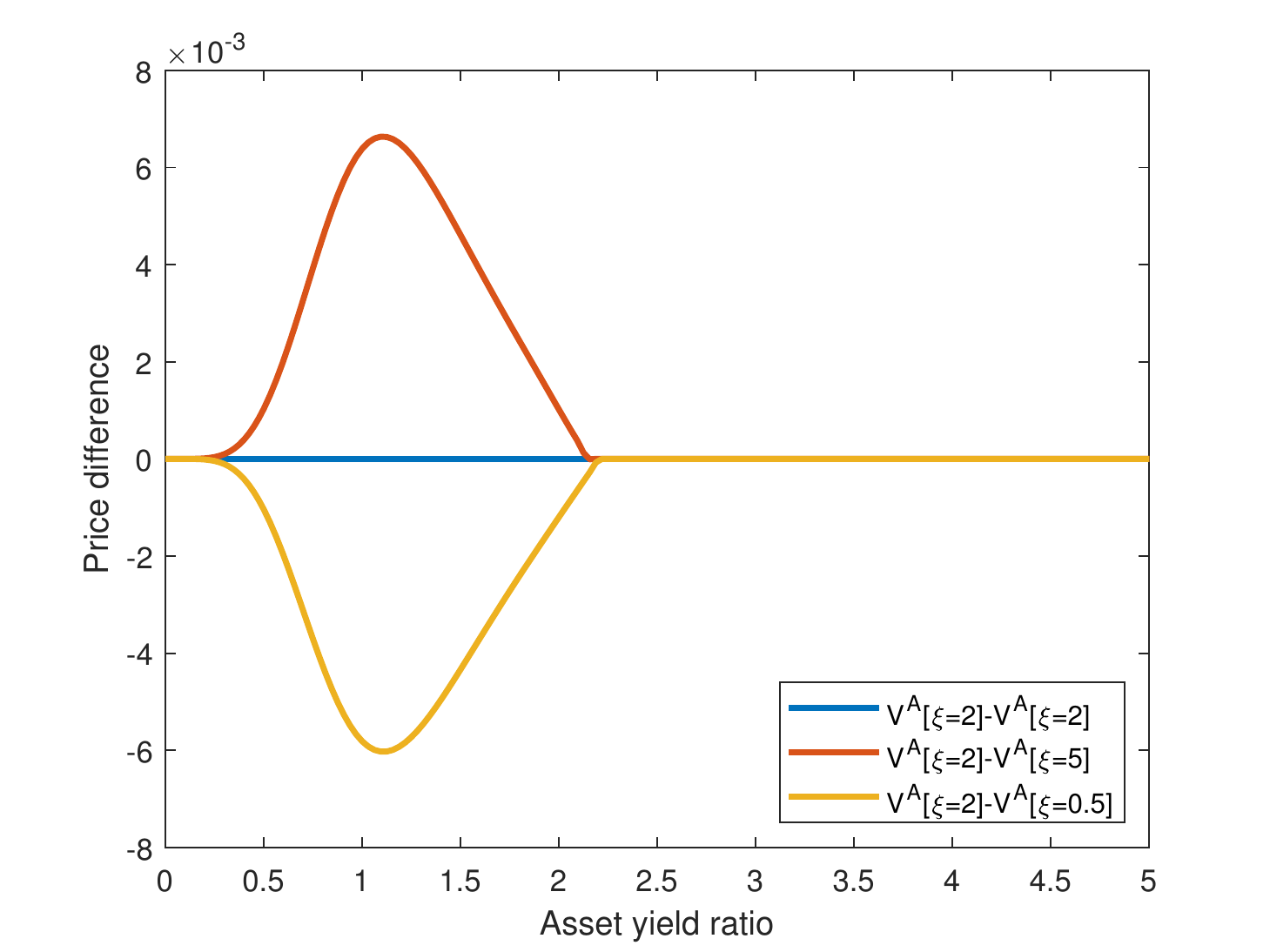}}
\caption{Price differences in European and American exchange option prices for various values of the rate $\xi$ of mean reversion of the variance process at $t=0$ and $v=0.56$.}
\label{fig-ComStat-Xi}
\end{figure}

\begin{figure}
\centering
\subfloat[European exchange option]{
	\includegraphics[width = 0.4\linewidth]{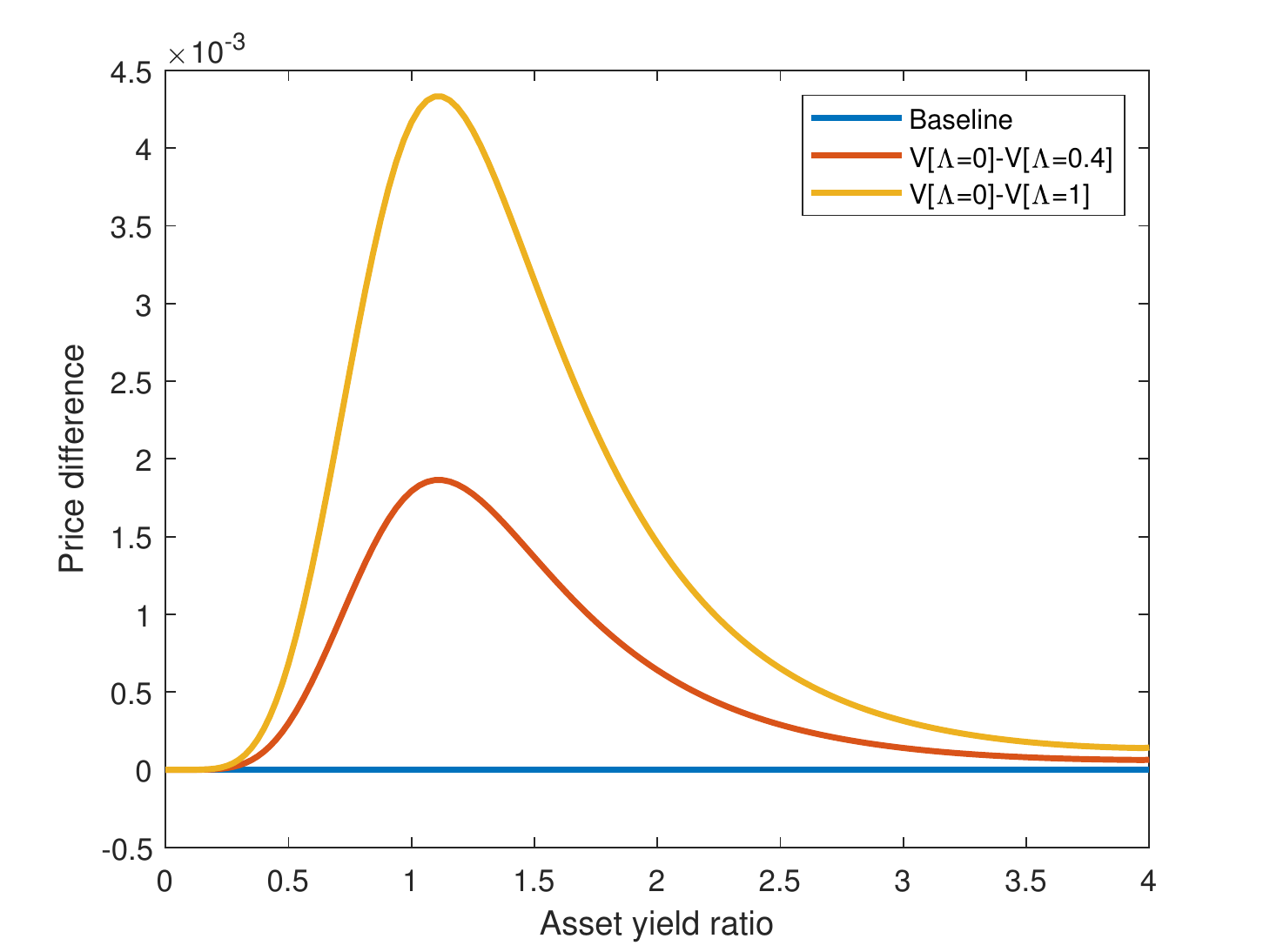}}
\subfloat[American exchange option]{
	\includegraphics[width = 0.4\linewidth]{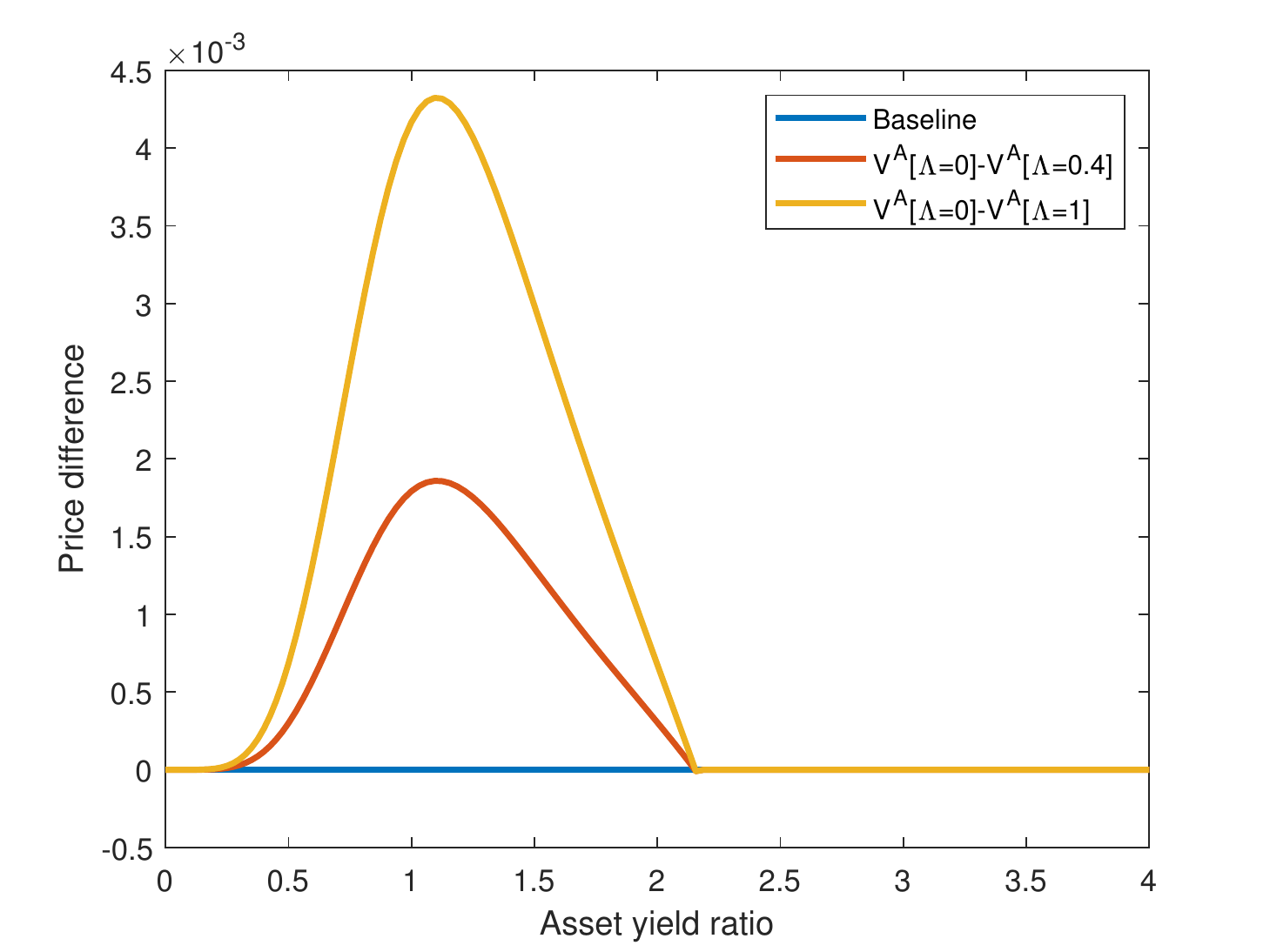}}
\caption{Price differences in European and American exchange option prices for various values of the market price of volatility risk parameter $\Lambda$ at $t=0$ and $v=0.56$.}
\label{fig-ComStat-MktVol}
\end{figure}

\begin{figure}
\centering
\subfloat[European exchange option]{
	\includegraphics[width = 0.4\linewidth]{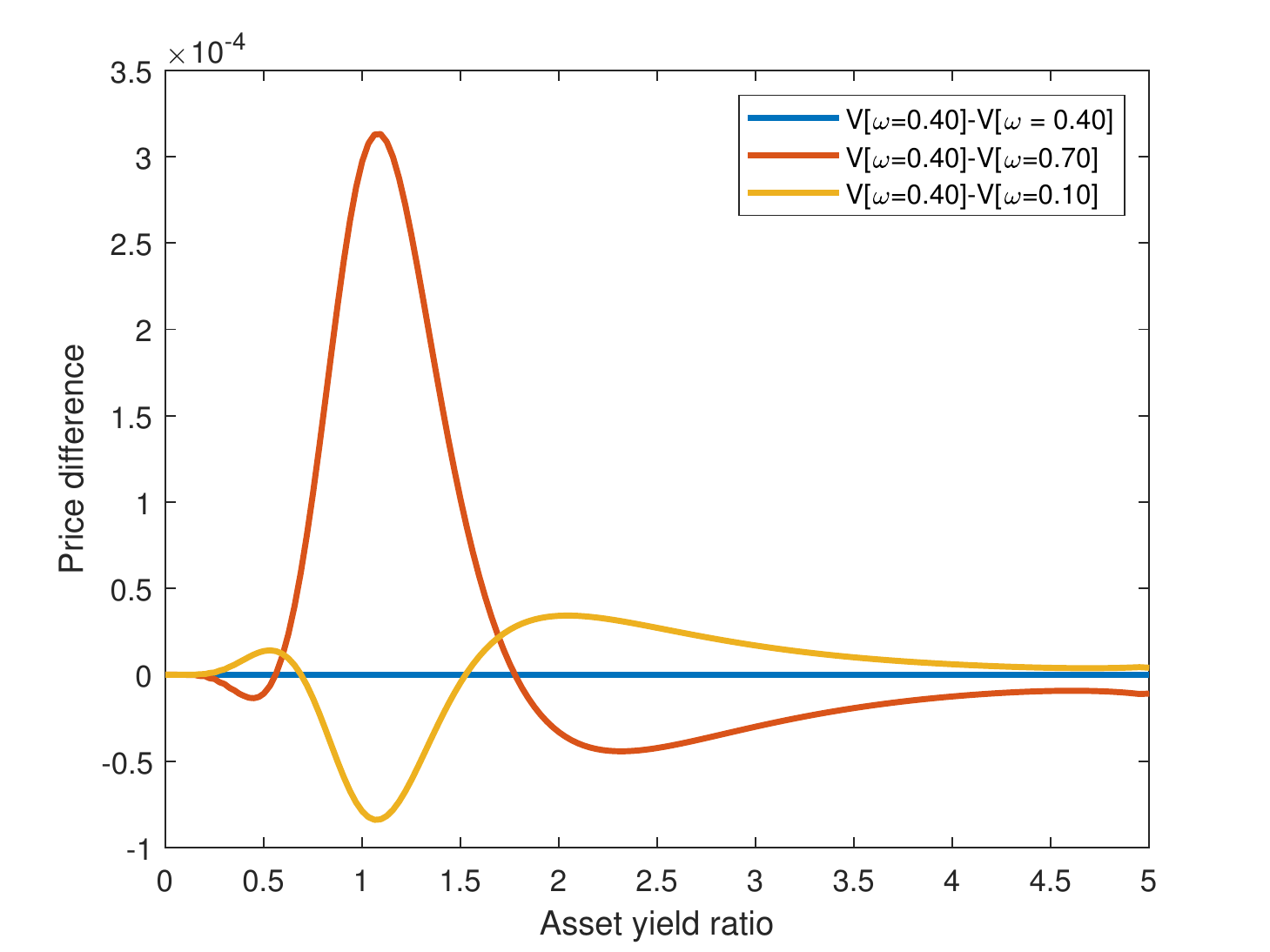}}
\subfloat[American exchange option]{
	\includegraphics[width = 0.4\linewidth]{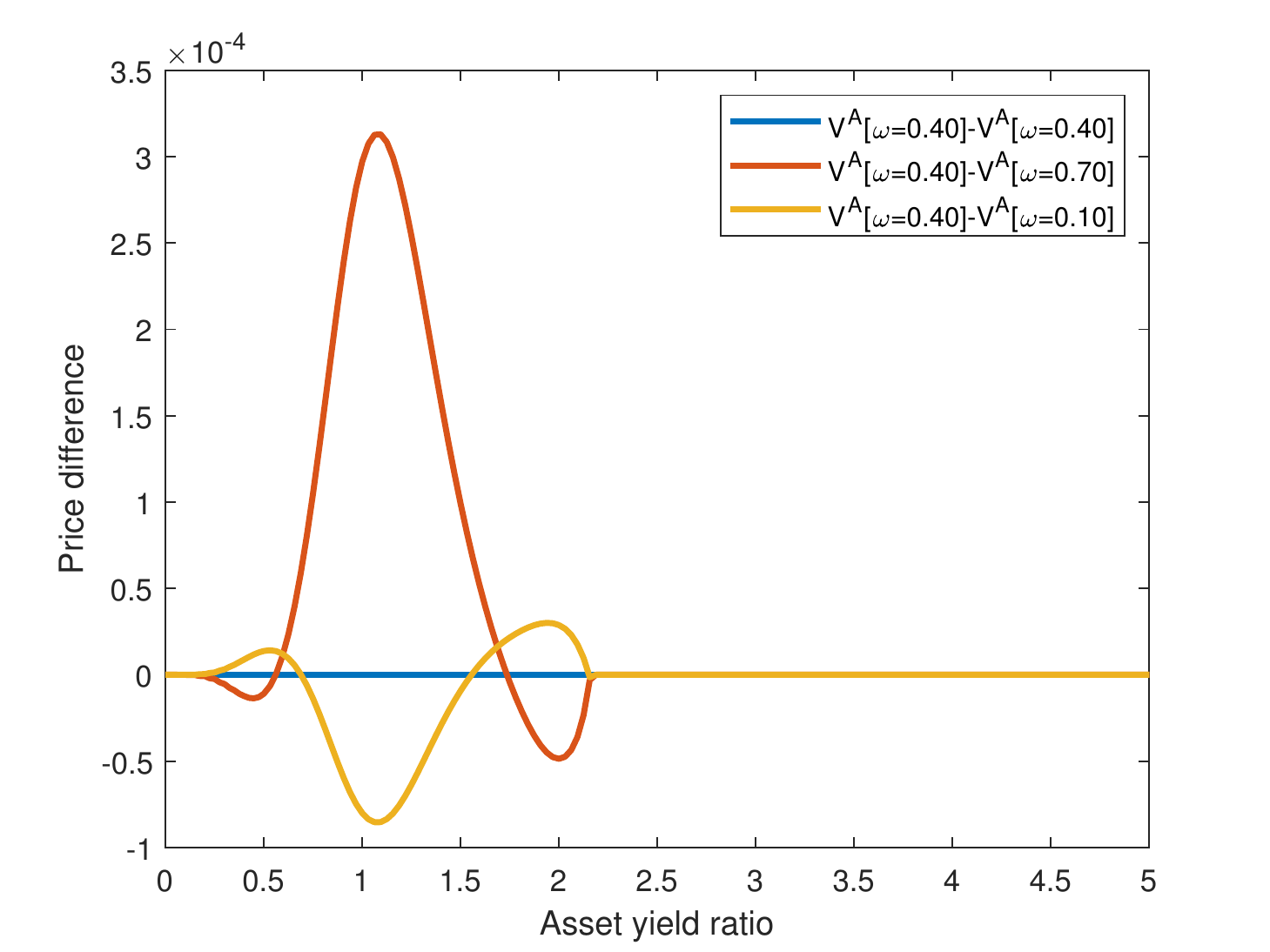}}
\caption{Price differences in European and American exchange option prices for various values of the volatility of volatility $\omega$ at $t=0$ and $v=0.56$.}
\label{fig-ComStat-Omega}
\end{figure}

Figures \ref{fig-ComStat-Xi}, \ref{fig-ComStat-MktVol}, and \ref{fig-ComStat-Omega} exhibit the price differences resulting from varying the values of the stochastic volatility parameters. From Figure \ref{fig-ComStat-Xi}, it can be seen that increasing the rate of mean reversion tends to decrease the discounted exchange option prices. The same conclusion can be drawn for the market price of volatility, although there is a slight price difference even when the European option is deep in-the-money. 

As seen from Figure \ref{fig-ComStat-Omega}, however, the effect of the volatility of volatility is less straightforward. Increasing $\omega$ (relative to the default $\omega=0.40$) results to prices which are lower when the option is near-the-money but higher when the option is either deeply out-of-the-money or in-the-money. The reverse is true for when the volatility of volatility is decreased. There also seems to be a noticeable difference for deeply-in-the-money European exchange options.

\begin{figure}
\centering
\subfloat[$\rho_w = 0.5$]{
	\includegraphics[width = 0.4\linewidth]{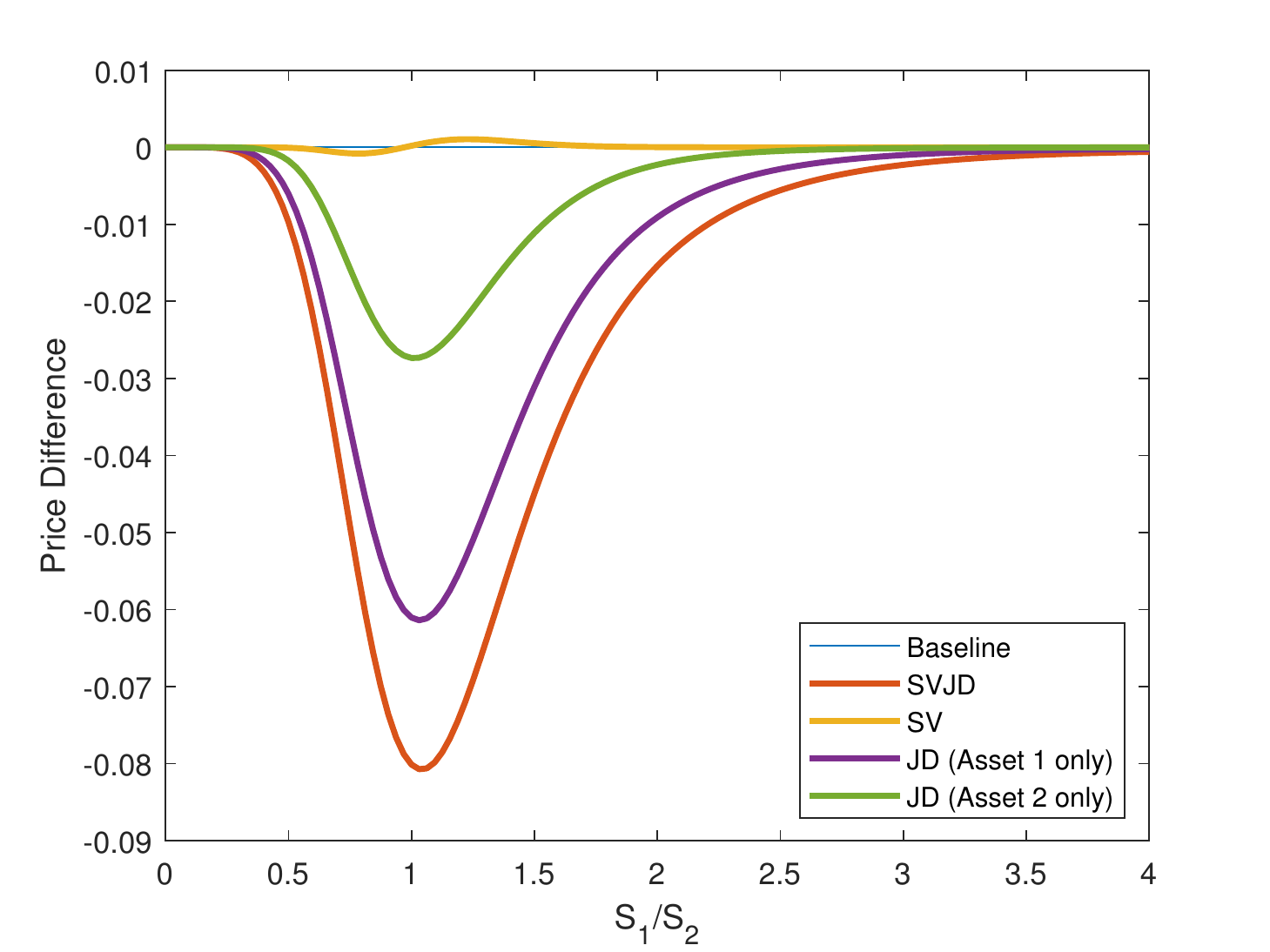}}
\subfloat[$\rho_w = -0.5$]{
	\includegraphics[width = 0.4\linewidth]{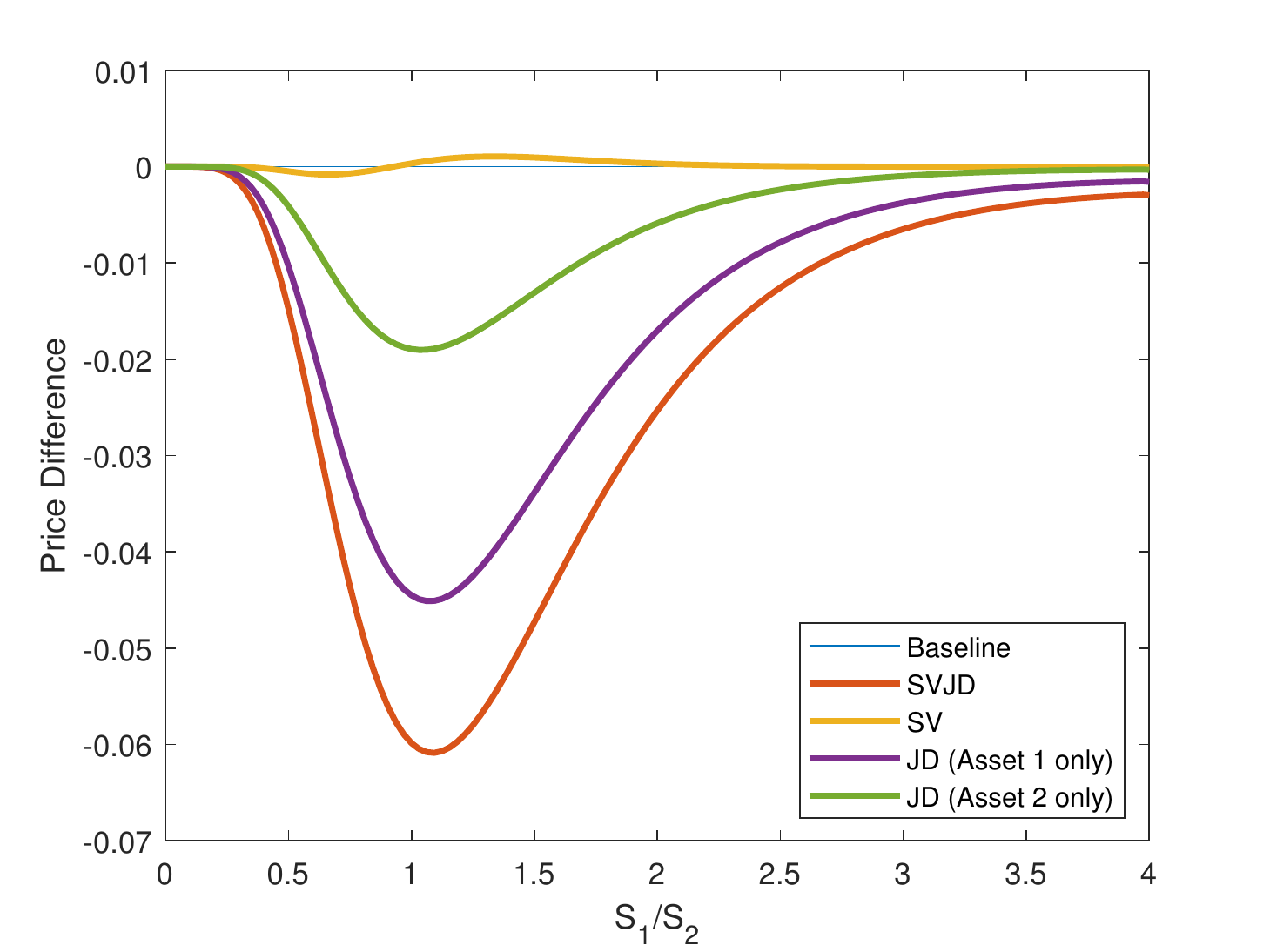}}
\caption{Comparison of discounted European exchange option prices at $t=0$ and $v=0.56$ generated under pure diffusion, stochastic volatility, SVJD (jumps in asset 1 only), SVJD (jumps in asset 2 only), and SVJD (jumps in both assets).}
\label{fig-CompDiffusion}
\end{figure}

The next numerical experiment is concerned with assessing the impact of stochastic volatility and jumps to the discounted price of European exchange options. The base prices correspond to the pure diffusion case, for which a formula has been provided by \citet{Margrabe-1978}. For simplicity, we assume no dividend yields for both assets. Furthermore, since the variance process is mean-reverting, we assume that the constant volatility for the pure diffusion case is $\sigma_i\sqrt{\eta_i}$ for asset $i=1,2$, as was done by \citet{ChiarellaZiveyi-2013}. The stochastic volatility case was simulated by setting the jump intensities to zero, but these were subsequently allowed to have nonzero values eventually building up to the default SVJD case. When exactly one of the jump intensities was equal to zero, it was assumed that this asset is driven by stochastic volatility dynamics whereas the other asset had both stochastic volatility and jumps. Here, we focus only on the European case, since as exhibited in the previous numerical experiments the American case behaves similarly, except for when the asset yield ratio exceeds the exercise boundary where the price differences vanish. Cases for a negative and a positive correlation between asset price processes are considered.

As seen in Figure \ref{fig-CompDiffusion}, stochastic volatility prices are slightly higher than the constant volatility case when the option is out-of-the-money, but are lower when the option is in-the-money. This observation is consistent with the numerical results of \citet{Heston-1993} and \citet{Chiarella-2009} for options on a single asset. The addition of jumps by allowing nonzero jump intensities generated higher option prices irrespective of the moneyness of the option, with the highest prices attained when both assets have SVJD dynamics. We find that price differences eventually vanish for large enough asset price ratios in the positive correlation case, but the differences persist in the negative correlation case (at least within the assumed range of values for the asset price ratio). In both positive and negative correlation cases, the SV and SVJD option prices converge to the Margrabe price when the option is deeply out-of-the-money. In contrast to the results of \citet{Chiarella-2009}, who found that the price differences invert from positive to negative (and vice versa) depending on the moneyness of the option, we find that price differences are consistently negative when jumps are involved. This may be attributed to how the jump and stochastic volatility parameters are chosen and the possibility that jumps may dominate stochastic volatility in terms of contributions to the overall variance in asset prices.

\subsection{Effect of Alternative Boundary Conditions at $v=v_M$}
\label{sec-NumericalVenttsel}

\begin{figure}
\centering
\subfloat[Early exercise boundary]{
	\includegraphics[width = 0.33\linewidth]{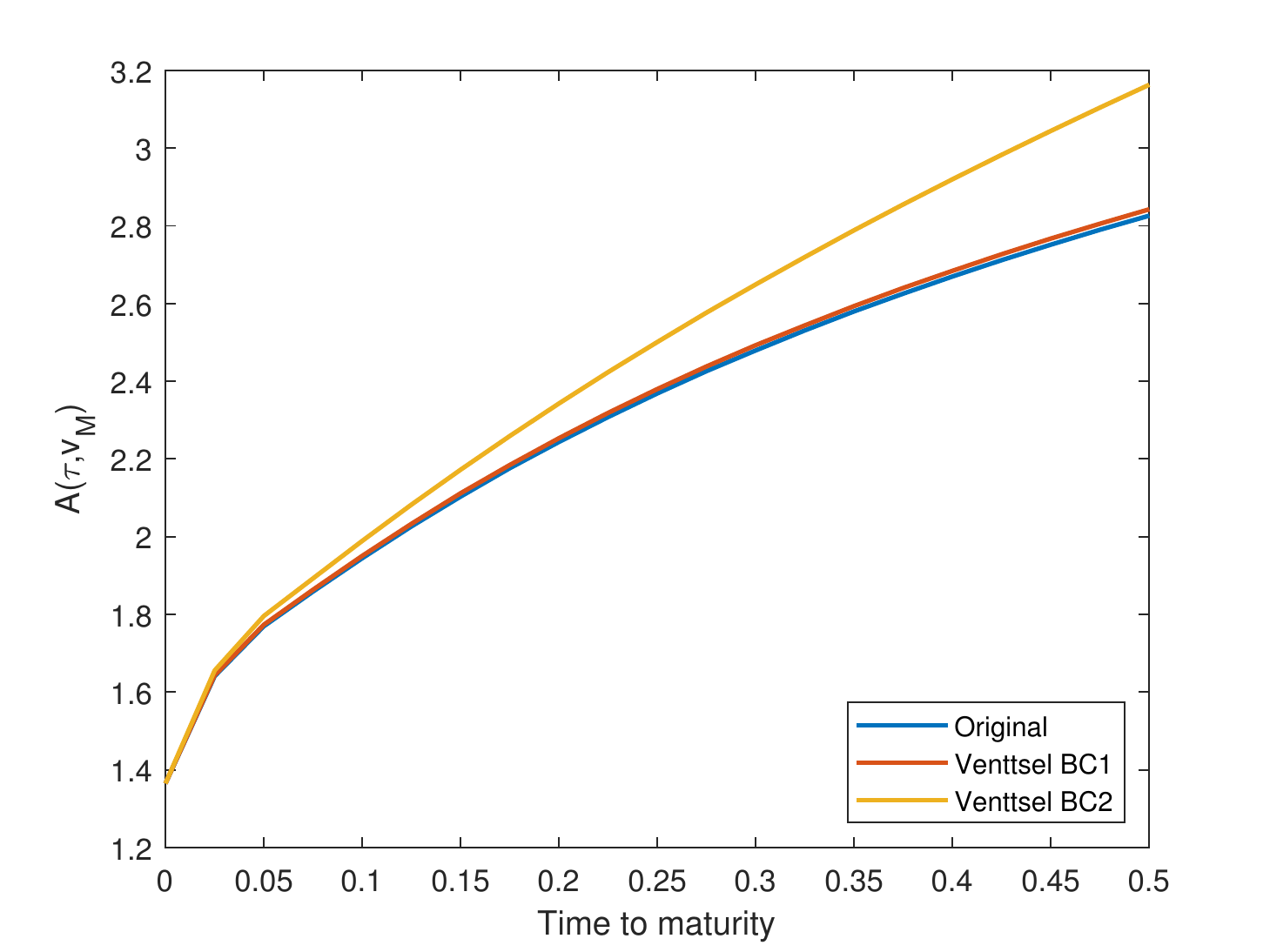}}
\subfloat[Discounted prices]{
	\includegraphics[width = 0.33\linewidth]{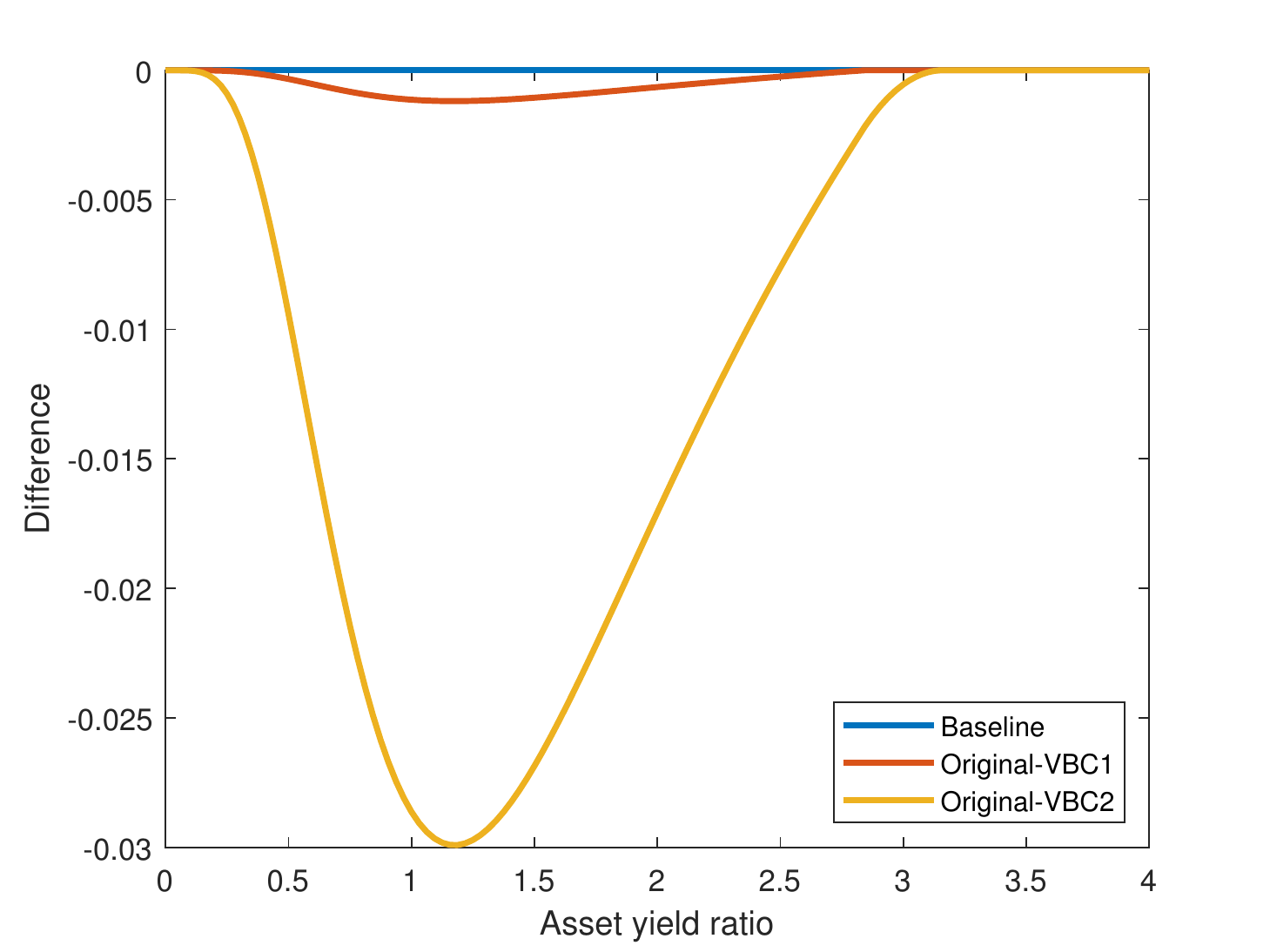}}
\subfloat[Discounted deltas]{
	\includegraphics[width = 0.33\linewidth]{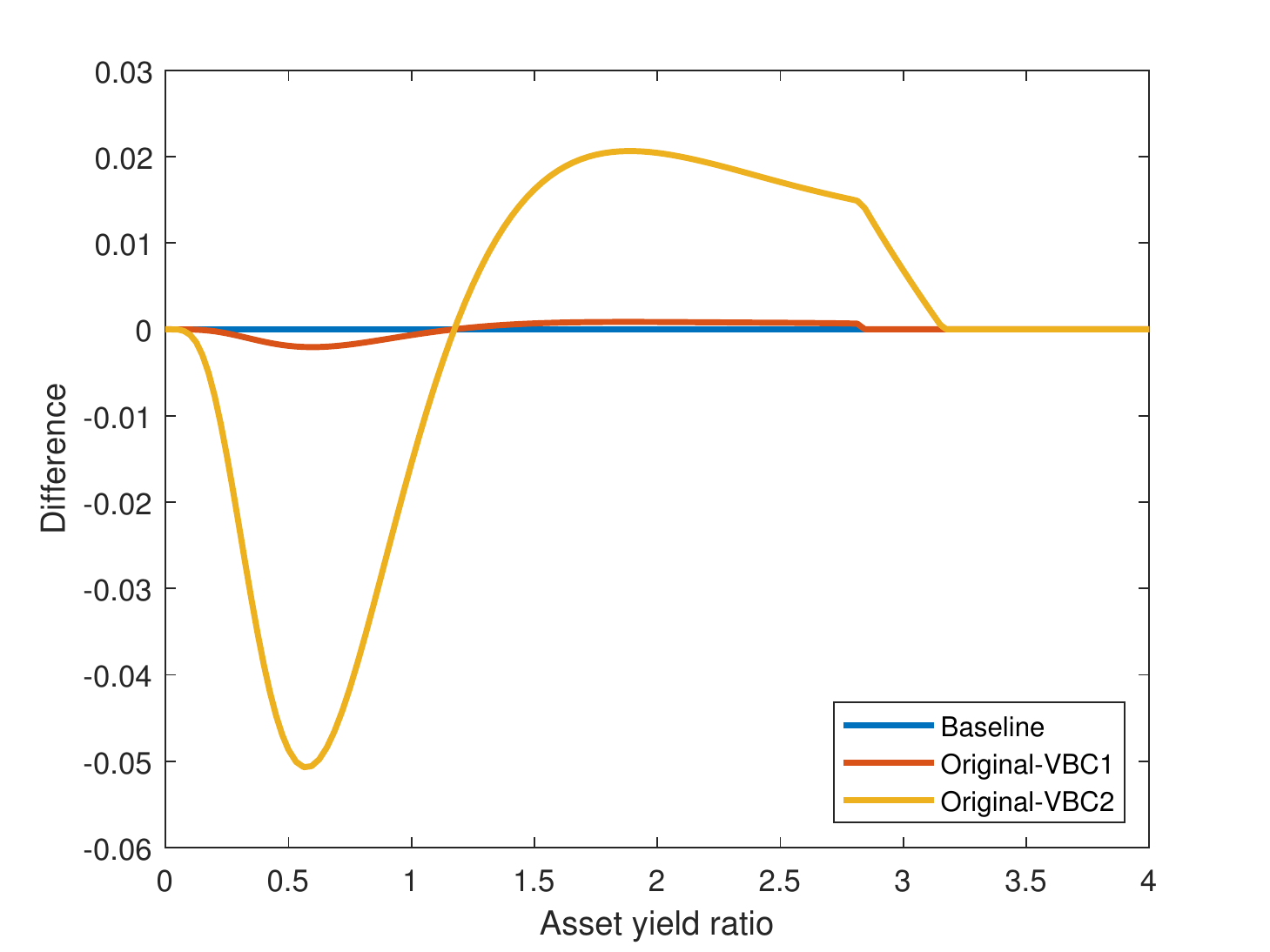}}
\caption{Comparison of early exercise boundaries $A(t,v)$, discounted option prices, and discounted deltas for the American exchange option among various choices of boundary conditions at $v=v_M$.}
\label{fig-Venttsel}
\end{figure}

The final numerical experiment explores how MOL approximations of option prices and the early exercise boundary are affected by the use of Venttsel boundary conditions. Figure \ref{fig-Venttsel} shows the early exercise boundary, American option price, and American delta computed for each Venttsel boundary condition \eqref{eqn-Fichera-VBC1} and \eqref{eqn-Fichera-VBC1} using parameters values in Table \ref{tab-ParameterValues}. For the option price and delta, we exhibit the differences at $t=0$ and $v=v_M$ when compared to the base case that uses \eqref{eqn-MOLEu-vBC}. As can be seen from the graphs, there is very little difference in the MOL approximations when either \eqref{eqn-MOLEu-vBC} or \eqref{eqn-Fichera-VBC1} is used. This is rather surprising as \eqref{eqn-Fichera-VBC1} is notably simpler than when \eqref{eqn-MOLEu-vBC} is applied to \eqref{eqn-MOLEu-IPDE}, which still contains second-order derivatives in $v$. We furthermore note that the computation time when \eqref{eqn-Fichera-VBC1} is active is shorter than when \eqref{eqn-MOLEu-vBC} is used, indicating that using the Venttsel boundary condition may be computationally more efficient while delivering the similar results. The differences in option price and delta are much more pronounced when \eqref{eqn-Fichera-VBC2} is used, particularly when the option is at-the-money. Using boundary condition \eqref{eqn-Fichera-VBC2} also results to a much larger approximation of the early exercise boundary at $t=0$ compared to when either \eqref{eqn-MOLEu-vBC} or \eqref{eqn-Fichera-VBC1} is used.


We end this section with a disclaimer that these observations and conclusions were made for the specific parameter values we assumed in Table \ref{tab-ParameterValues} and the modifications thereof introduced in each numerical experiment. We anticipate that the insights explored here may still hold for other parameter values or when a full calibration exercise with respect to actual data is implemented, but these conclusions might no longer be true in general.

\section{Concluding Remarks}
\label{sec-Conclusion}

This paper discusses the application of \citeauthor*{Bjerskund-1993}'s put-call transformation technique in pricing European and American exchange options under stochastic volatility and jump-diffusion dynamics. This technique allows us to reduce the number of dimensions in the main problem and write the option price and the associated IPDE as a function of time, the asset yield ratio, and the instantaneous variance level. With the inhomogeneous form of the IPDE for the American exchange option, we were also able to analyze the behavior of the early exercise boundary near maturity. It was found that the limit of the boundary at maturity is strongly influenced by the magnitude of the difference between asset dividend yields and the jump component of each asset price process. The numerical implementation in this paper complements the integral representations of the exchange option prices obtained by \citet{GarcesCheang-2020} using Fourier and Laplace transforms.

Given the reduction in dimensions, we then formulate a method of lines algorithm to numerically solve the option pricing IPDE, thereby detailing and extending the method presented by \citet{Chiarella-2009}. As noted by researchers who have applied the MOL in option pricing, this method is particularly useful as it naturally computes for the option delta and gamma and minimal adjustments are required to calculate the free boundary associated with American options. While for simplicity our implementation uses constant parameter values, Algorithms \ref{pseudo-MOL-EuExcOp} and \ref{pseudo-MOL-AmExcOp} can easily be extended to parameters that are deterministic functions of time. However, adding more stochastic elements will result to a higher number of dimensions in the MOL implementation, which will then require additional iteration levels \citep[see][for example]{Kang-2014}. As reported in previous work \citep[among others]{Chiarella-2009, ChiarellaZiveyi-2013, Kang-2014}, the MOL performs just as efficiently, if not more, than Monte Carlo simulations, numerical integration, and other numerical methods for solving IPDEs and PDEs, such as finite difference methods, componentwise splitting methods, and sparse grid approaches. Our numerical analysis in Section \ref{sec-MOL} confirms that the MOL indeed performs more efficiently to the \citet{LongstaffSchwartz-2001} Monte Carlo approach.Algorithms \ref{pseudo-MOL-EuExcOp} and \ref{pseudo-MOL-AmExcOp} and the accompanying discussion in Section \ref{sec-MOL} can be easily modified and implemented to accommodate other payoff structures or underlying asset price dynamics.

Using the MOL approach, we were also able to assess the impact of the model parameters on the exchange option prices and the early exercise boundary. We find that the jump intensities have substantial impact on option prices as it is able to shift the early exercise boundary curves upward and generate the largest price differences relative to the default prices computed using parameters in Table \ref{tab-ParameterValues}. We also find the correlation parameters and the volatility constants of proportionality have a considerable impact, while the parameters of the variance process have the least impact. The presence of stochastic volatility and or jumps in one or both assets also has a considerable effect on option prices when the asset prices are positively or negatively correlated with one another. We note however that these conclusions are true for the specific set of parameter values used in the numerical approximation and may not hold in full generality. A complete analytical comparative static analysis remains to be seen in literature, but nonetheless we have shown that jumps and stochastic volatility have considerable impact on option prices and the early exercise boundary.

Our analysis also showed that \ref{eqn-Fichera-VBC1} at the far variance boundary is a plausible alternative to the often-used vanishing vega assumption, producing comparable results with less computational time. We note however that the choice of boundary conditions is strongly affected by the calibration of model parameters. Specifically, if calibration yields parameters which violate either the Feller condition or the condition for the viability of Venttsel boundary condition \ref{eqn-Fichera-VBC1}, then alternative conditions must be imposed on the boundary of the computational domain. 
 
A formal convergence analysis of the MOL is also an issue to be investigated in future work, although the algorithm we presented converges for all reported parameter values. An application of the put-call transformation transformation and/or the MOL in pricing multi-asset derivatives under other asset price model specifications (e.g. L\'evy processes, regime switching models) is also a topic that we aim to explore in future studies. In a future study, we also aim to see how the MOL can be extended when pricing takes place in the risk-neutral world (i.e. the second asset price then becomes a separate spatial variable) or when additional risk factors are included, such as multi-factor stochastic volatility \citep{Christoffersen-2009} and stochastic interest rates.  



\section*{Disclosure Statement}

The authors report no potential conflict of interest arising from the results of this paper.

\bibliographystyle{tfcad}
\bibliography{arXivMOLExchangeOption}

\begin{appendix}

\section{Proof of Proposition \ref{prop-PutCall-EEBLimit}}
\label{app-Proof-EEBLimit}

The method of \citet{ChiarellaZiogas-2009}, adapted to our situation, is as follows.\footnote{\citet{ChiarellaZiogas-2009} proposed this method as an alternative to the local analysis of the option PDE for small time-to-maturity options as was done by \citet{Wilmott-1993} in the pure diffusion case.} First, we set the inhomogeneous term $\Xi(t,\tilde{S},v)$ (given by equation \eqref{eqn-PutCall-InhomogeneousTerm}) to zero and evaluate the result at $t=T$ and $\tilde{S}=B(T^-,v)e^{(q_1-q_2)T}$. The resulting expression is then rearranged to yield equation \eqref{eqn-PutCall-EEBLimit}.

Performing the first step yields the equation
\begin{align}
\begin{split}
\label{eqn-PutCall-EEBLimit-Step1}
0 & = e^{-q_2 T}\left(q_1 B(T^-,v)-q_2\right)\\
	& \qquad - \tilde{\lambda}_1\int_{-\infty}^{-\ln\left[\frac{B(T,v)e^{(q_1-q_2)T}}{B(T^-,v)e^{(q_1-q_2)T}}\right]}\left[\tilde{V}^A\left(T, B(T^-,v)e^{(q_1-q_2)T}e^y, v(T)\right)\right.\\
	& \hspace{100pt} \left.-e^{-q_2 T}\left(B(T^-,v)e^{y}-1\right)\right]G_1(y)\dif y\\
	& \qquad - \tilde{\lambda}_2\int_{\ln\left[\frac{B(T,v)e^{(q_1-q_2)T}}{B(T^-,v)e^{(q_1-q_2)T}}\right]}^{\infty}\left[\tilde{V}^A\left(T, B(T^-,v)e^{(q_1-q_2)T}e^{-y}, v(T)\right)\right.\\
	& \hspace{100pt} \left.-e^{-q_2 T}\left(B(T^-,v)e^{-y}-1\right)\right]G_2(y)\dif y.
\end{split}
\end{align}
At maturity $t=T$, the option will be exercised if $\tilde{S}(T) \geq e^{(q_1-q_2)T}$, and so $B(T,v)=1$. Thus in the above calculation, setting $\tilde{S}=B(T^-,v)e^{(q_1-q_2)T}$ induces the stopping criterion in equation \eqref{eqn-PutCall-StoppingContinuationRegions2} for $\calS(T)$ since $B(T^-,v)\geq 1$. This implies that $\vm{1}(\calA(T))=1$ in the inhomogeneous term \eqref{eqn-PutCall-InhomogeneousTerm}. Furthermore, terminal condition \eqref{eqn-PutCall-BoundaryConditions-Vtilde} allows us to simplify
equation \eqref{eqn-PutCall-EEBLimit-Step1} as
\begin{align*}
0 & = q_1 B(T^-,v)-q_2 +\tilde{\lambda}_1\int_{-\infty}^{-\ln B(T^-,v)}\left[B(T^-,v)e^{y}-1\right]G_1(y)\dif y\\
	& \qquad +\tilde{\lambda}_2\int_{\ln B(T^-,v)}^{\infty}\left[B(T^-,v)e^{-y}-1\right]G_2(y)\dif y.
\end{align*}
Rearranging the terms yields the equation $$B(T^-,v) = \frac{q_2+\tilde{\lambda}_1\int_{-\infty}^{-\ln B(T^-,v)} G_1(y)\dif y + \tilde{\lambda}_2\int_{\ln B(T^-,v)}^\infty G_2(y)\dif y}{q_1 +\tilde{\lambda}_1\int_{-\infty}^{-\ln B(T^-,v)} e^y G_1(y)\dif y + \tilde{\lambda}_2\int_{\ln B(T^-,v)}^\infty e^{-y} G_2(y)\dif y}.$$

We note lastly from \citet{BroadieDetemple-1997} that $B(t,v)\geq 1$ for any $t\in[0,T]$ and $v\in(0,\infty)$. Therefore, we must enforce a lower bound of 1 on $B(T^-,v)$ via the maximum function. The result stated in the proposition thus holds.

\section{Proof of Proposition \ref{prop-PutCall-EEBLimit-Existence}}
\label{app-Proof-EEBLimit-Existence}

Our proof adapts the arguments made by \citet[pp. 34-35]{ChiarellaKangMeyer-2015}. For $x\in(0,\infty)$, define the function
\begin{align*}
f(x)	& = q_2+\tilde{\lambda}_1\int_{-\infty}^{-\ln x} G_1(y)\dif y + \tilde{\lambda}_2\int_{\ln x}^\infty G_2(y)\dif y\\
			& \qquad - x\left(q_1 +\tilde{\lambda}_1\int_{-\infty}^{-\ln x} e^y G_1(y)\dif y + \tilde{\lambda}_2\int_{\ln x}^\infty e^{-y} G_2(y)\dif y\right).
\end{align*}
Denote by $x^*$ a zero of $f$ (i.e. $f(x^*)=0$) on $(0,\infty)$, if any exist.

By calculating $f'(x)$, we find that $f$ is strictly decreasing on $(0,\infty)$ if $q_1>0$. We also observe that $$\lim_{x\to 0^+}f(x) = q_2+\tilde{\lambda_1}\int_{-\infty}^\infty G_1(y)\dif y+\tilde{\lambda}_2\int_{-\infty}^\infty G_2(y)\dif y = q_2+\tilde{\lambda}_1+\tilde{\lambda}_2>0.$$ Furthermore, we note that for a fixed $x>0$, $0< xe^y G_1(y) \leq G_1(y)$ for all $y\leq -\ln x$. Thus, $$0<x\int_{-\infty}^{-\ln x}e^y G_1(y)\dif y \leq \int_{-\infty}^{-\ln x}G_1(y)\dif y \to 0 \quad \text{as $x\to\infty$}.$$ A similar argument can be used to show that $$0<x \int_{\ln x}^\infty e^{-y}G_2(y)\dif y \leq \int_{\ln x}^\infty G_2(y)\dif y \to 0 \quad \text{as $x\to\infty$}.$$ As such, we find that
\begin{align*}
\lim_{x\to\infty} f(x)
	& = \lim_{x\to\infty}\Bigg(q_2 - q_1 x + \tilde{\lambda}_1\left[\int_{-\infty}^{-\ln x}G_1(y) \dif y - x\int_{-\infty}^{-\ln x}e^y G_1(y)\dif y\right]\\
	& \qquad + \tilde{\lambda}_2 \left[\int_{\ln x}^\infty G_2(y)\dif y - x\int_{\ln x}^\infty\ e^{-y}G_2(y)\dif y\right]\Bigg)\\
	& = \lim_{x\to\infty}(q_2-q_1 x).
\end{align*}
Thus, if $q_1>0$, then $\lim_{x\to\infty}f(x)<0$ and so $f$ strictly decreases from positive to negative values as $x$ increases on $(0,\infty)$. Therefore, there exists a unique $x^*\in(0,\infty)$ such that $f(x^*)=0$.

Now suppose $q_1>0$. Evaluating $f$ at $x=1$ gives us $$f(1) = q_2-q_1+\tilde{\lambda}_1\int_{-\infty}^0 (1-e^y)G_1(y)\dif y+\tilde{\lambda}_2\int_0^\infty (1-e^{-y})G_2(y)\dif y.$$ If $f(1)\leq 0$, then $x^*$ must be in the interval $(0,1]$ since $f$ is strictly decreasing. Otherwise, $x^*>1$ if and only if $f(1)>0$, which is the condition stated in the proposition. Lastly, by Proposition \ref{prop-PutCall-EEBLimit}, we have $B(T^-,v) = \max\{1,x^*\}$.

\section{Coefficients of the MOL Approximation \eqref{eqn-MOLEu-IPDE-Approx2}}
\label{app-MOL-Coeff}

The coefficients of \eqref{eqn-MOLEu-IPDE-Approx2} are given by
\begin{align}
\begin{split}
\label{eqn-MOLEu-IPDE-Approx2-Coeff}
a(s,v_m) & = \max\left\{\frac{1}{2}\sigma^2 v_m s^2, 10^{-4}\right\} \qquad \text{(regularized coefficient)}\\
b(s,v_m) & = -\left(\tilde{\lambda}_1\tilde{\kappa}_1+\tilde{\lambda}_2\tilde{\kappa}_2^-\right)s\\
c(\tau_n,s,v_m) & = \frac{\omega^2 v_m}{(\Delta v)^2} + (\tilde{\lambda}_1+\tilde{\lambda}_2) + \frac{\max\left\{\xi\eta-(\xi+\Lambda)v_m,0\right\}}{\Delta v}\\
				 & \qquad - \frac{\min\left\{\xi\eta-(\xi+\Lambda)v_m,0\right\}}{\Delta v} + \begin{cases}
						1/\Delta\tau & \text{if $n=1,2$}\\
						3/(2\Delta\tau) & \text{if $n\geq 3$}
						\end{cases}\\
F(\tau_n,s,v_m) & = -\frac{\omega^2 v_m}{2(\Delta v)^2}\left[V_{n,m+1}(s) + V_{n,m-1}(s)\right]\\
				 & \qquad -\frac{\omega(\sigma_1\rho_1-\sigma_2\rho_2)}{2\Delta v} v_m s \left[\scrV_{n,m+1}(s)-\scrV_{n,m-1}(s)\right]\\
				 & \qquad -\frac{\max\left\{\xi\eta-(\xi+\Lambda)v_m,0\right\}}{\Delta v} V_{n,m+1}(s)\\
				 & \qquad +\frac{\min\left\{\xi\eta-(\xi+\Lambda)v_m,0\right\}}{\Delta v} V_{n,m-1}(s)\\
				 & \qquad - \begin{cases}
						V_{n-1,m}(s)/\Delta\tau & \text{if $n=1,2$}\\
						\left[4V_{n-1,m}(s)-V_{n-2,m}(s)\right]/(2\Delta\tau) & \text{if $n\geq 3$},
						\end{cases}
\end{split}
\end{align}
for $m=1,\dots,M-1$.

At the last variance line $m=M$, boundary condition \eqref{eqn-MOLEu-vBC} implies that some of the coefficients above have to be redefined as
\begin{align}
\begin{split}
\label{eqn-MOLEu-IPDE-Approx2-Coeff-vBC}
b(s,v_M) & = \frac{\omega(\sigma_1\rho_1-\sigma_2\rho_2)}{2\Delta v}v_M s-\left(\tilde{\lambda}_1\tilde{\kappa}_1+\tilde{\lambda}_2\tilde{\kappa}_2^-\right)s\\
c(\tau_n,s,v_M) & = \frac{\omega^2 v_M}{2(\Delta v)^2} + (\tilde{\lambda}_1+\tilde{\lambda}_2)\\
				 & \qquad - \frac{\min\left\{\xi\eta-(\xi+\Lambda)v_M,0\right\}}{\Delta v} + \begin{cases}
						1/\Delta\tau & \text{if $n=1,2$}\\
						3/(2\Delta\tau) & \text{if $n\geq 3$}
						\end{cases}\\
F(\tau_n,s,v_M) & = -\frac{\omega^2 v_m}{2(\Delta v)^2} V_{n,M-1}(s) - \frac{\omega(\sigma_1\rho_1-\sigma_2\rho_2)}{2\Delta v} v_M s \scrV_{n,M-1}(s)\\
				 & \qquad +\frac{\min\left\{\xi\eta-(\xi+\Lambda)v_M,0\right\}}{\Delta v} V_{n,M-1}(s)\\
				 & \qquad - \begin{cases}
						V_{n-1,M}(s)/\Delta\tau & \text{if $n=1,2$}\\
						\left[4V_{n-1,M}(s)-V_{n-2,M}(s)\right]/(2\Delta\tau) & \text{if $n\geq 3$},
						\end{cases}
\end{split}
\end{align}

\end{appendix}

\end{document}